\documentclass[acmtog]{acmart}

\usepackage{enumitem}

\usepackage[utf8]{inputenc}
\usepackage[english]{babel}

\usepackage{graphicx}
\usepackage{tikz}
\newtheorem{remark}{Remark}

\newtheorem{theorem}{Theorem}

\usepackage{algpseudocode}

\algnewcommand\algorithmicinput{\textbf{Input:}}
\algnewcommand\Input{\item[\algorithmicinput]}
\algnewcommand\algorithmicoutput{\textbf{Output:}}
\algnewcommand\Output{\item[\algorithmicoutput]}
\usepackage{caption}
\makeatletter
\DeclareRobustCommand*\cal{\@fontswitch\relax\mathcal}
\makeatother
\DeclareMathOperator*{\argmin}{arg\,min}
\DeclareMathOperator{\vol}{vol}

\DeclareMathOperator{\tr}{tr}
\DeclareMathOperator{\p}{\partial}
\DeclareMathOperator{\rank}{rank}
\DeclareMathOperator{\R}{\mathbb R}
\newcommand{\eps}{\varepsilon}

\acmJournal{TOG}

\begin{document}
\title{Practical lowest distortion mapping}

\author{Vladimir Garanzha}
\author{Igor Kaporin}
\author{Liudmila Kudryavtseva}

\affiliation{\institution{Dorodnicyn Computing Center FRC CSC RAS}, Moscow, Russia, \institution{Moscow Institute of Physics and Technology}, Moscow, Russia}
\author{Francois Protais}
\author{David Desobry}
\author{Dmitry Sokolov}
\authornote{Corresponding author: dmitry.sokolov@loria.fr}
\affiliation{\institution{Université de Lorraine}, \institution{CNRS}, \institution{Inria}, \institution{LORIA}, F-54000 Nancy, France}

\begin{abstract}

Construction of optimal deformations is one of the long standing problems of computational mathematics.
We consider the problem of computing quasi-isometric deformations with minimal possible quasi-isometry constant (global estimate for relative length change).

We build our technique upon~\cite{garanzha2021foldoverfree},
a recently proposed numerical optimization scheme that provably untangles 2D and 3D meshes with inverted elements by partially solving a finite number of minimization problems.
In this paper we show the similarity between continuation problems for mesh untangling and for attaining prescribed deformation quality threshold.
Both problems can be solved by a finite number of partial solutions of optimization problems which are based on finite element approximations of parameter-dependent hyperelastic functionals.
Our method is based on a polyconvex functional which admits a well-posed variational problem.

To sum up, we reliably build 2D and 3D mesh deformations with smallest known distortion estimates (quasi-isometry constants) as well as stable quasi conformal parameterizations for very stiff problems.

\end{abstract}

%
%


%
%


\begin{teaserfigure}
	\centerline{
		\begin{minipage}[!t]{.12\linewidth}\vspace{0pt}
			\vspace{4mm}
			\centering
			\includegraphics[width=.7\linewidth]{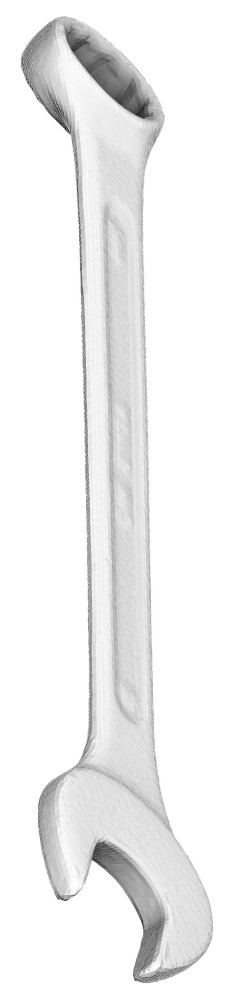}\\
			\vspace{4.25mm}
			\textbf{(a)}
		\end{minipage}
		\begin{minipage}[!t]{.29\linewidth}\vspace{0pt}
			\centering
			\includegraphics[width=\linewidth]{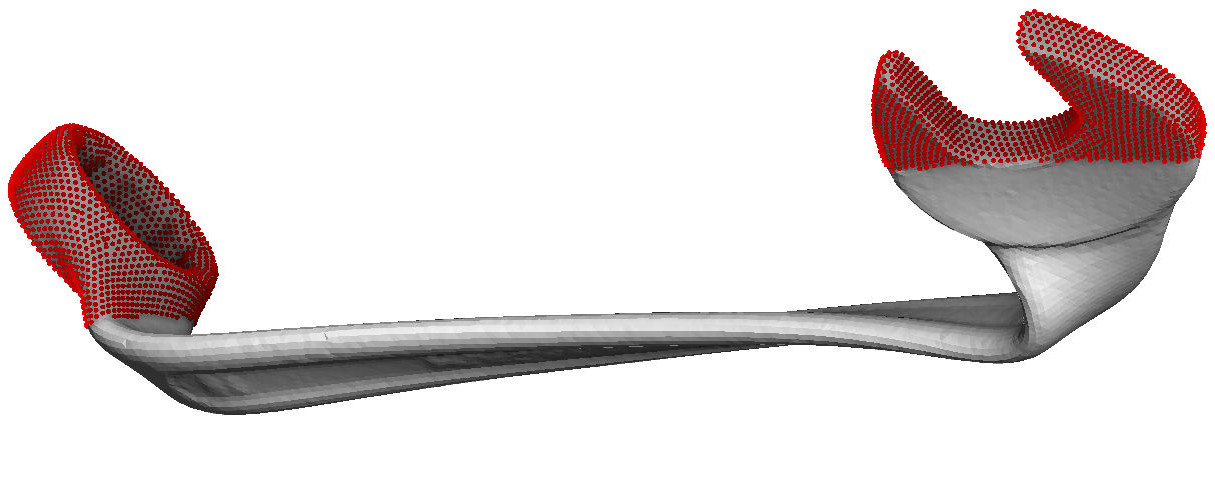}
			\includegraphics[width=\linewidth]{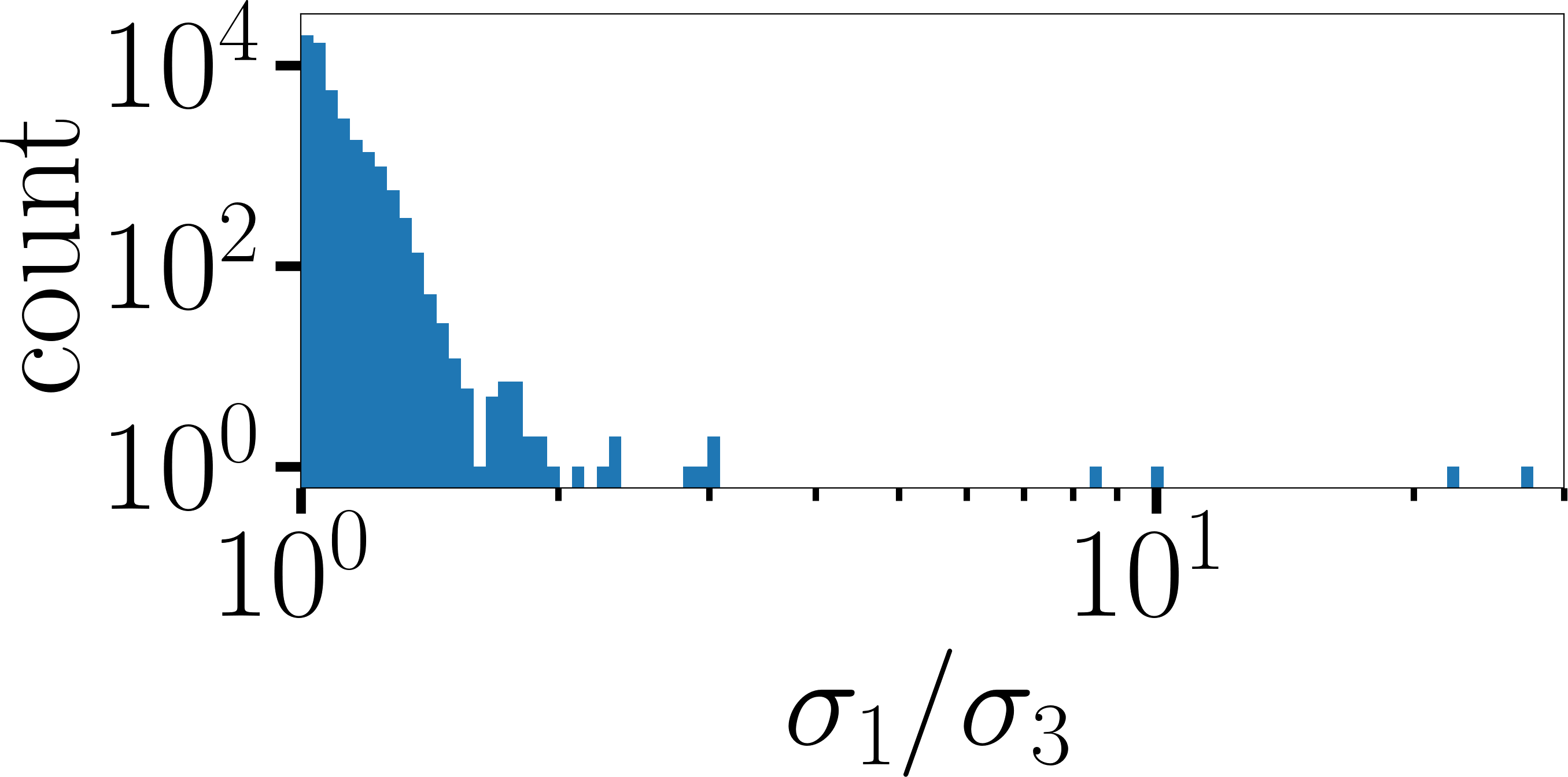}
			\includegraphics[width=\linewidth]{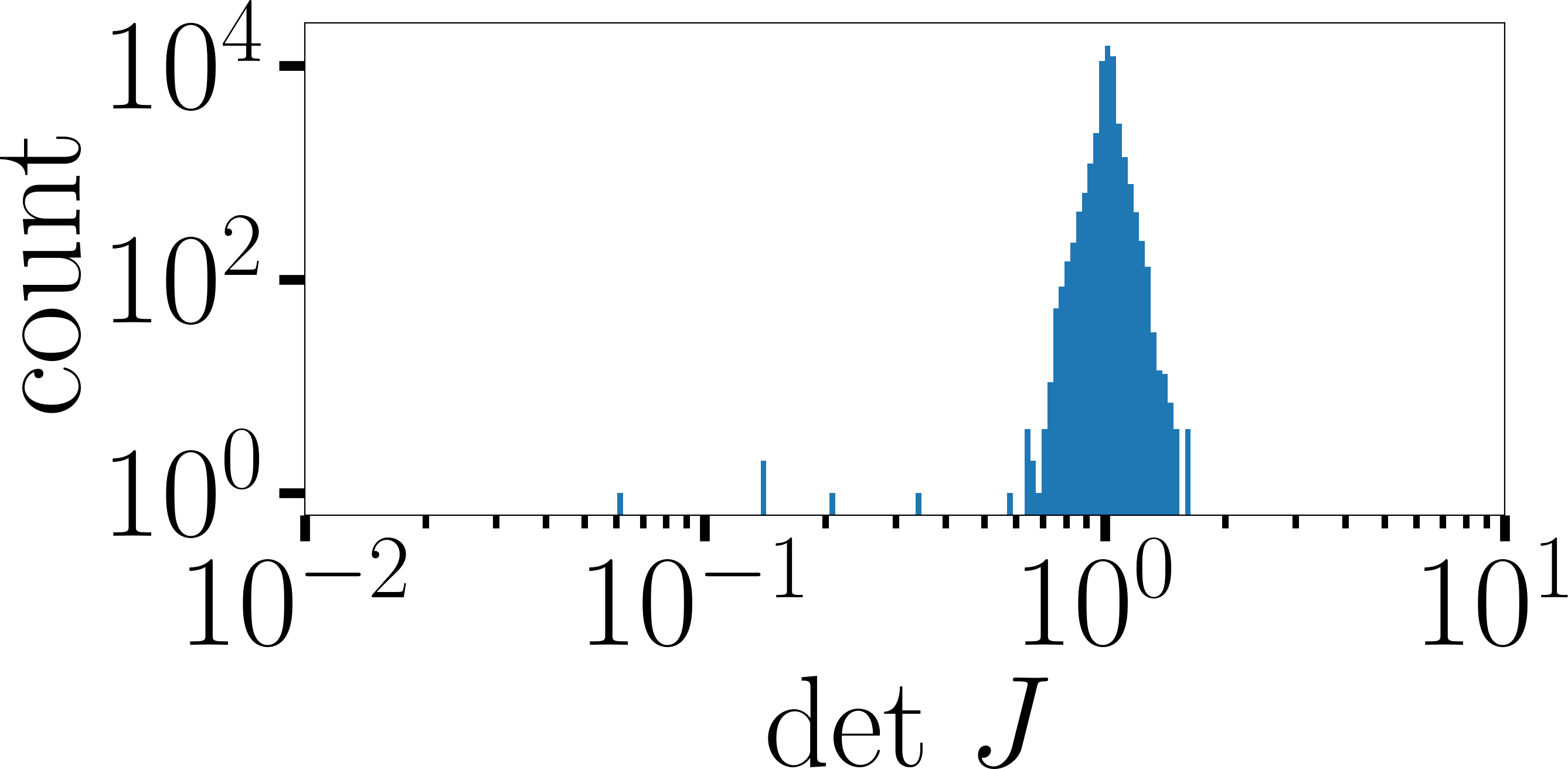}
			\textbf{(b)}
		\end{minipage}
		\begin{minipage}[!t]{.29\linewidth}\vspace{0pt}
			\centering
			\includegraphics[width=\linewidth]{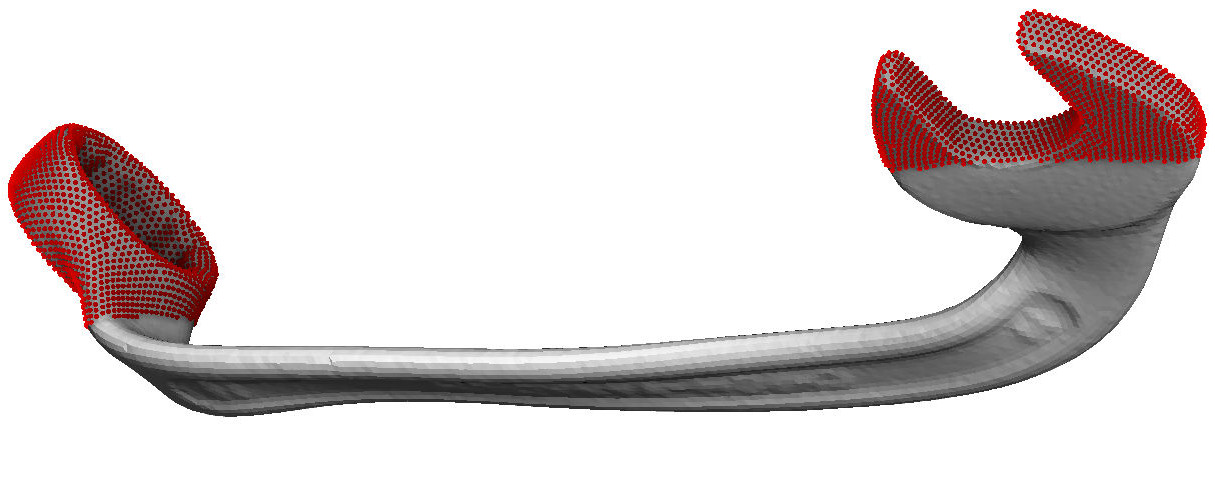}
			\includegraphics[width=\linewidth]{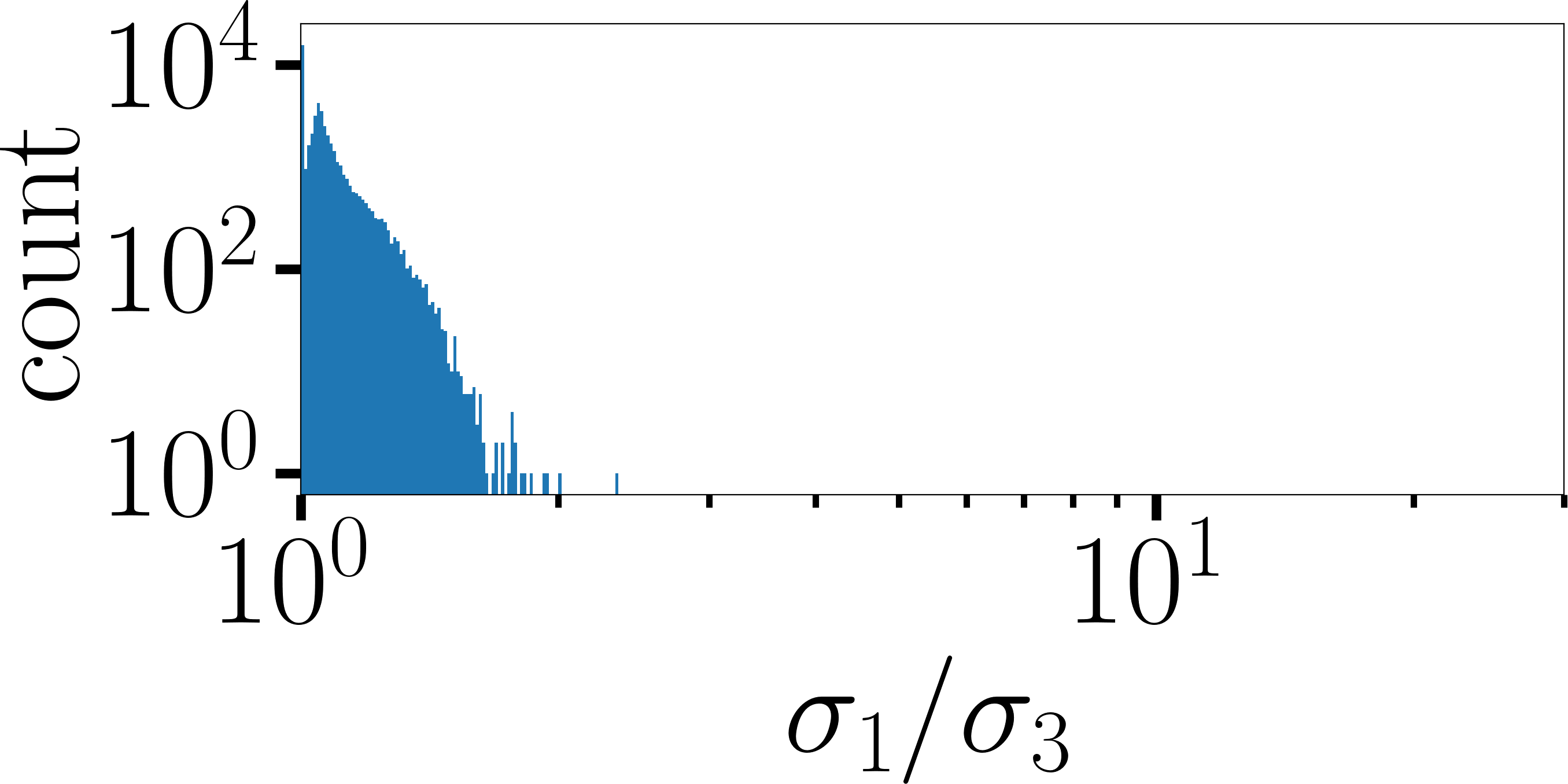}
			\includegraphics[width=\linewidth]{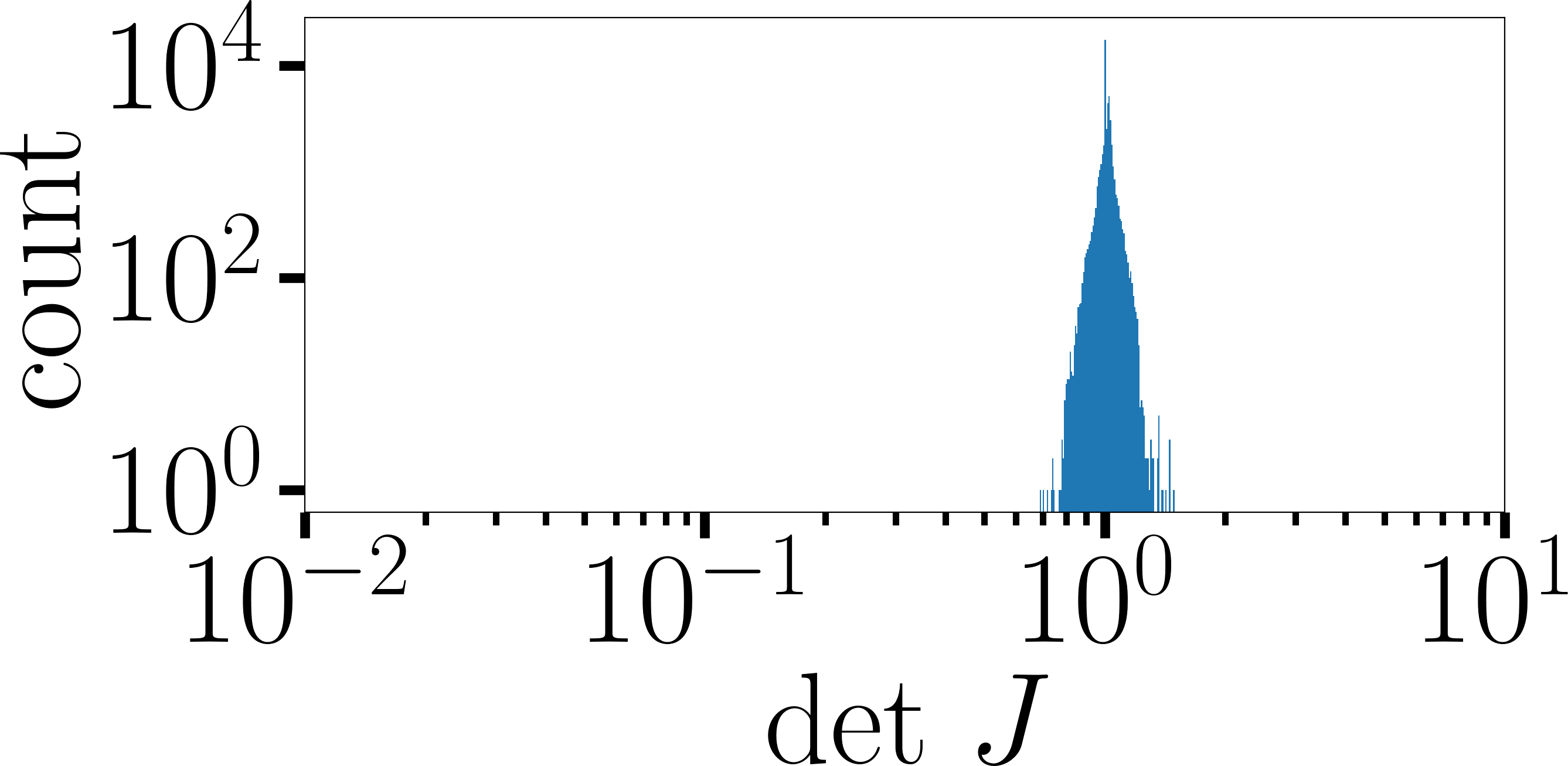}
			\textbf{(c)}
		\end{minipage}
		\begin{minipage}[!t]{.29\linewidth}\vspace{0pt}
			\centering
			\includegraphics[width=\linewidth]{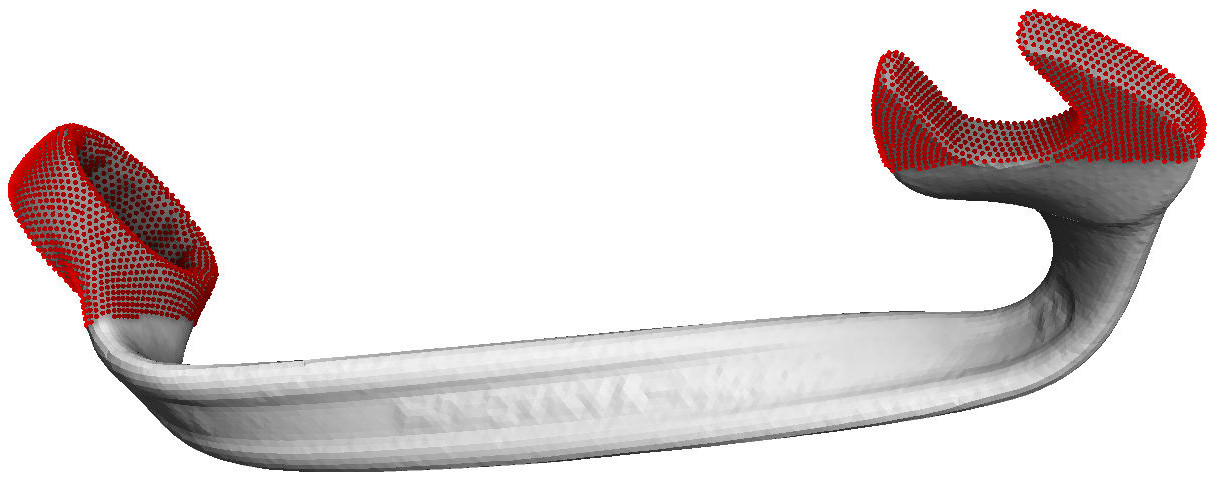}
			\includegraphics[width=\linewidth]{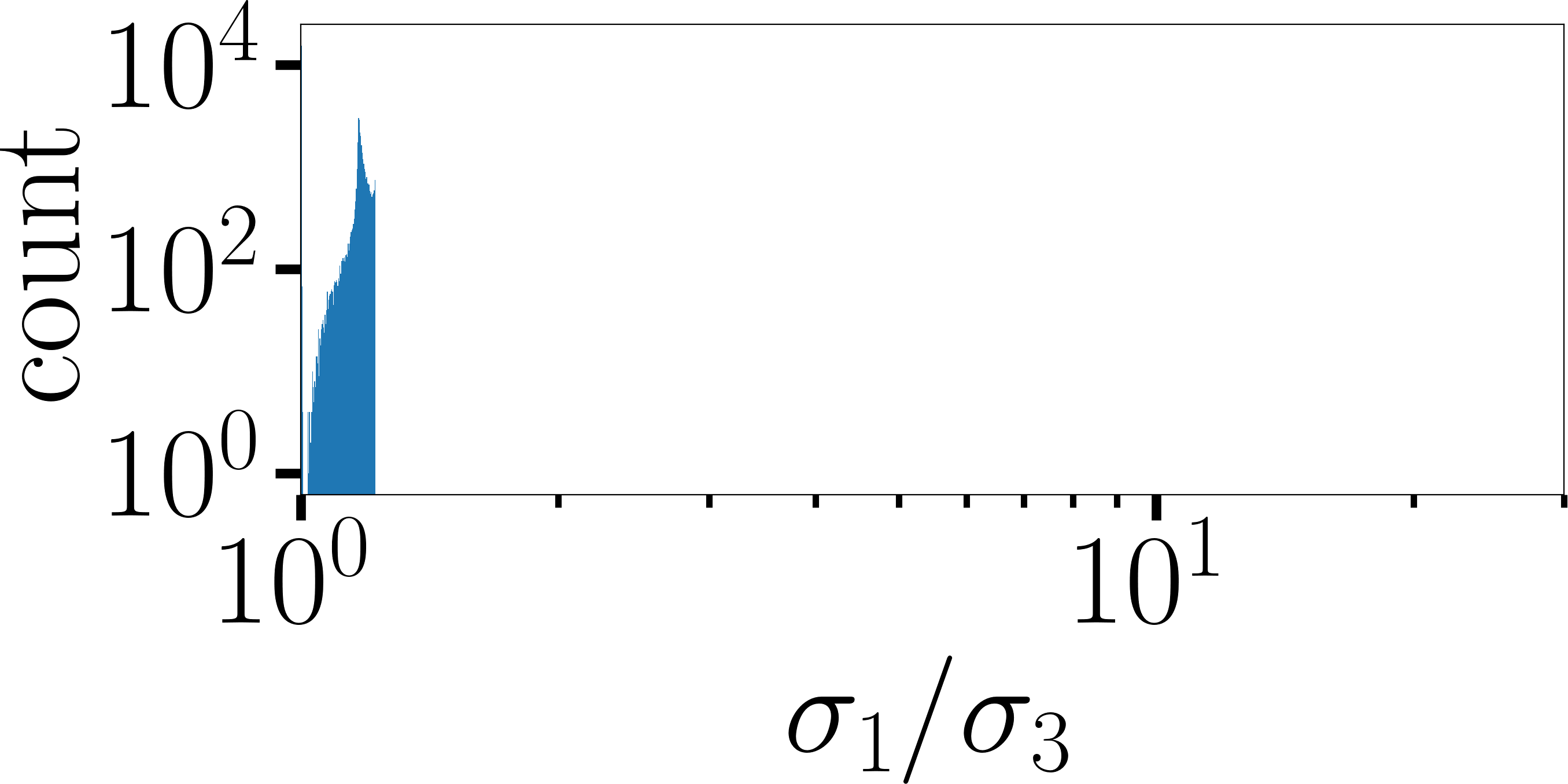}
			\includegraphics[width=\linewidth]{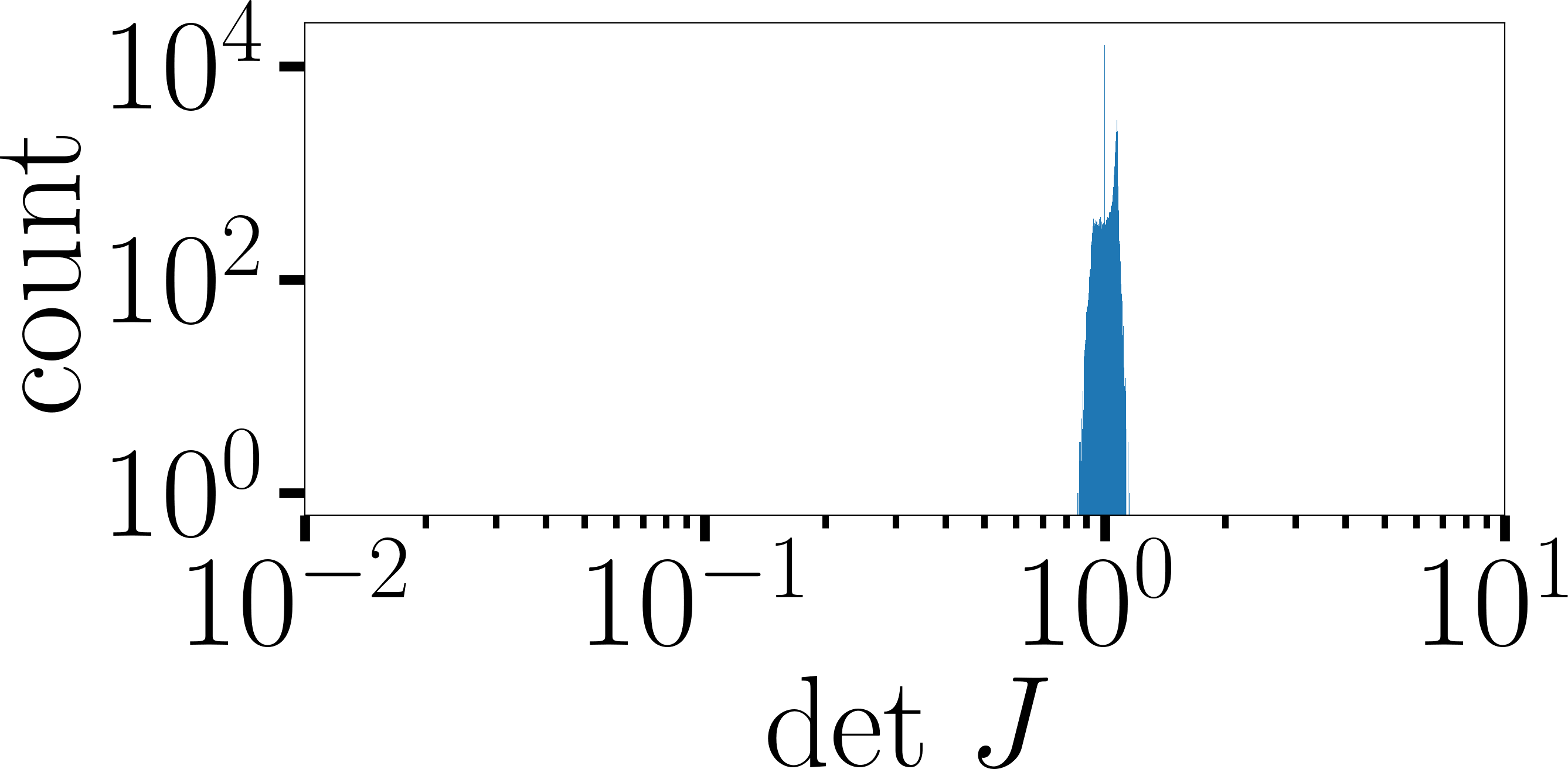}
			\textbf{(d)}
		\end{minipage}
	}
	\caption{Tetrahedral mesh deformation, locked vertices are shown in red. \textbf{(a):} Rest shape, \textbf{(b):} ABCD~\cite{Naitsat2019}, \textbf{(c):} untangling \cite{garanzha2021foldoverfree}, \textbf{(d)} our lowest distortion mapping with $\theta=\frac12$. \textbf{Top row:} deformations, \textbf{middle and bottom rows:} log-log histograms of element quality (condition number of the Jacobian matrix and Jacobian determinant).}
	\label{fig:teaser}
\end{teaserfigure}

\maketitle
\section{Introduction}

Construction of optimal deformations is one of central themes in mesh generation research.
Generally, for computing mesh deformations, using elasticity analogy was found to be very fruitful and resulted in efficient engineering mesh generation algorithms with sound theoretical foundations~\cite{jacquotte1988mechanical,rumpf1996variational}.
The idea is to say that a mesh represents an elastic material, whose stored energy of deformation can be measured as $\int\limits_\Omega f(J)$,
where $\Omega$ is the input domain, $J$ is the Jacobian matrix of elastic deformation, and $f$ is a measure of distortion.
Then, obviously, we want to minimize the stored energy of deformation.

More precisely, let us say that we want to compute a map $\vec x(\vec \xi) : \Omega\subset \mathbb R^d \rightarrow \mathbb R^d$,
where $d$ stands for the dimension (2 or 3), and the arrow denotes a $d$-dimensional vector.
Consider the following variational problem:
\begin{equation}
	\argmin\limits_{\vec{x}(\vec{\xi})}\int\limits_\Omega f(J)  \,d \xi,
	\label{eq:winslow}
\end{equation}
where $J$ is again the Jacobian matrix of the mapping $\vec x(\vec \xi)$.


While this formulation allows to minimize distortion on average, it does not allow to limit maximum distortion.
The problem of constructing bounded distorted deformations has a long-standing history in elasticity research and goes back to 1957 Prager's work on ``Ideal locking materials'' \cite{prager1957stiffening}.
Now this problem is referred to as ``stiffening'' of elastic material and generally is formulated as a set of nonlinear constraints on the Jacobian matrix of elastic deformation \cite{ciarlet1985stiffening}.

Constructing elastic deformations with bounded distortion constraints is a very hard non-convex and non-linear problem.
There were numorous attempts made to solve the problem.
For example, \cite{Sorkine2002} propose to lay triangles in a plane in a greedy manner without exceeding a user-specified distortion bound.
Obviously, the mesh is cut during the procedure, and since it is possible to lay individual triangles without any distortion, the method succeeds.
\cite{Fu2016} propose to enforce the distortion constraints with a penalty method, leading to conflicts between multiple terms in the energy to minimize.
\cite{LargeScaleBD:2015} alternate between energy optimization and a non-trivial projection to the highly non-convex set of constraints.
\cite{Lipman2012} formulates the problem as a second-order cone programming, relying on elaborated commercial solvers such as MOSEK~\cite{Andersen2000}.

All these papers try to incorporate the boundedness constraint into different black-box optimization toolboxes.
While the approach may work reasonably well in practice, it is hard to obtain any guarantees,
and solutions may exhibit undesirable, hard to explain and eliminate artifacts, say, noise and loss of symmetries.
We propose to explore another research direction based on a {\bf unconditional} optimization, avoiding altogether all issues related to constraints.

A very interesting idea~\cite{garanzha2000barrier} is to consider a quasi-isometric hyperelastic material,
which unlike other known models, provides admissible deformations with bounded global distortion (bounded quasi-isometry constant) as minimizers of elastic energy.
Invertibility theorem for deformation of this material was established in 2D and 3D cases  \cite{Garanzha2014}.
The main idea is to use controlled stiffening of material which is incorporated directly into definition of the density of deformation energy in such a way that when local measure of deformation exceeds certain threshold, the elastic material becomes infinitely stiff.
The stiffening threshold is introduced as a parameter, and max-norm optimization problem for deformation is formulated as a continuation problem for polyconvex functional
(minimization of stiffening threshold, or, alternatively, maximization of the quality threshold).
Unfortunately, this work lacks a robust strategy of stiffening parameter choice.

Generally speaking, the stiffening technique may be expressed in the following way.
Having a distortion measure $f(J)$, we can try to minimize the following energy:
\[ \int\limits_\Omega w \cdot f(J) \, d \xi, \]
where $w$ is a weight function.
Typically, $w$ is chosen to be large in the regions where small values of $f(J)$ are required.
We can use this general weighted formulation to control pointwise behaviour of the spatial distribution of the distortion measure.
This idea was suggested in \cite{garanzha2000barrier} to build orthogonal mappings near boundaries which is one of the key requirements for CFD meshes.
\cite{MIQ2009} used such weights in a heuristic procedure for mesh untangling.
If an adaptation metric is prescribed in the computational domain, one can compute the weight function $w$ according to this metric~\cite{Ivanenko2000}.

Some other methods have also used this technique.
For example, by setting
\[
w := f^p(J), \ \ p > 0,
\]
one can get a power law enhancement, thus penalizing large values of local distortion~\cite{10.1007/978-3-030-10934-9_35}.
The same idea was used in IDP algorithm \cite{Fang2021IDP} where the authors suggested to use $p=4$ as a rule of thumb.
If function $f$ is polyconvex, function $f^{1 + p}(J)$ is polyconvex as well.
Note however, that in the continuous case one cannot prove that $f^{1 + p}(J)$ is bounded from above, and in the discrete case one cannot prove that $f(J)$ does not grow to infinity under mesh refinement.

Another weight
\[
w := \frac1{1 - t f(J)}, \quad 0 \leq t < 1
\]
corresponds to algorithm from \cite{garanzha2000barrier}.
The resulting functional has an inifinite barrier (refer to \S\ref{sec:stiffening} for more details) on the boundary of the set of quasi-isometric deformations \cite{Garanzha2014},
thus solving problem of quasi-isometric map generation formulated by Godunov~\cite{Godunov1995} for general domains.

Having carefully designed a strategy of choice for the stiffening parameter $t$, we obtain lowest distortion maps with our quasi-isometric stiffening (QIS) algorithm~(Alg.~\ref{alg:stiffening} + Eq.~\eqref{eq:t.update}).
With this new contribution, we were able to confirm the 20-years old conjecture that variational problem \cite{garanzha2000barrier} allows to build best deformations compared to state-of-the-art algorithms.

\paragraph*{Our contributions}

We propose a very simple algorithm that allows us to reliably build 2D and 3D mesh deformations with \textbf{smallest known distortion estimates} (quasi-isometry constants).
To the best of our knowledge, we are the first to provide \textbf{theoretical guarantees}

to this long standing problem.
Our approach is a discretization of a well-posed variational scheme, and it has an advantage that type, size and quality of mesh elements in the deformed object have a weak influence on the computed deformation.
We show that attainable quality threshold estimates (quasi-isometry constants) \textbf{do not deteriorate under mesh refinement} which is a unique property of the proposed algorithm.

By coupling our technique with~\cite{garanzha2021foldoverfree}, we obtain a complete unified mapping pipeline.
For a better reproductibility, we publish a complete \textbf{C++ implementation} \cite{supplemental}.
We start from an arbitrary initial deformation, untangle the mesh in a finite number of steps,
minimizing mean distortion, and finally we minimize the maximum distortion.
Just like for the untangling step, we can obtain a deformation with a prescribed quality threshold in a finite number of steps.
Both parts of the pipeline build upon the same ideas, and require only a linear solver~\cite{Hestenes1952MethodsOC} for positive definite matrices (if Newton minimization is adopted) or a L-BFGS solver~\cite{LBFGS} for a quasi-Newton scheme.

Last, but not least, we bring more robustness to global parameterizations.
\cite{garanzha2021foldoverfree} produce maps free of inverted elements, but do not prevent $k$-coverings, thus losing local injectivity.
We provide a very simple but effective way to handle this problem: we guarantee \textbf{local injectivity for global parameterizations}.

\section{Variational formulation for untangling and distortion minimization}

This section gives a necessary background on elastic deformations.
First, in \S\ref{sec:energy} we revisit main issues of mesh deformation based on the elasticity theory.
Then, in \S\ref{sec:untangling} we present the core idea behind the untangling procedure described in~\cite{garanzha2021foldoverfree}.
Next, in \S\ref{sec:stiffening}, we describe the idea behind generation of deformations of a given quality, until recently very heuristic.
Finally, having prepared all necessary concepts, we can present our latest guarantees and results (\S\ref{sec:finite-bounded-distortion} and \S\ref{sec:results}).

\subsection{Elastic material choice and main issues}
\label{sec:energy}

For our meshes we chose a material whose stored energy of deformation $\int\limits_\Omega f(J)  \,d \xi$ can be measured with density $f$ defined as follows~\cite{garanzha2000barrier,Hormann2000MIPS}:
\begin{equation}
	f(J) := (1 - \theta) f_s(J) + \theta f_v(J),
	\label{eq:distortion}
\end{equation}
where shape distortion is defined as
\begin{equation}
	f_s(J) := \left\{ \begin{array}{ll} \displaystyle\frac1d\frac{ \tr J^\top J}{(\det J)^\frac2d}, & \det J > 0 \\
		+\infty, & \det J \leq 0 \end{array} \right.
	\label{eq:shape}
\end{equation}
while volumetric distortion is defined 
\begin{equation}
	f_v(J) := \left\{ \begin{array}{ll} \frac12 \left( \det J + \displaystyle\frac{1}{\det J} \right), & \det J > 0 \\
		+\infty, & \det J \leq 0 \end{array} \right.
	\label{eq:volume}
\end{equation}
Note that functions $f_s(J)$ and $f_v(J)$ have concurrent goals, one preserves angles and the other preserves volumes, and thus $\theta$ serves as a trade-off parameter.
Density~\eqref{eq:distortion} is a polyconvex function satisfying ellipticity conditions, it is therefore very well suited for a numerical optimization.

Polyconvexity of the energy, ellipticity of the PDE and invertibility of the minimizer are very strong results;
moreover, the minimizer of Prob.~\eqref{eq:winslow} has a minimum average distortion.
There are, however, two issues with this formulation:
\begin{enumerate}[wide=0pt,itemindent=2em]
	\item[\textbf{(a)}] the variational problem makes sense and allows for minimization only when an \textbf{initial guess} is in the \textbf{admissible} domain\linebreak ${\min\limits_\Omega\det J>0}$, so a special untangling procedure is required for an arbitrary initial guess;
	\item[\textbf{(b)}] the fact that functional~\eqref{eq:winslow} is bounded \textbf{does not imply} that distortion measure~\eqref{eq:distortion} is \textbf{bounded}.
	It optimizes average deformation and admits integrable singularities. In order to provably suppress this behaviour, one has to consider modified hyperelastic material which forbids deformations with local distortion above prescribed threshold. 
\end{enumerate}
Two following sections discuss both points and lay the ground for our contribution in controlled stiffening of a hyperelastic material,
that provably allows us to build best known quasi-isometric maps.

\subsection{Untangling}
\label{sec:untangling}

With a slight abuse of notations, the density~\eqref{eq:distortion} can be rewritten as follows:
$$
f := (1 - \theta) \frac1d \frac{\tr J^\top J}{\left(\max(0, \det J)\right)^\frac2d} +  \theta \frac12  \frac{1 + \det^2 J}{\max(0, \det J)}
$$

\begin{figure}[!t]
	\centering
	\includegraphics[width=.6\linewidth]{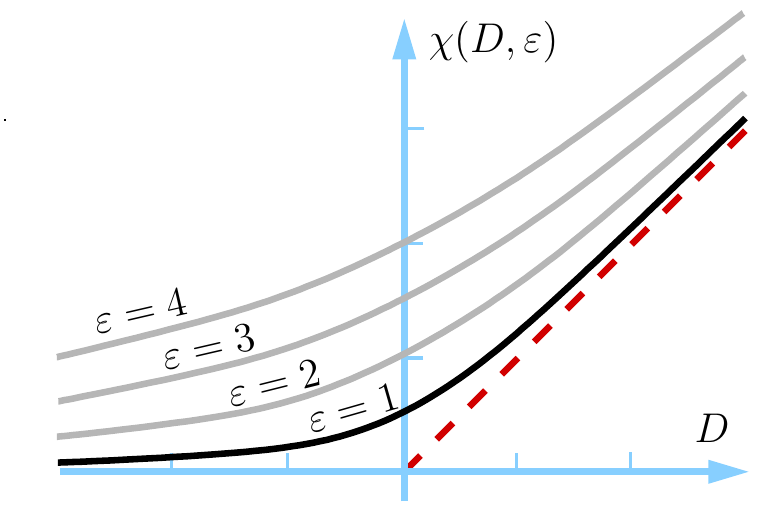}
	\caption{Regularization function for the denominator in Eq.~\eqref{eq:feps}.
		When $\varepsilon$ tends to zero, $\chi(\varepsilon, D)$ tends to $\max(0, D)$.}
	\label{fig:chi}
\end{figure}

\setcounter{figure}{\value{table}}
\begin{figure*}[!t]
	\centering
	\begin{minipage}{.45\textwidth}
		\captionsetup{name=Alg.}
		\hrulefill
		\begin{algorithmic}[1]
			\Input $X^{0}$ \Comment{arbitrary initial guess (vector of size $\#V \times d$)} 
			\Output $X$ \Comment{final foldover-free map (vector of size $\#V \times d$)}
			\State $k \leftarrow 0$;
			\Repeat
			\State compute $\varepsilon^k$; \Comment{decreasing sequence}
			\State $X^{k+1} \leftarrow \argmin\limits_{X} F(X,\varepsilon^k)$;
			\State $k \leftarrow k+1$;
			\Until{\footnotesize
				$\min\limits_{t\in 1\dots\#T} \det J_t^{k}>0$ \textbf{~and~} $F(X^{k}, \varepsilon^{k})>(1-10^{-3})\, F(X^{k-1}, \varepsilon^{k-1})$ 
			}
			\State $X\leftarrow X^k$;
		\end{algorithmic}
		\hrulefill
		\caption{Computation of a foldover-free map}\label{alg:untangling}
	\end{minipage}
	\hfill
	\begin{minipage}{.45\textwidth}
		\captionsetup{name=Alg.}
		\hrulefill
		\begin{algorithmic}[1]
			\Input $X^{0}$ \Comment{untangled initial guess (vector of size $\#V \times d$)} 
			\Output $X$ \Comment{final bounded distortion map (vector of size $\#V \times d$)}
			\State $k \leftarrow 0$;
			\Repeat
			\State compute $t^k$; \Comment{increasing sequence, Eq.~\eqref{eq:t.update}}
			\State $X^{k+1} \leftarrow \argmin\limits_{X} W(X,t^k)$;
			\State $k \leftarrow k+1$;
			\Until{\footnotesize
				$W(X^{k}, t^{k})>(1-10^{-3})\, W(X^{k-1}, t^{k-1})$
			}
			\State $X\leftarrow X^k$;
		\end{algorithmic}
		\hrulefill
		\caption{Quasi-isometric stiffening (QIS)}\label{alg:stiffening}
	\end{minipage}
	
\end{figure*}
\setcounter{figure}{\value{figure}}

Note that if an initial guess is not admissible (has inverted elements), then the function is not defined.
To overcome this problem, we can introduce function $\chi(D,\varepsilon)$ that will serve as a regularization for $\max(0,D)$ in the denominator:
\begin{equation}
	\chi(D, \eps) := \frac{D+\sqrt{\eps^2 + D^2}}{2}
	\label{eq:chi}
\end{equation}
When $\varepsilon$ tends to zero, $\chi(\varepsilon, D)$ tends to $\max(0, D)$, refer to Fig.~\ref{fig:chi} for an illustration.
Then, the density can be regularized as follows~\cite{Garanzha99}:
\begin{equation}
	f_\varepsilon(J) := (1 - \theta) \frac1d\frac{ \tr J^\top J}{\left(\chi(\det J, \varepsilon)\right)^\frac2d} +  \theta \frac12  \frac{1 + \det^2 J}{\chi(\det J, \varepsilon)}
	\label{eq:feps}
\end{equation}

Finally, Prob.~\eqref{eq:winslow} can be rewritten as follows:
\begin{equation}
	\lim\limits_{\varepsilon \rightarrow 0^+}\argmin\limits_{\vec{x}(\vec{\xi})}\int\limits_\Omega \ f_\varepsilon(J)  \,d \xi
	\label{eq:continuous}
\end{equation}
This formulation suggests an algorithm: build a decreasing sequence of the $\varepsilon^k\rightarrow 0$, and for each value $\varepsilon^k$ solve an optimization problem.
In other words, Prob.~\eqref{eq:continuous} offers a way of getting rid of foldovers by solving a continuation problem with respect to the parameter $\eps$.

The simplest way to discretize Prob.~\eqref{eq:continuous} is with first-order FEM:
the map $\vec x$ is piecewise affine with the Jacobian matrix $J$ being piecewise constant,
and can be represented by the coordinates of the vertices in the computational domain $\{\vec{x}_i\}_{i=1}^{\#V}$.
Let us denote the vector of all variables as $X := \left(\vec{x}_1^\top \dots \vec{x}_{\#V}^\top \right)^\top$.

A discretization of Prob.~\eqref{eq:continuous} can be written as follows:
\begin{gather}
	\label{eq:discrete}
	\lim\limits_{\varepsilon \rightarrow 0^+}\argmin\limits_{X} F(X,\varepsilon), \\
	\text{ where }\quad  F(X, \varepsilon) : =\sum\limits_{k=1}^{\#T} f_\eps(J_k) \vol(T_k) \nonumber,
\end{gather}
$\#V$ is the number of vertices, $\#T$ is the number of simplices, $J_k$ is the Jacobian matrix for the $k$-th simplex  and $\vol(T_k)$ is the signed volume of the simplex $T_k$ in the parametric domain. 

Prob.~\eqref{eq:discrete} can be solved with Alg.~\ref{alg:untangling}.
The algorithm itself is very simple, and has been published more than 20 years ago~\cite{Garanzha99}.
Note, however, that the crucial step here is the choice of the regularization sequence $\varepsilon^k$ (Alg.~\ref{alg:untangling}--line 3),
and until very recently only heuristics were used.
Last year~\cite{garanzha2021foldoverfree} have presented a way to build such a sequence that offers theoretical guarantees on untangling (refer to \S\ref{sec:finite-bounded-distortion} for a complete formulation).

Function $F(X,0)$ has an impenetrable infinite barrier on the boundary of the set of meshes with positive cell volumes\footnote{
	As a side note, this set has a quite complicated structure.
	For $k$-th simplex $\vol(\vec x(T_k))$ is a polylinear function of coordinates of its vertices, hence each term in \eqref{eq:discrete-admissible-set} defines a non-convex set.
	One can hardly expect that intersection of the sets in \eqref{eq:discrete-admissible-set} would result in a convex domain.
	Moreover, Ciarlet~\cite{Ciarlet} has proved that barrier property and convexity of the density of deformation energy are incompatible.
	Fortunately, barrier distortion measures can be polyconvex, as shown by J.~Ball \cite{ball1976convexity}.}

\begin{equation}
	\label{eq:discrete-admissible-set}
	\frac{\vol(\vec x(T_k))}{\vol(T_k)}  > 0, \quad k = 1, \dots, \#T
\end{equation}
which is a finite-dimensional approximation of the set ~$\det J > 0$.
Untangling in Prob.~\eqref{eq:discrete} is guaranteed because~\cite{garanzha2021foldoverfree} build a decreasing sequence $\varepsilon^k\rightarrow 0$
that forces the mesh to fall into the feasible set~\eqref{eq:discrete-admissible-set} of untangled meshes.
With some assumptions on the minimization toolbox chosen, untangling is guaranteed to succeed in a finite number of steps.

\subsection{Controlled quasi-isometric stiffening (QIS): minimization of maximum distortion}
\label{sec:stiffening}

In addition to untangling, by solving Prob.~\eqref{eq:discrete} we minimize \textbf{average} distortion of a map.
In this section we present the idea behind  \cite{garanzha2000barrier} that will allow us to compute a deformation with prescribed quality,
i.e. optimize the \textbf{maximum} distortion until it reaches a given bound.

Consider following variational problem related to construction of deformations with prescribed quality $0 \leq t^* < 1$:
\begin{equation}
	\argmin\limits_{\vec{x}(\vec{\xi})}\int\limits_\Omega \frac{f(J)}{1 - t^* f(J)}  \,d \xi,
	\label{eq:quality-integral}
\end{equation} 
Recall that $f(J)\geq 1$, so for this integral to be finite, a necessary condition is
\begin{equation}
	f(J) < \frac1{t^*}
	\label{eq:local-quality}
\end{equation}
Here parameter $t^*$ plays the role of the lower quality bound of the deformation. 
Note that the density in Prob.~\eqref{eq:quality-integral} is polyconvex and this variational problem is well-posed \cite{Garanzha2014}.

We propose following finite-dimensional approximation of Prob.~\eqref{eq:quality-integral}:
\begin{gather}
	\label{eq:discrete-quality}
	\lim\limits_{t \rightarrow t^*}\argmin\limits_{X} W(X,t), \\
	\text{ where }\quad  W(X, t) : =\sum\limits_{k=1}^{\#T} \frac{f(J_k)}{1 - t f(J_k)} \vol(T_k) \nonumber,
\end{gather}

It is important to note that a solution of Prob.~\eqref{eq:discrete} corresponds to a solution of Prob.~\eqref{eq:discrete-quality} with $t^*=0$,
i.e. when no bound on the maximum deformation is imposed.
But then, having reached the set~\eqref{eq:discrete-admissible-set}, we can build an increasing sequence of $t^k\rightarrow t^*$
to contract the set until the mesh falls into the set of deformations with prescribed quality $t^*$:
\begin{equation}
	\label{eq:discrete-admissible-set2}
	f(J_k)<\frac{1}{t^*}, \quad k = 1, \dots, \#T
\end{equation}

Alg.~\ref{alg:stiffening} sums up the optimization procedure, note how closely it is related to Alg.~\ref{alg:untangling}.
While the general idea was published more than 20 years ago, until now it remained unclear how to build this sequence $\{t^k\}$,
and this constitutes the main theoretical contribution of the present paper.

While we assume that parameter $t^*$ exists, fortunately we are not obliged to know it to make QIS algorithm work.
It is an important advantage over optimization algorithms which use prescribed distortion bound like LBD~\cite{LargeScaleBD:2015}, since in practice even rough estimates of this bound are not available.
Essentially, QIS algorithm by itself serves as a distortion bound estimation tool for problems of any complexity. 
Refer to App.~\ref{app:relation} for the argumentation of the fact that minimizing our distortion measure implies minimization of the upper bound for the quasi-isometry constant.

\section{Theoretical guarantees for quasi-isometric stiffening}
\label{sec:finite-bounded-distortion}

This section provides our main result, namely, a way to build an increasing sequence $\{t^k\}$ that allows us to effectively contract the feasible set until we reach the goal.
Untangling and stiffening are very closely related, so let us first restate the main result of \cite[Theorem 1]{garanzha2021foldoverfree},
it will allow us to highlight the similarity between the approaches.

\begin{theorem}
	\label{th:th}
	Let us suppose that the feasible set of untangled meshes~\eqref{eq:discrete-admissible-set} is not empty.
	We also suppose that for solving $X^{k+1} \leftarrow \argmin\limits_{X} F(X,\varepsilon^k)$
	we have a minimization algorithm satisfying the following efficiency conditions for some $0<\sigma<1$:
	
	For each iteration $k$,
	\begin{itemize}
		\item \textbf{either} the essential descent condition holds
		\begin{equation}
			F(X^{k+1}, \eps^k) \leq (1-\sigma) F(X^{k}, \eps^k),
		\end{equation}
		\item \textbf{or} the vector $X^{k}$ satisfies the quasi-minimality condition:
		\begin{equation}
			\min\limits_{X} F(X, \eps^k) > (1-\sigma)F(X^{k}, \eps^k).
		\end{equation}
	\end{itemize}
	Then the feasible set~\eqref{eq:discrete-admissible-set} is reachable by solving a finite number of minimization problems in $X$ with $\eps^k$ fixed for each problem.
\end{theorem}

In this theorem Garanzha et. al. not only proved that there exists a regularization parameter sequence $\{\varepsilon^k\}_{k=0}^K$ leading to $F(X^{K}, 0) < + \infty$,
but also provided an actual update rule for $\varepsilon^k$, refer to~\cite[Eq. (6)]{garanzha2021foldoverfree}.
Inspired by these results, we formulate a very similar theorem allowing us to build maps with bounded distortion in a finite number of steps.

We also provide a way to build an increasing sequence $\{t^k\}$ to be used in Alg.~\ref{alg:stiffening}--line 3:
denote by $f_i(X^{k+1})$ the distortion for the element $i$, and by $f^{k+1}_+$ the maximal distortion value over the mesh $X^{k+1}$,
$f_+^{k+1} := \max_i f_i(X^{k+1})$.
We  propose to use the following update rule for $t^{k+1}$:
\begin{equation}
	t^{k+1} := t^k + \sigma \frac{1 - t^k f_+^{k+1}}{f_+^{k+1}},
	\label{eq:t.update}
\end{equation}
where $0<\sigma<1$ is again the performance index of the minimization toolbox.
Clearly, formula~\eqref{eq:t.update} does not involve $t^*$.
Alg.~\ref{alg:stiffening} along with this update rule define our quasi-isometric stiffening (QIS) algorithm.

Now we are ready to formulate the stiffening theorem.
\begin{theorem}
\label{th:th2}
	Let us suppose that the feasible set of bounded distortion meshes~\eqref{eq:discrete-admissible-set2} is not empty, namely there exists a constant
	$0 < t^* < 1$ and a mesh $X^*$ satisfying $W(X^*,t^*) < + \infty$.
	We also suppose that
	for solving $X^{k+1} \leftarrow \argmin\limits_{X} W(X,t^k)$
	we have a minimization algorithm satisfying the following efficiency conditions for some $0<\sigma<1$:
	
	For each iteration $k$,
	\begin{itemize}
		\item \textbf{either} the essential descent condition holds
		\begin{equation}
			\label{teorem2.cond6}
			W(X^{k+1}, t^k) \leq (1-\sigma) W(X^{k}, t^k),
		\end{equation}
		\item \textbf{or} the vector $X^{k}$ satisfies the quasi-minimality condition:
		\begin{equation}
			\label{teorem2.cond6-2}
			\min\limits_{X} W(X, t^k) > (1-\sigma)W(X^{k}, t^k).
		\end{equation}
	\end{itemize}
	Then the feasible set~\eqref{eq:discrete-admissible-set2} is reachable by solving a finite number of minimization problems in $X$ with $t^k$ fixed for each problem.
	In other words, under a proper choice of the continuation parameter sequence $\{t^k\}_{k=0}^K$, we obtain $W(X^{K}, t^K) < + \infty$.
\end{theorem}

\begin{proof}
	The main idea is very simple: update rule~\eqref{eq:t.update} defines an increasing sequence $\{t^k\}_{k=0}^\infty$.
	We will show that the corresponding sequence $\{W(X^k, t^k)\}_{k=0}^\infty$ is bounded from above.
	Then we can prove that the admissible set~\eqref{eq:discrete-admissible-set2} is reachable in a finite number of steps by a simple \emph{reductio ad absurdum} argument.
	More precisely, if the feasible set is not reachable, then $W(X^k, t^k)$ must grow without bounds, what contradicts the boundedness.
	
	To prove that $\{W(X^k, t^k)\}_{k=0}^\infty$ is bounded from above, we analyse the behavior at some iteration $k$.
	First of all, if we couple update rule~\eqref{eq:t.update} with the fact that for any $t_2 > t_1$ the ratio $ \frac{1 - t_1 \psi}{1 - t_2 \psi}$ is increasing function of argument $\psi$,
	we can see that the following inequality holds:
	\begin{equation}
		\label{f_bounded}
		(1 - \sigma) W(X^{k+1}, t^{k+1}) \leq W(X^{k+1},t^{k}).
	\end{equation}
	More precisely,
	\[
	W(X^{k+1}, t^{k+1}) = \sum_i \frac{1 - t^k f_i(X^{k+1})}{1 - t^{k+1} f_i(X^{k+1})} \frac{f_i(X^{k+1})}{1 - t^k f_i(X^{k+1})} \leq
	\]
	\[
	\sum_i \frac{1 - t^k f_+(X^{k+1})}{1 - t^{k+1} f_+(X^{k+1})} \frac{f_i(X^{k+1})}{1 - t^k f_i(X^{k+1})} = \frac1{1 - \sigma} W(X^{k+1},t^{k}).
	\]
	
	Then, at each iteration $k$, either condition~\eqref{teorem2.cond6} or condition~\eqref{teorem2.cond6-2} must be satisfied. Let us consider both cases.
	\begin{itemize}[wide=0pt,itemindent=2em]
		\item[\textbf{Cond.~\eqref{teorem2.cond6} holds:}] in this case function $W$ actually decreases. Eq.~\eqref{f_bounded} combined with Cond.~\eqref{teorem2.cond6} directly imply that 
		$$
		W(X^{k+1}, t^{k+1}) \leq W(X^k,t^{k}).
		$$
		\item[\textbf{Cond.~\eqref{teorem2.cond6-2} holds:}] under assumption that $t^{k+1} < t^*$, by combining Eq.~\eqref{f_bounded} and Cond.~\eqref{teorem2.cond6-2}, we obtain:
		\[
		W(X^{k+1}, t^{k+1}) \leq \frac1{1-\sigma} W(X^{k+1}, t^{k}) \leq
		\]
		\[
		\frac1{(1-\sigma)^2} \min_X W(X, t^{k}) \leq \frac1{(1-\sigma)^2}  W(X^*, t^{k}) \leq 
		\]
		\[
		\leq \frac1{(1-\sigma)^2}  W(X^*, t^*)
		\]
	\end{itemize}
	To sum up, in the first case the value of function $W$ decreases, and in the second case it does not exceed a global bound, therefore, the sequence $\{W(X^k, t^k)\}_{k=0}^\infty$ is bounded from above.
	
	Now we are ready for the main result. Suppose that for an infinite sequence $\{X^k, t^k\}_{k=0}^{\infty}$ we never reach the given quality bound $t^*$.
	In other words, we have $t^k < t^*, \ k=0,\dots\infty$.
	Then the following inequality holds (apply formula~\eqref{eq:t.update} $k$ times):
	\[
	t^k - t^0 = \sigma \sum\limits_{j=0}^{k-1} \left(\frac1{f_+^{j+1}} - t^j\right) \leq 1.
	\]
	In the infinite sum each term is strictly positive, hence we can extract a subsequence
	\[
	1 - t^{j_s} f_+^{j_s+1} \to 0^+.
	\]
	This fact obviously contradicts boundedness of the functional, allowing us to conclude our proof.
\end{proof}

\begin{remark}
	An important corollary of Th.~\ref{th:th2} is that, provided that the admissible set~\eqref{eq:discrete-admissible-set2} is not empty,
	there exists an iteration $K<\infty$ such that the global minimum of the function $W(X, t^K)$ belongs to the admissible set.
	The proof is rather obvious: suppose we have an idealized minimizer such that $X^{k+1} = \argmin\limits_X W(X, t^k)$.
	This minimizer always satisfies the conditions of Th.~\ref{th:th2}, therefore it can reach the distortion bound in a finite number of steps.
\end{remark}

In practice, just like in \cite{garanzha2021foldoverfree}, the global estimate $\sigma$ is not known in advance.
Garanzha et al. suggest to compute the local descent coefficient for each minimization step, and so do we.
Instead of $\sigma$ in the update rule \eqref{eq:t.update}, we use $\sigma^k$ defined as follows:
\[
\sigma^k := \max \left(1 -  \frac{W(X^{k+1}, t^{k})}{ W(X^{k}, t^{k})}, \sigma_0\right),
\]
where $\sigma_0 > 0$ is a constant (we chose $\sigma_0=1/10$ in all our experiments).

\section{Results and discussion}
\label{sec:results}

Recall that our contribution is two-fold:
\textbf{(a)} we provide an algorithm to compute a deformation where all cells have a deformation quality above given threshold,
and \textbf{(b)} we guarantee local invertibility of a global parameterization for quad generation.
This section is organized in two subsections accordingly.

\subsection{Quality optimization results}

In order to check the ability of the variational method \eqref{eq:discrete-quality} to improve worst quality elements in the deformed meshes, we have performed a series of tests.

\begin{figure}[!t]
	\centering
	\includegraphics[width=.6\linewidth]{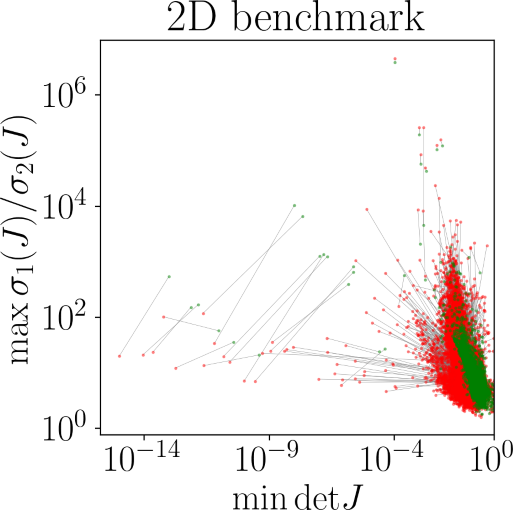}
	
	\vspace{5mm}
	
	\includegraphics[width=.6\linewidth]{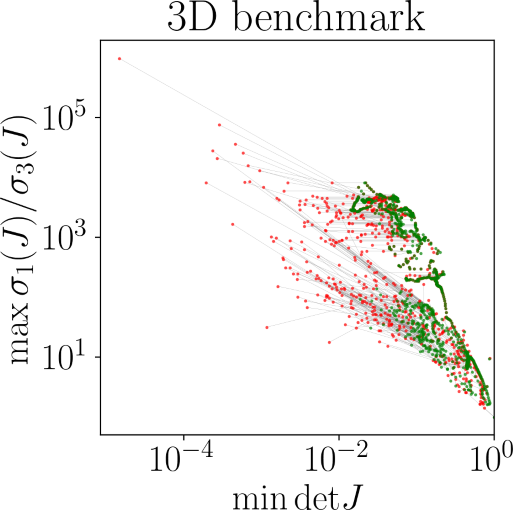}
	\caption{Quality plot of the resulting locally injective maps for every challenge from the database provided by~\cite{Du2020}.
		Each dot corresponds to a quality of the corresponding map reduced to two numbers: the maximum stretch and the minimum scale.
		Untangling results~\cite{garanzha2021foldoverfree} are shown in red, whereas our results are shown in green (we took $\theta=\frac12$).
		The gray lines show the correspondence in the results.
		\textbf{Top:} mapping quality on the 2D dataset (10743 challenges). \textbf{Bottom:} mapping quality on the 3D dataset (904 challenges).
	}
	\label{fig:scatter}
\end{figure}

\paragraph*{Constrained boundary injective mapping}
Along with their paper~\cite{Du2020}, Du et al. have published a benchmark
with a huge number of 2D and 3D constrained boundary injective mapping challenges.
We have computed initial injective maps with~\cite{garanzha2021foldoverfree}, and optimized the quality of the maps.
Fig.~\ref{fig:scatter} shows the scatter plot for every mapping challenge from the database.
For each problem the plot has two dots:
the red one corresponds to the input injective map reduced to two numbers, namely, the maximum stretch and the minimum scale.
The green dot is the quality of the map after our optimization.
Gray segments show the correspondence between the dots.
As can be seen from the plot, in the vast majority of cases our optimization improves both the maximum stretch and minimum scale quality measures
despite the fact that the boundary is locked.


\begin{figure*}[!t]
	\begin{minipage}[!t]{.3\linewidth}\vspace{0pt}
		\centering
		\includegraphics[width=\linewidth]{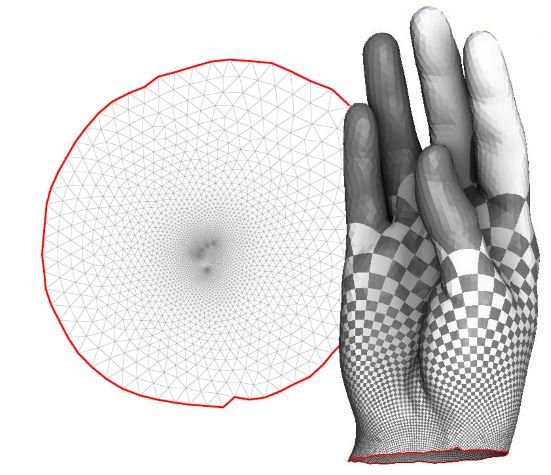}
		$\max \frac{\sigma_1}{\sigma_2} \approx 1.61$
		\includegraphics[width=.8\linewidth]{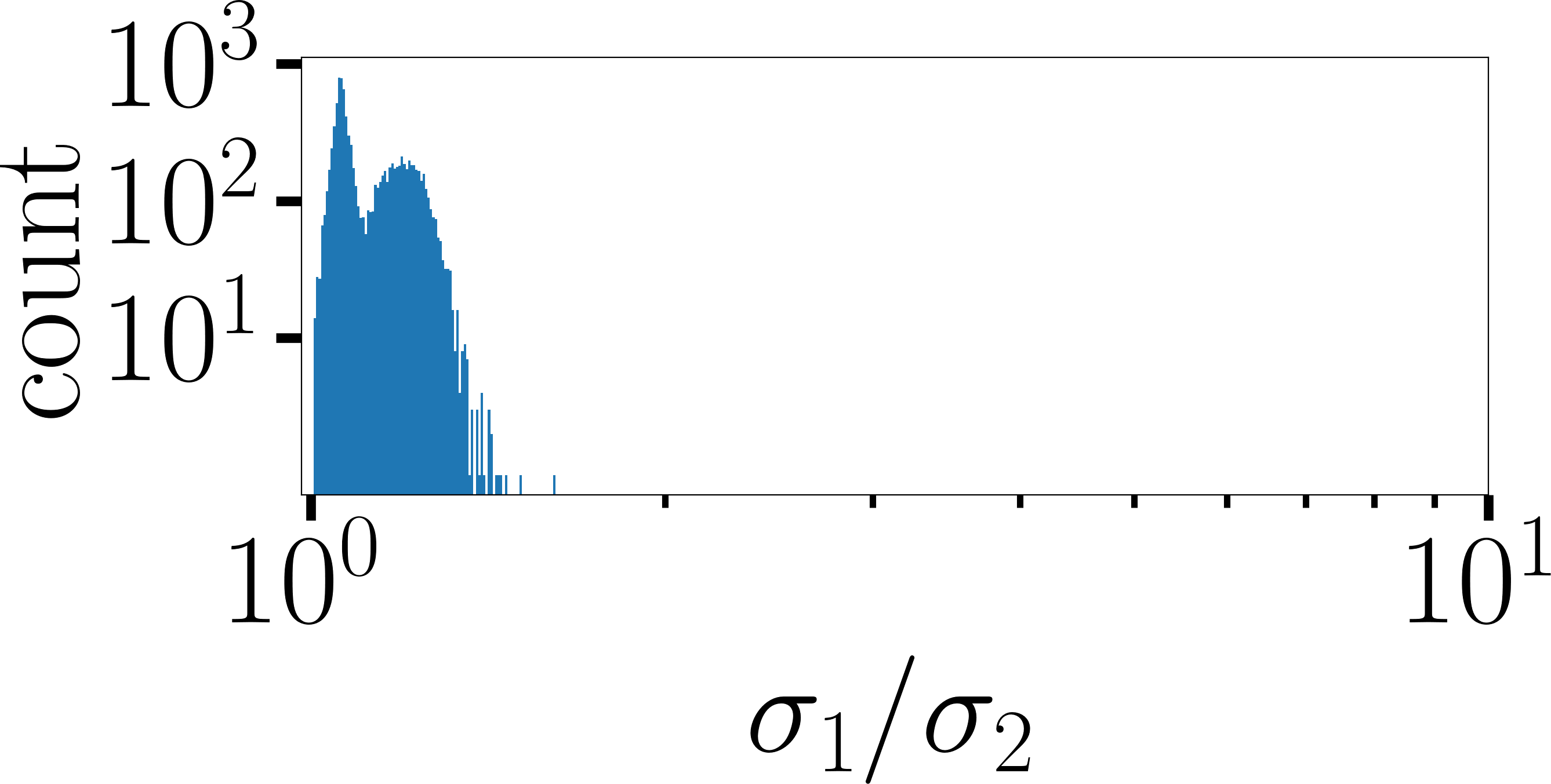}\\
		\textbf{(a)}
	\end{minipage}
	\begin{minipage}[!t]{.3\linewidth}\vspace{0pt}
		\centering
		\includegraphics[width=\linewidth]{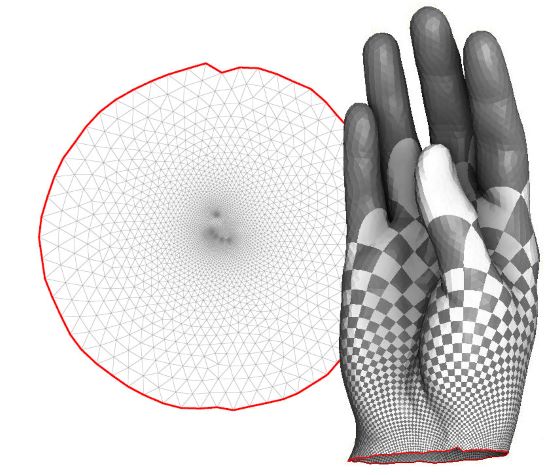}
		$\max \frac{\sigma_1}{\sigma_2} \approx 2.0$
		\includegraphics[width=.8\linewidth]{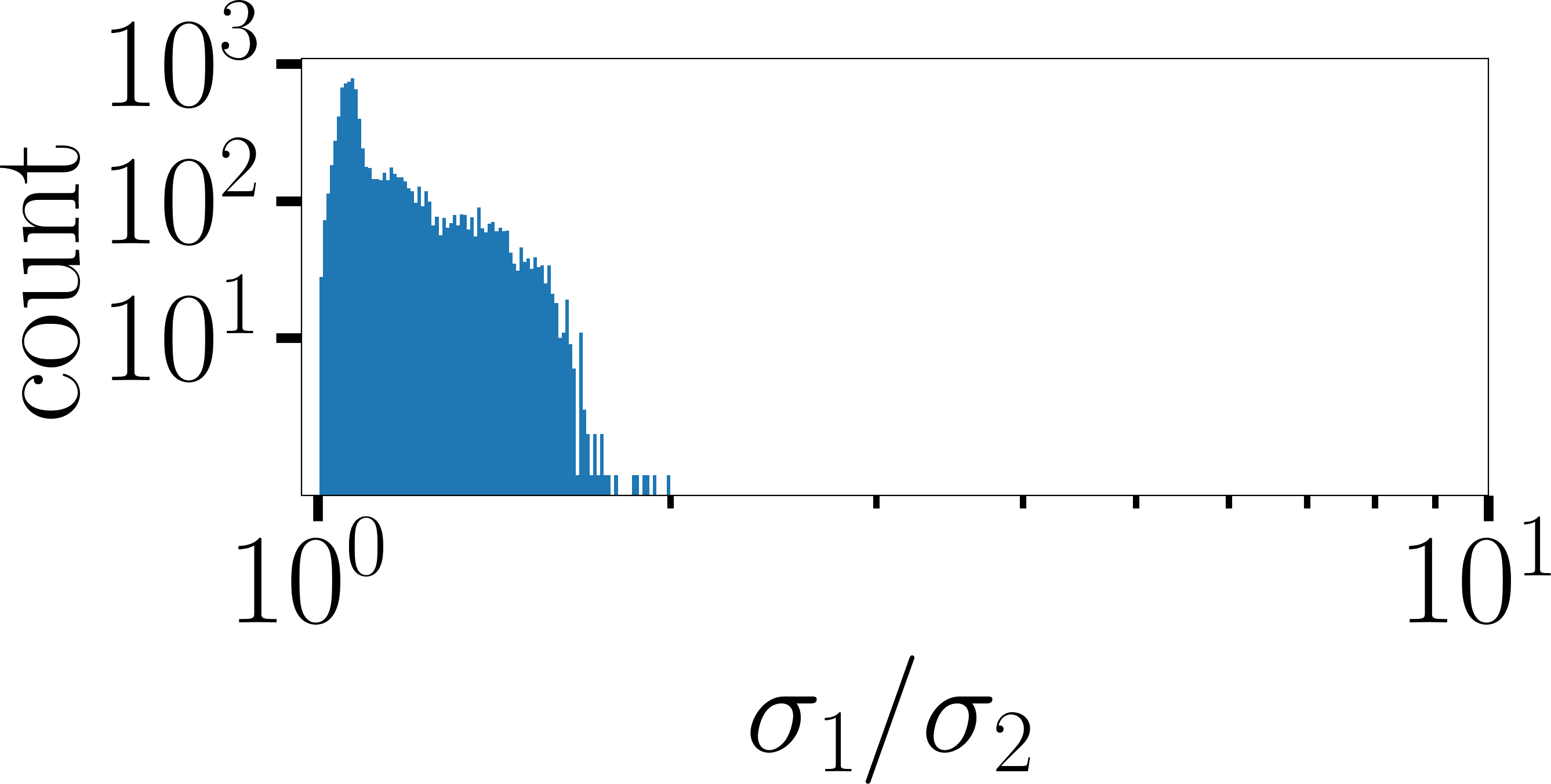}\\
		\textbf{(c)}
	\end{minipage}
	\begin{minipage}[!t]{.3\linewidth}\vspace{0pt}
		\centering
		\includegraphics[width=\linewidth]{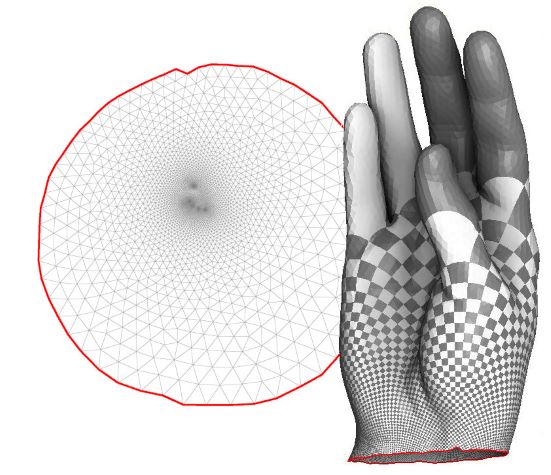}
		$\max \frac{\sigma_1}{\sigma_2} \approx 1.33$
		\includegraphics[width=.8\linewidth]{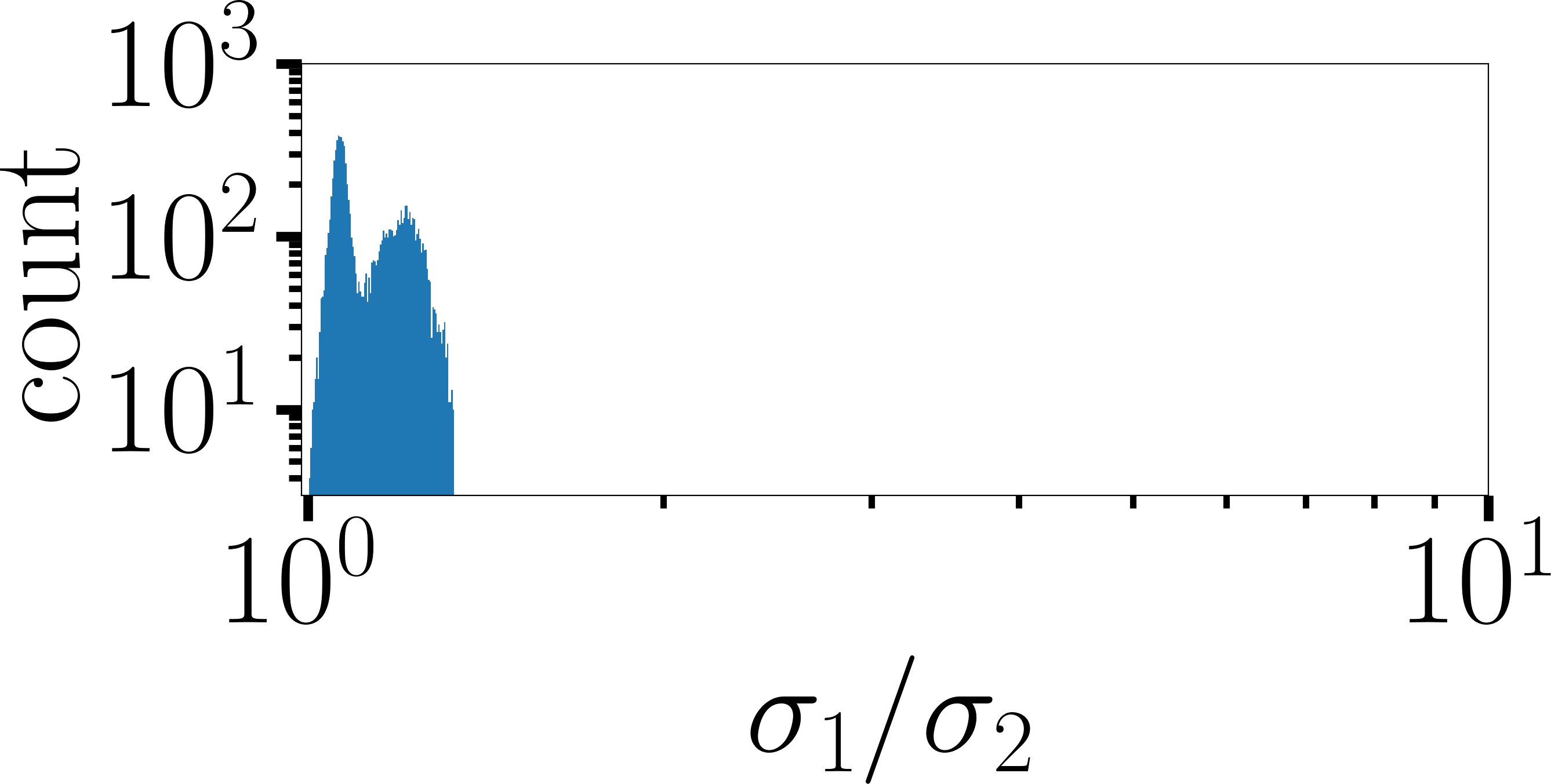}\\
		\textbf{(e)}
	\end{minipage}
	
	\begin{minipage}[!t]{.3\linewidth}\vspace{0pt}
		\centering
		\includegraphics[width=\linewidth]{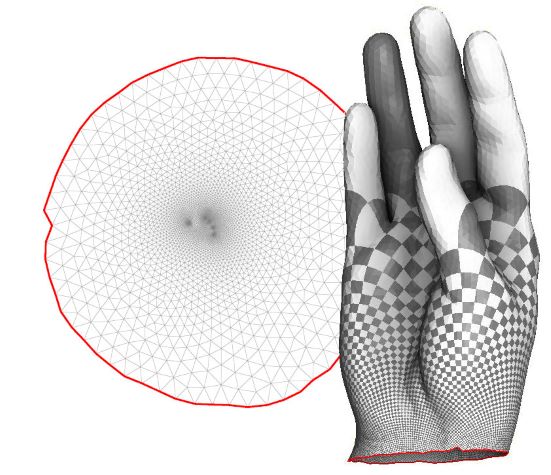}
		$\max \frac{\sigma_1}{\sigma_2} \approx 1.63$
		\includegraphics[width=.8\linewidth]{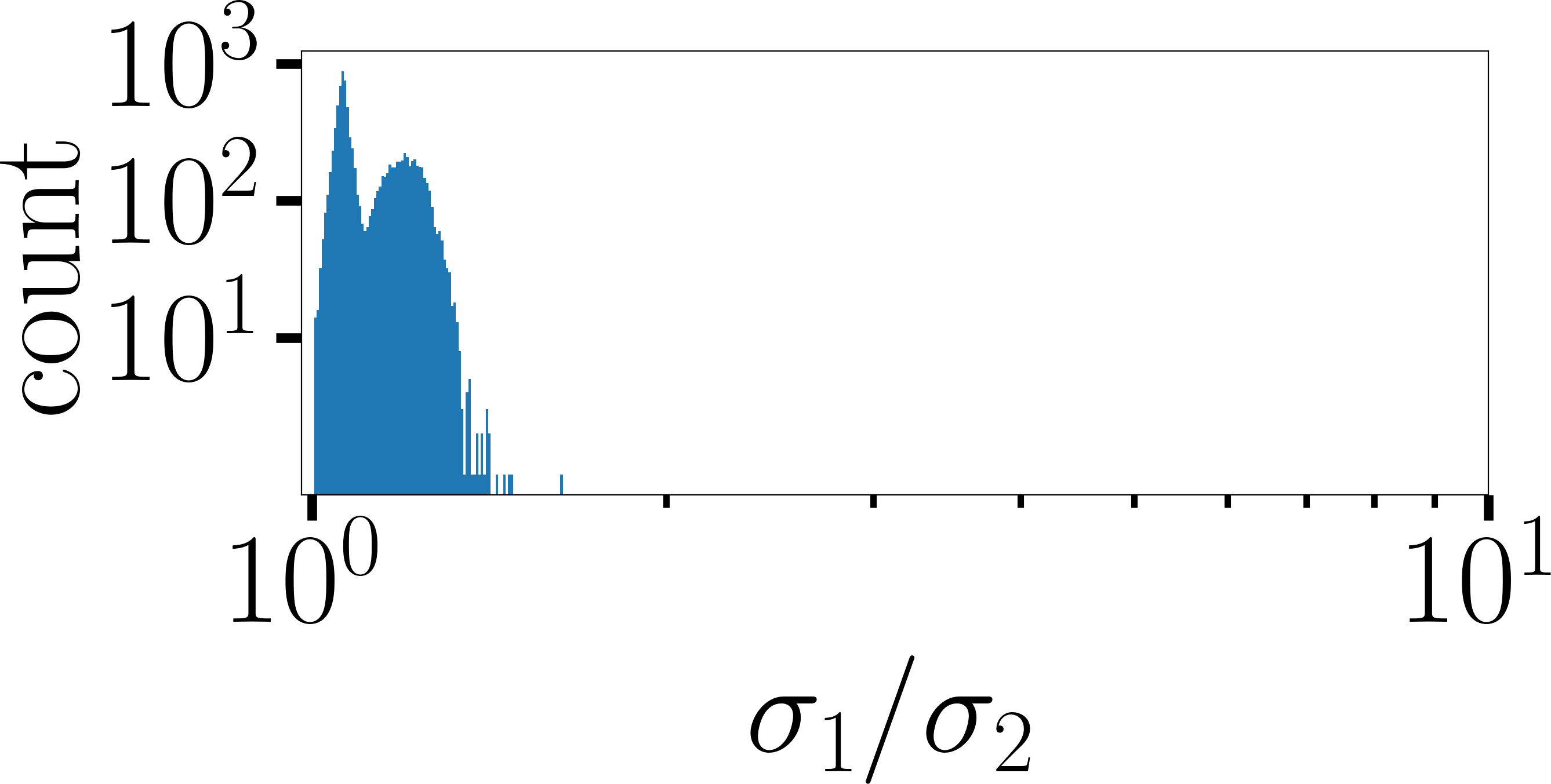}\\
		\textbf{(b)}
	\end{minipage}
	\begin{minipage}[!t]{.3\linewidth}\vspace{0pt}
		\centering
		\includegraphics[width=\linewidth]{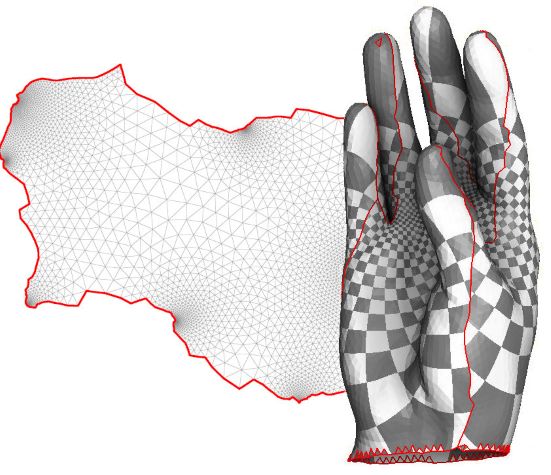}
		$\max \frac{\sigma_1}{\sigma_2} \approx 1.87$
		\includegraphics[width=.8\linewidth]{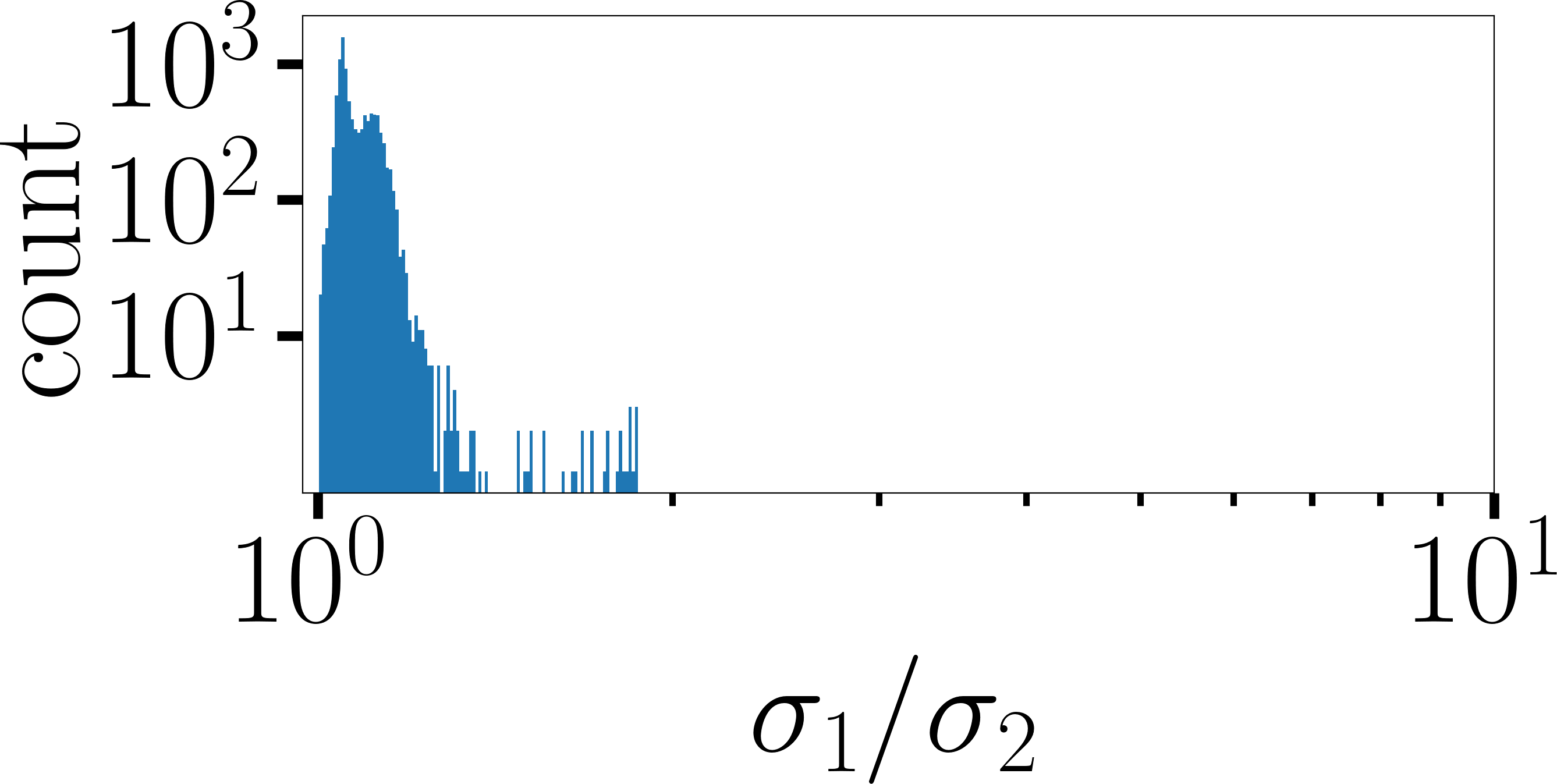}\\
		\textbf{(d)}
	\end{minipage}
	\begin{minipage}[!t]{.3\linewidth}\vspace{0pt}
		\centering
		\includegraphics[width=\linewidth]{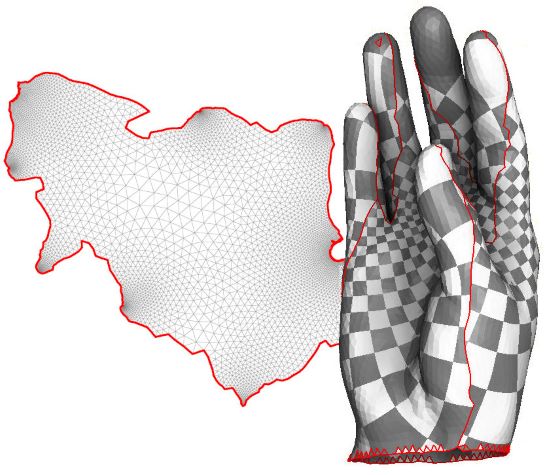}
		$\max \frac{\sigma_1}{\sigma_2} \approx 1.28$
		\includegraphics[width=.8\linewidth]{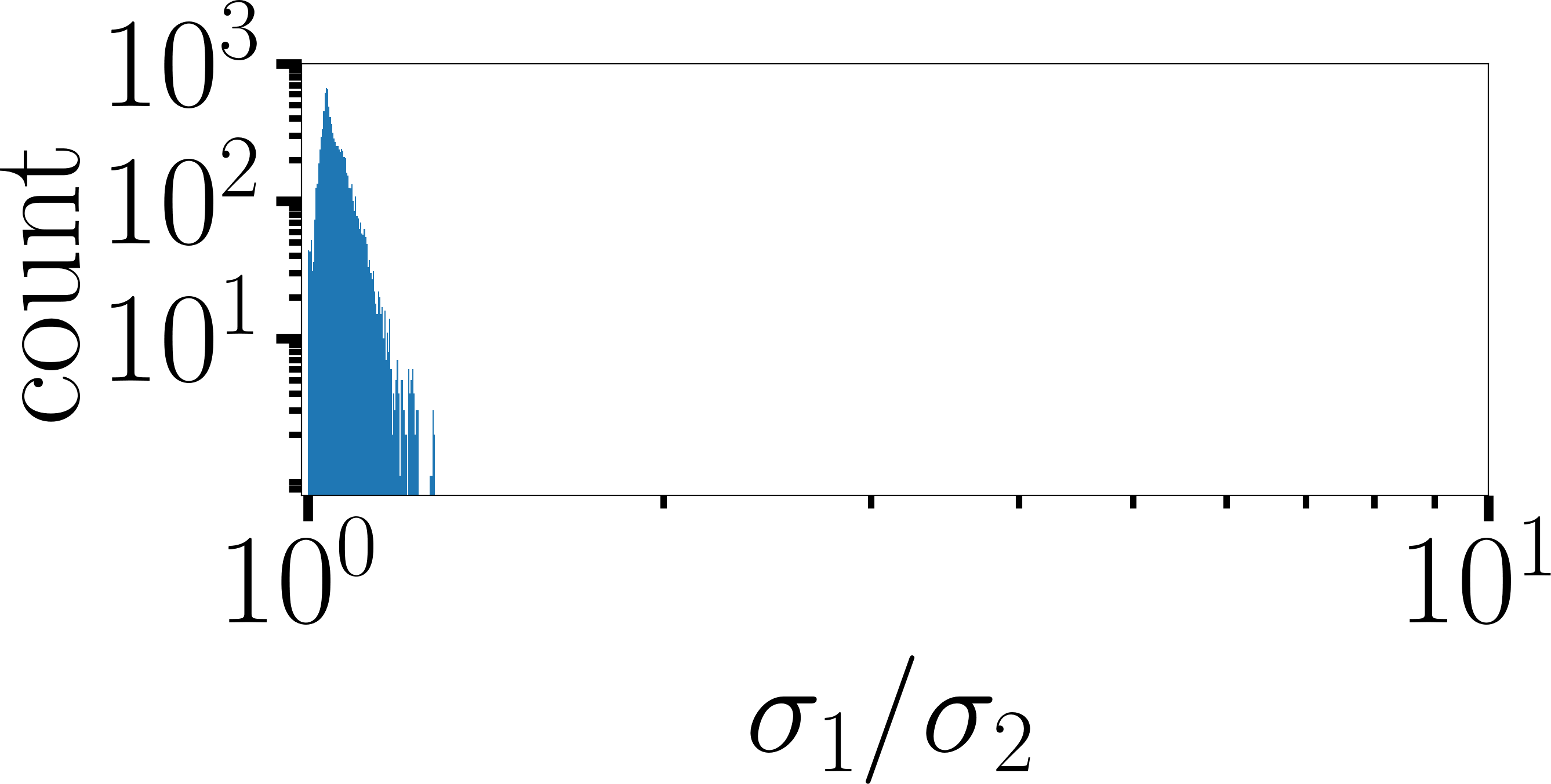}\\
		\textbf{(f)}
	\end{minipage}
	\caption{Discrete conformal mapping test. \textbf{(a):} LSCM~\cite{levy2002}, \textbf{(b):} BFF~\cite{BFF}, \textbf{(c):} untangling~\cite{garanzha2021foldoverfree},
		\textbf{(d):} CEPS~\cite{CEPS}, \textbf{(e):} QIS applied to the untangling result (see \textbf{(c)}), \textbf{(f):} QIS applied to the surface cut by CEPS (see \textbf{(e)}).
	}
	\label{fig:conformal}
\end{figure*}

\paragraph*{Conformal mapping} While in theory conformal maps are beyond the scope of QIS algorithm, in practice it was found to be highly successful.
As was suggested in~\cite{garanzha2021foldoverfree}, approximation of conformal mappings with quasi-isometric ones is indeed a good and stable numerical solution.
We have computed six discrete conformal maps for a triangular mesh of a hand (refer to Fig.\ref{fig:conformal}).
To compare quality of the maps, we use the condition number of the Jacobian matrix $\frac{\sigma_1(J)}{\sigma_2(J)}$, where $\sigma_1$ and $\sigma_2$ stand for the (ordered) singular values of $J$.
First we have computed four maps with state-of-the art methods:
\begin{itemize}[wide=0pt,itemindent=2em]
	\item[\textbf{(a)}] The easiest one to compute is the least squares conformal map \cite{levy2002}.
	This well-known method requires solving one linear system with a symmetric positive definite matrix.
	The idea behind LSCM is to compute a $P^1$ finite element approximation of the Cauchy-Riemann conditions over all triangles of the mesh.
	\item[\textbf{(b)}] Second map was obtained by applying the boundary first flattening method \cite{BFF} with boundary log-scale factors set to zero.
	This choice of boundary condition leads to the conformal map with minimal area distortion \cite[App. E]{CETM}.
	\item[\textbf{(c)}] Third map is the result of the elliptic smoother~\cite{garanzha2021foldoverfree},
	we have obtained it by solving Eq.~\eqref{eq:quality-integral} with $\theta = 0$, $t=0$.
	\item[\textbf{(d)}] Fourth map was obtained with CEPS~\cite{CEPS}. Note that this method can alter the input triangulation, and the way it handles surfaces with boundary is
	to glue together two copies of the input mesh along the boundary, introduce cone singularities and cut the mesh.
	The seams are shown in red in Fig.~\ref{fig:conformal}-\textbf{(e)}.
\end{itemize}

Finally, starting from the untangling \textbf{(c)} and CEPS \textbf{(d)}, we have computed optimal discrete conformal maps by solving Eq.~\eqref{eq:quality-integral} with $\theta = 0$ while maximizing for $t$.
The results are shown in Fig.~\ref{fig:conformal}-\textbf{(e)} and Fig.~\ref{fig:conformal}-\textbf{(f)}.
Log-log element quality histograms show that we improve considerably the quality of input maps.
It is easy to see that the maximum condition number of the Jacobian matrix is consistently better in our maps (1.33 and 1.28, respectively).

We were quite surprised by the fact that QIS algorithm outperforms highly elaborated algorithms for conformal parameterizations, since this problem is beyond its derivation principles.
One possible explanation, which is referred to in \cite{garanzha2021foldoverfree}, is the hypothesis that approximation of conformal mappings by quasi-isometric mappings with growing quasi-isometry constants is a good theoretical and engineering way for conformal parameterizations.


\paragraph*{Comparison with ABCD}
Fig.~\ref{fig:teaser} shows a comparison with Adaptive Block Coordinate Descent for Distortion Optimization~\cite{Naitsat2019}.
We took a tetrahedral mesh of a combination wrench, and we imposed positional constraints on the vertices located on both ends of the wrench.
As before, we optimize the quality starting from the untangling result (Fig.~\ref{fig:teaser}-\textbf{(c)}).

While the measures of deformation differ (ABCD uses ARAP-like energy), most of our input map elements have comparable quality to ABCD.
Note however that ABCD produces several elements of very bad quality, and untangling gains an order of magnitude over the worst element quality:
the maximum stretch $\max \sigma_1/\sigma_3\approx 28$, the minimum scale $\min \det J \approx 0.06$ for ABCD
and $\max \sigma_1/\sigma_3\approx 2.35$, $\min \det J \approx 0.69$ for the untangling.
In our turn, we improve the quality even further, our quality measures are $\max \sigma_1/\sigma_3\approx 1.22$ and $\min \det J \approx 0.86$.


\paragraph*{Comparison with Injective Deformation Processing}

\begin{figure}[!p]
	\begin{minipage}[!t]{.45\linewidth}\vspace{0pt}
		\centering
		\includegraphics[width=\linewidth]{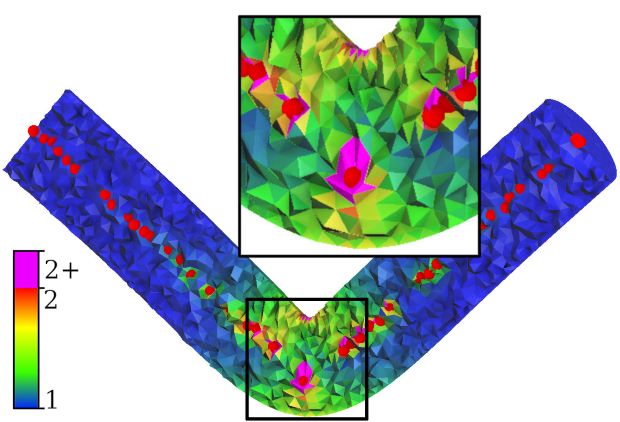}
		\includegraphics[width=\linewidth]{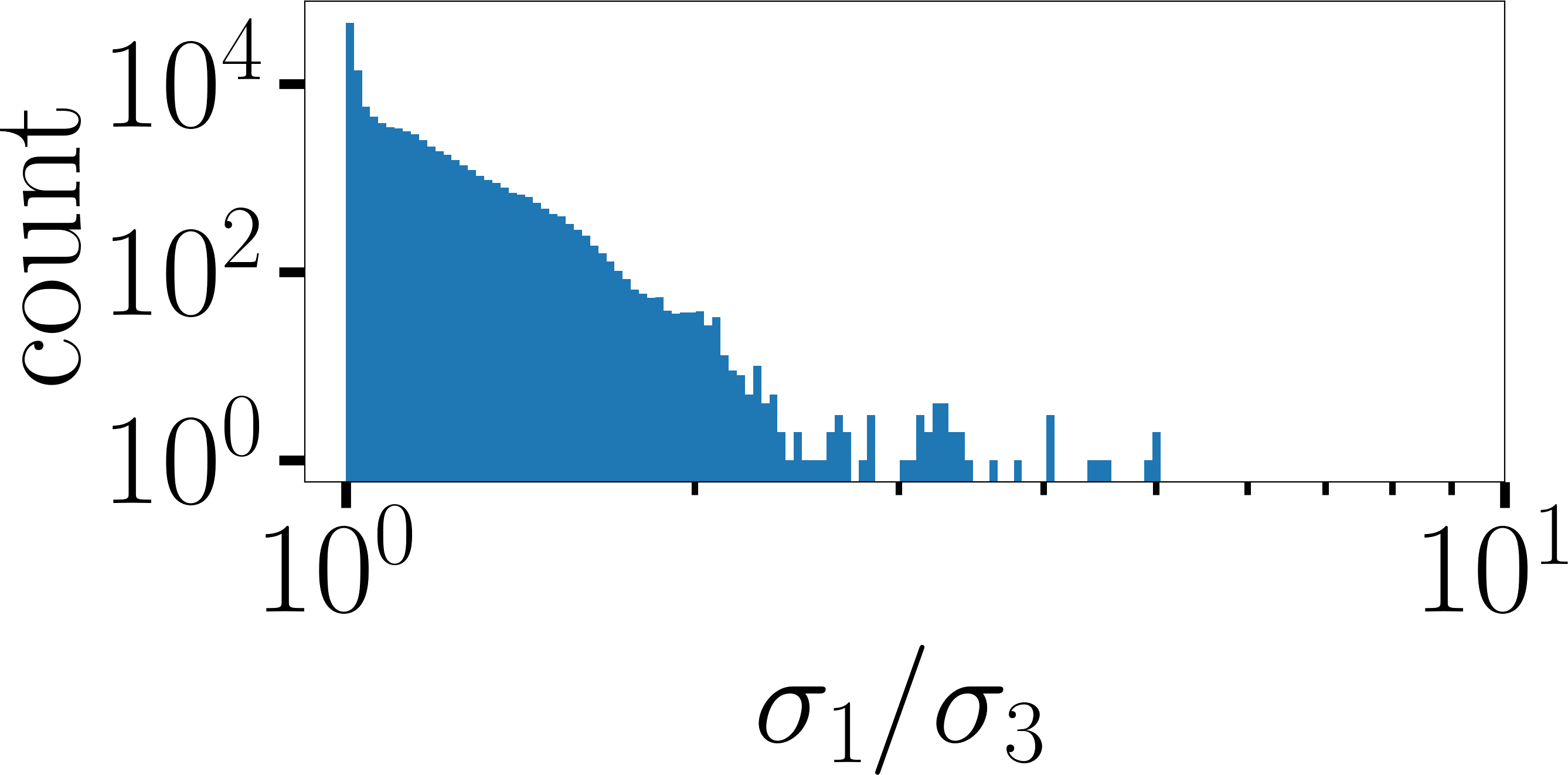}
		\vspace{2mm}
		\includegraphics[width=\linewidth]{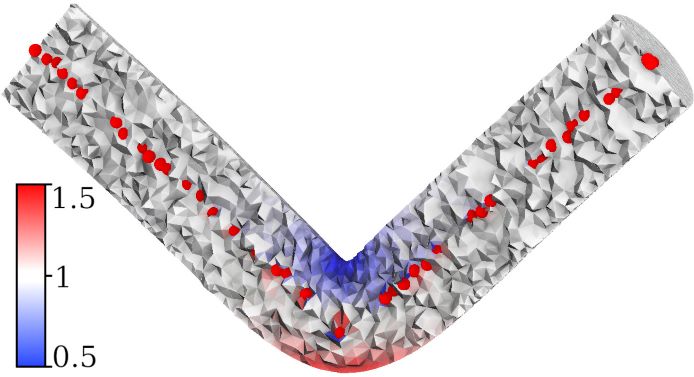}
		\includegraphics[width=\linewidth]{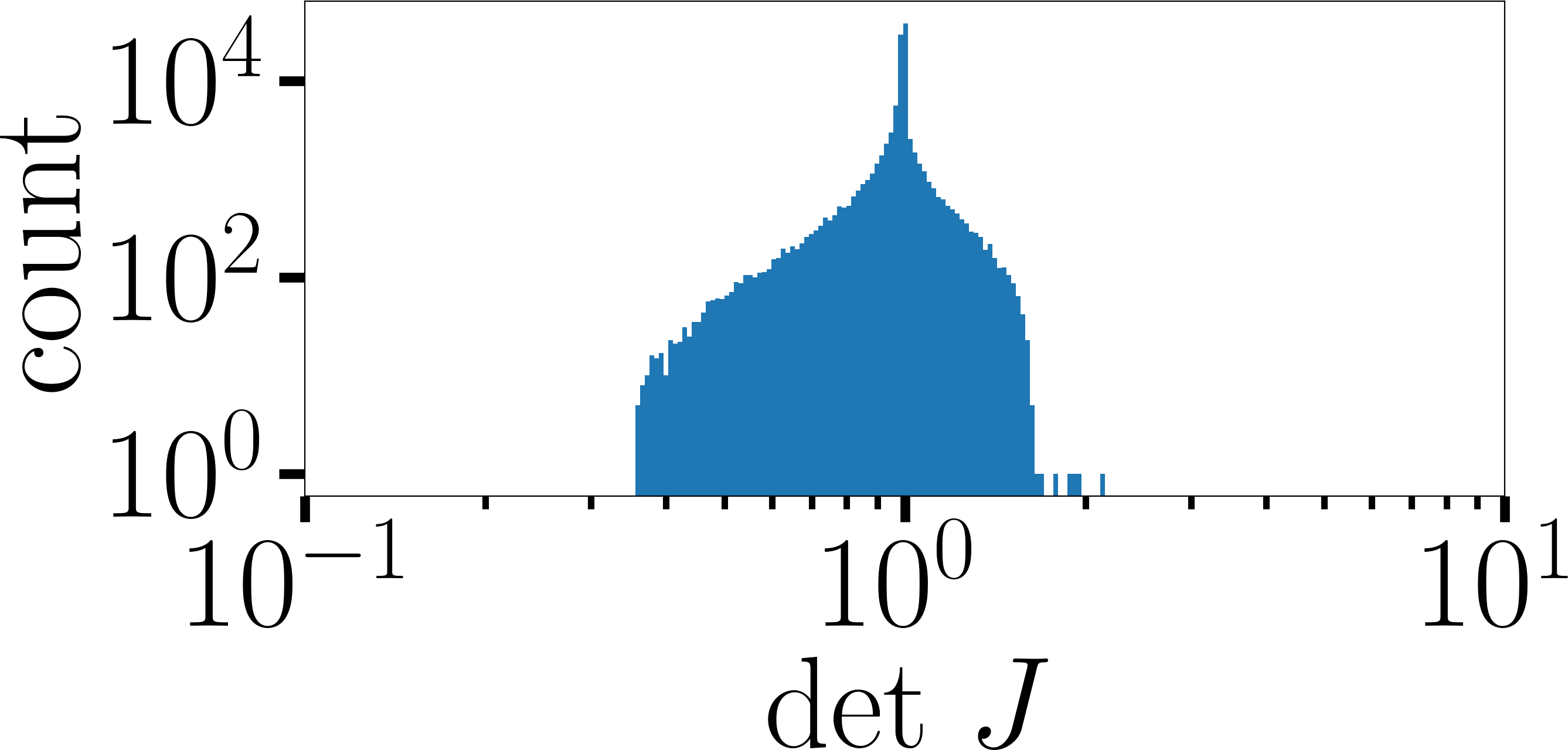}
		\textbf{(a)}
	\end{minipage}
	\begin{minipage}[!t]{.45\linewidth}\vspace{0pt}
		\centering
		\includegraphics[width=\linewidth]{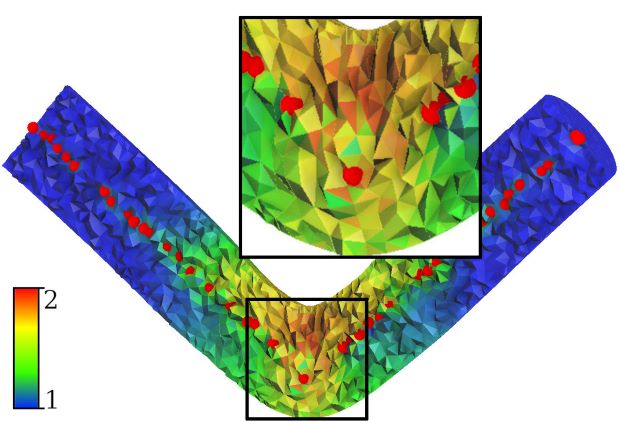}
		\includegraphics[width=\linewidth]{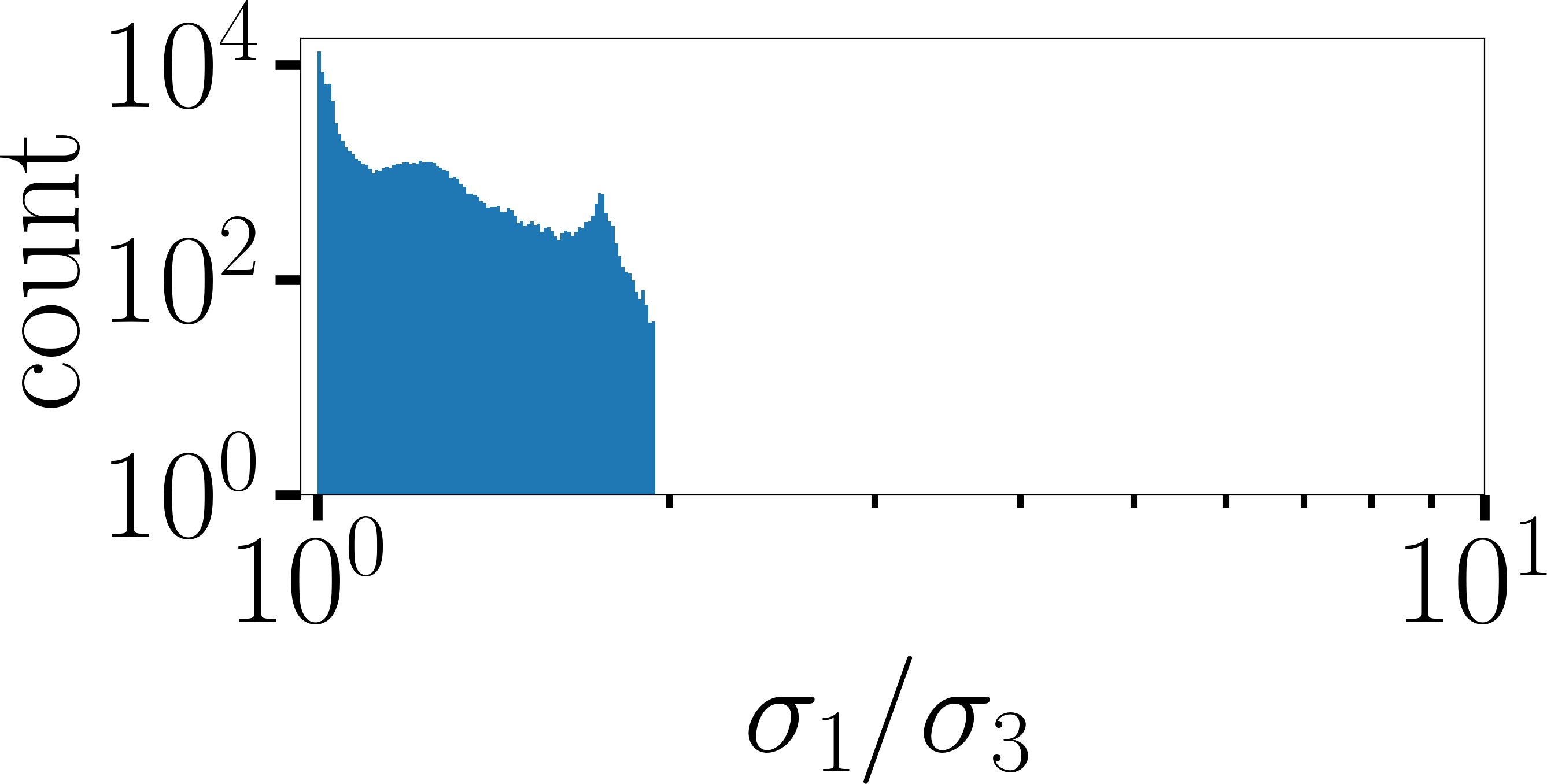}
		\vspace{2mm}
		\includegraphics[width=\linewidth]{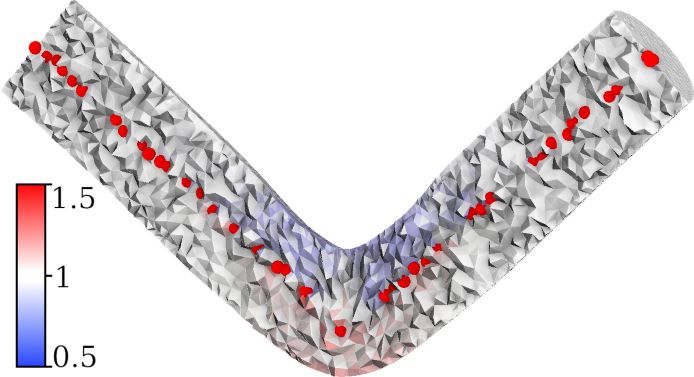}
		\includegraphics[width=\linewidth]{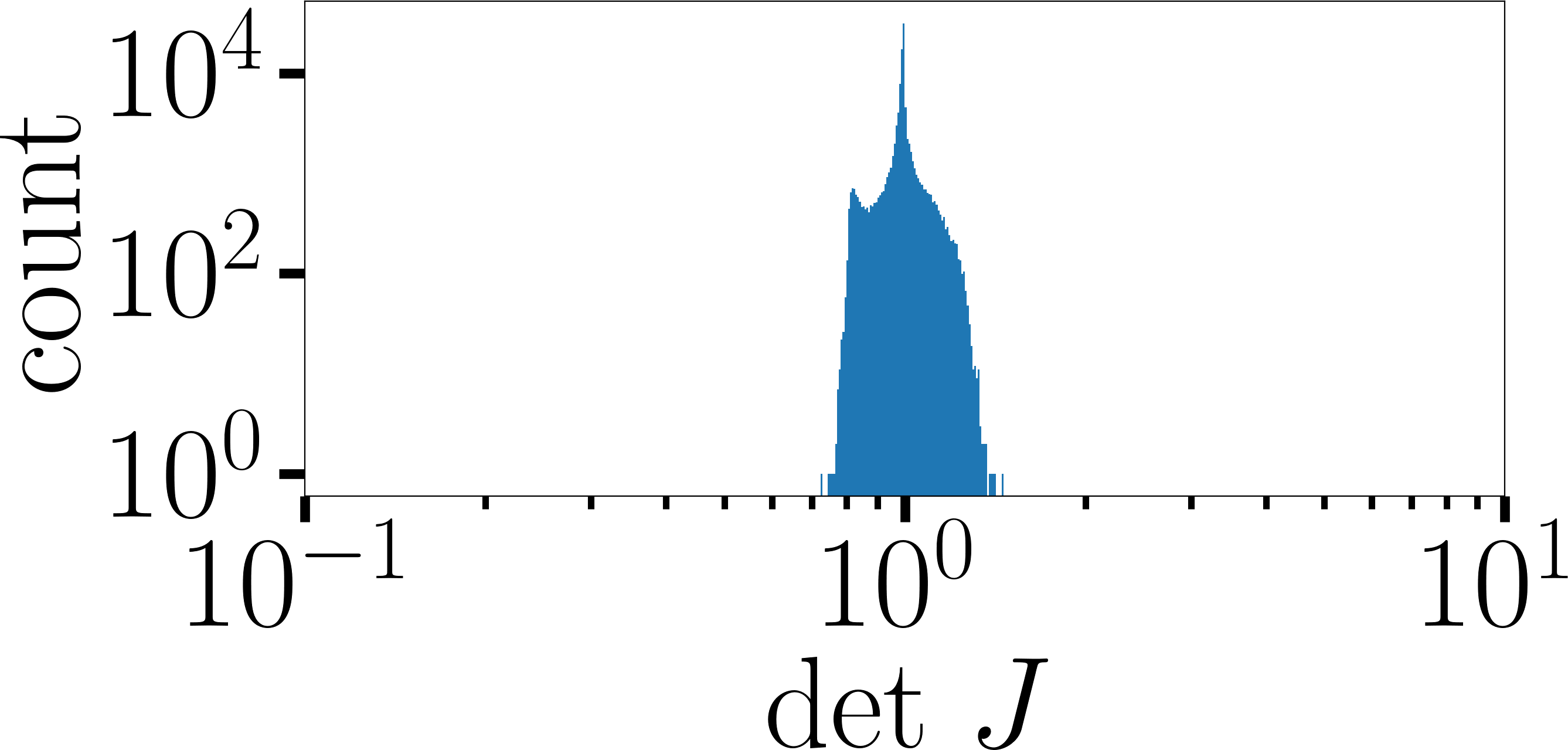}
		\textbf{(b)}
	\end{minipage}
	\caption{Bending test for a tetrahedral mesh of a cylinder, locked vertices are shown in red. \textbf{(a):} IDP~\cite{Fang2021IDP}, \textbf{(b):} QIS deformation with $\theta=\frac12$.
		From top to bottom: Jacobian matrix condition number and the Jacobian determinant are shown in log-log histograms and corresponding color plots.
	}
	\label{fig:idp}
	
	\vspace{10mm}
	
	\begin{minipage}[!t]{.15\linewidth}\vspace{0pt}
		\centering
		\includegraphics[width=\linewidth]{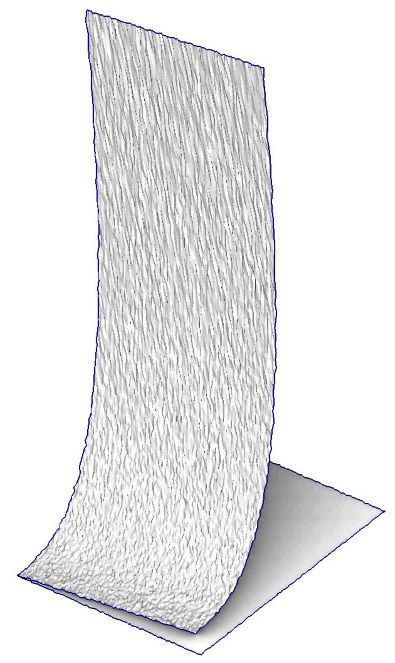}\\
		\textbf{(a)}
	\end{minipage}
	\begin{minipage}[!t]{.41\linewidth}\vspace{0pt}
		\centering
		\includegraphics[width=.18\linewidth]{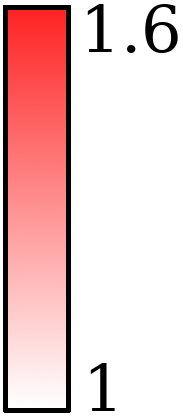}
		\includegraphics[width=.79\linewidth]{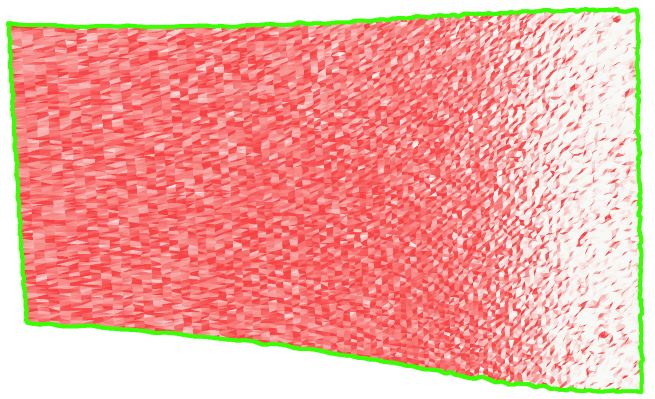}
		\includegraphics[width=\linewidth]{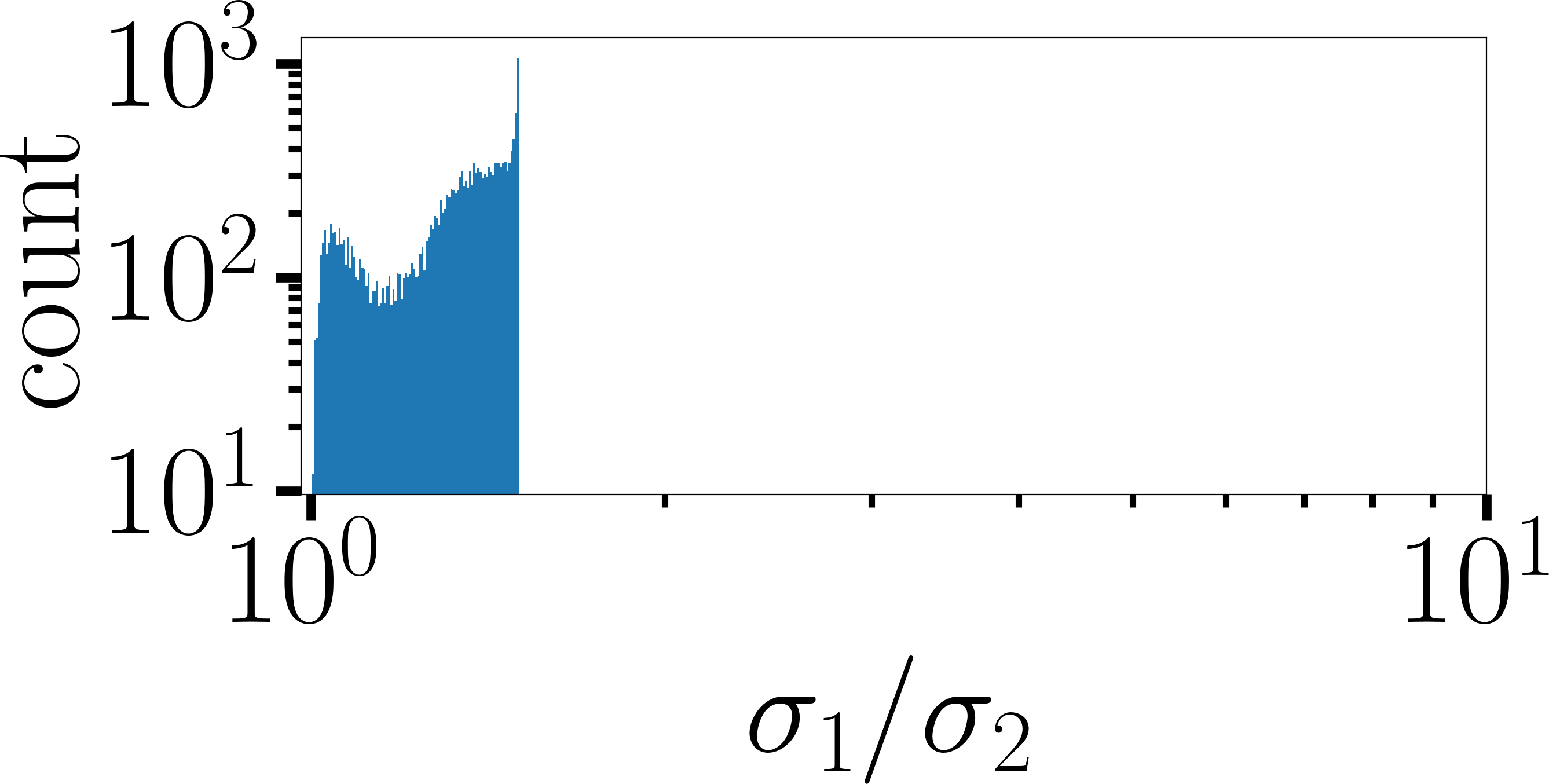}
		\includegraphics[width=\linewidth]{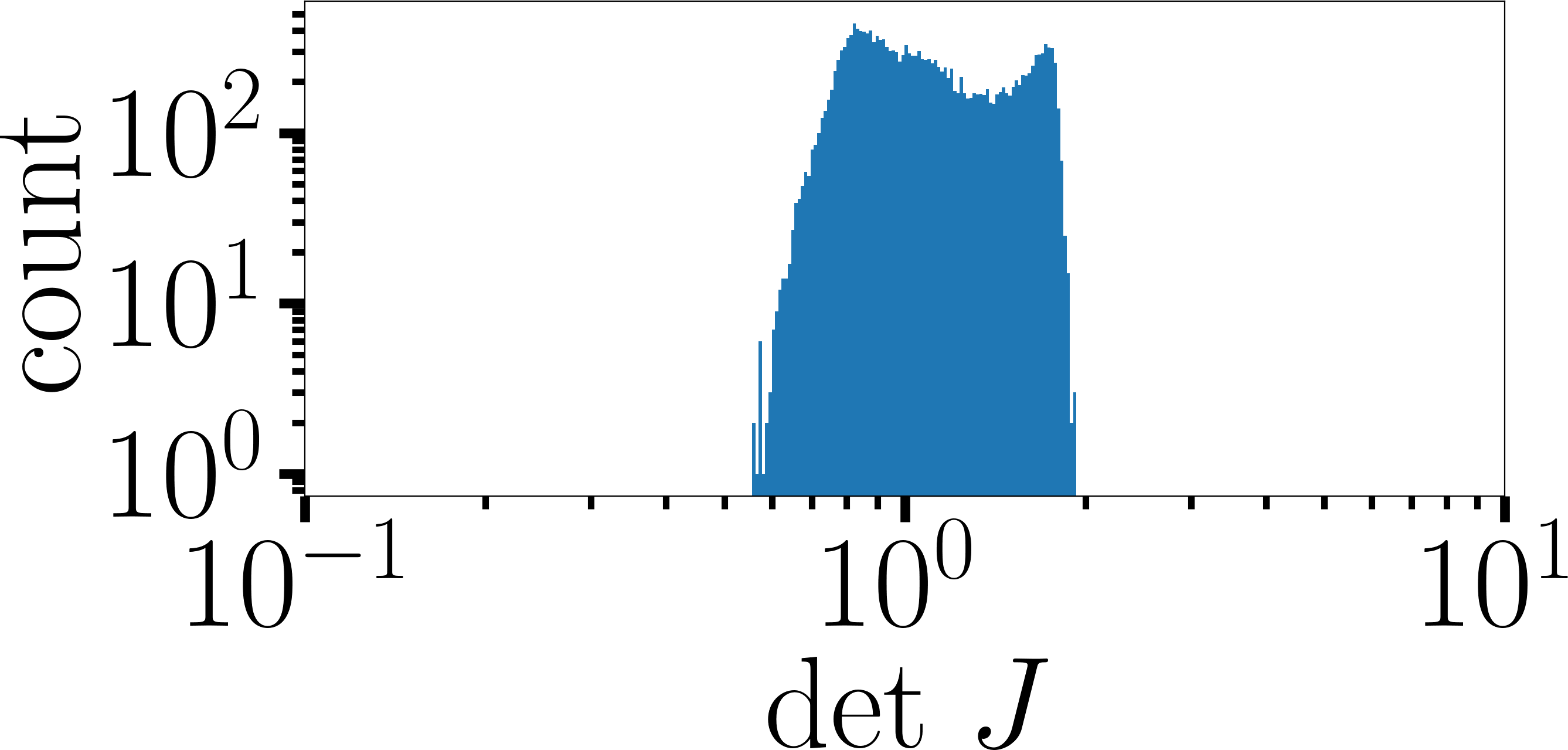}
		\textbf{(b)}
	\end{minipage}
	\begin{minipage}[!t]{.41\linewidth}\vspace{0pt}
		\centering
		\includegraphics[width=\linewidth]{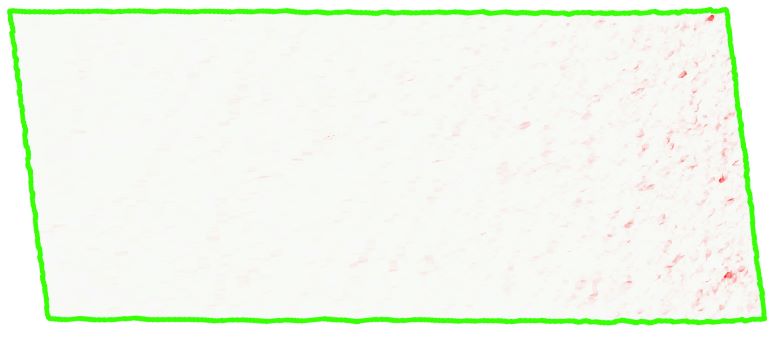}\\
		\vspace{2mm}
		\includegraphics[width=\linewidth]{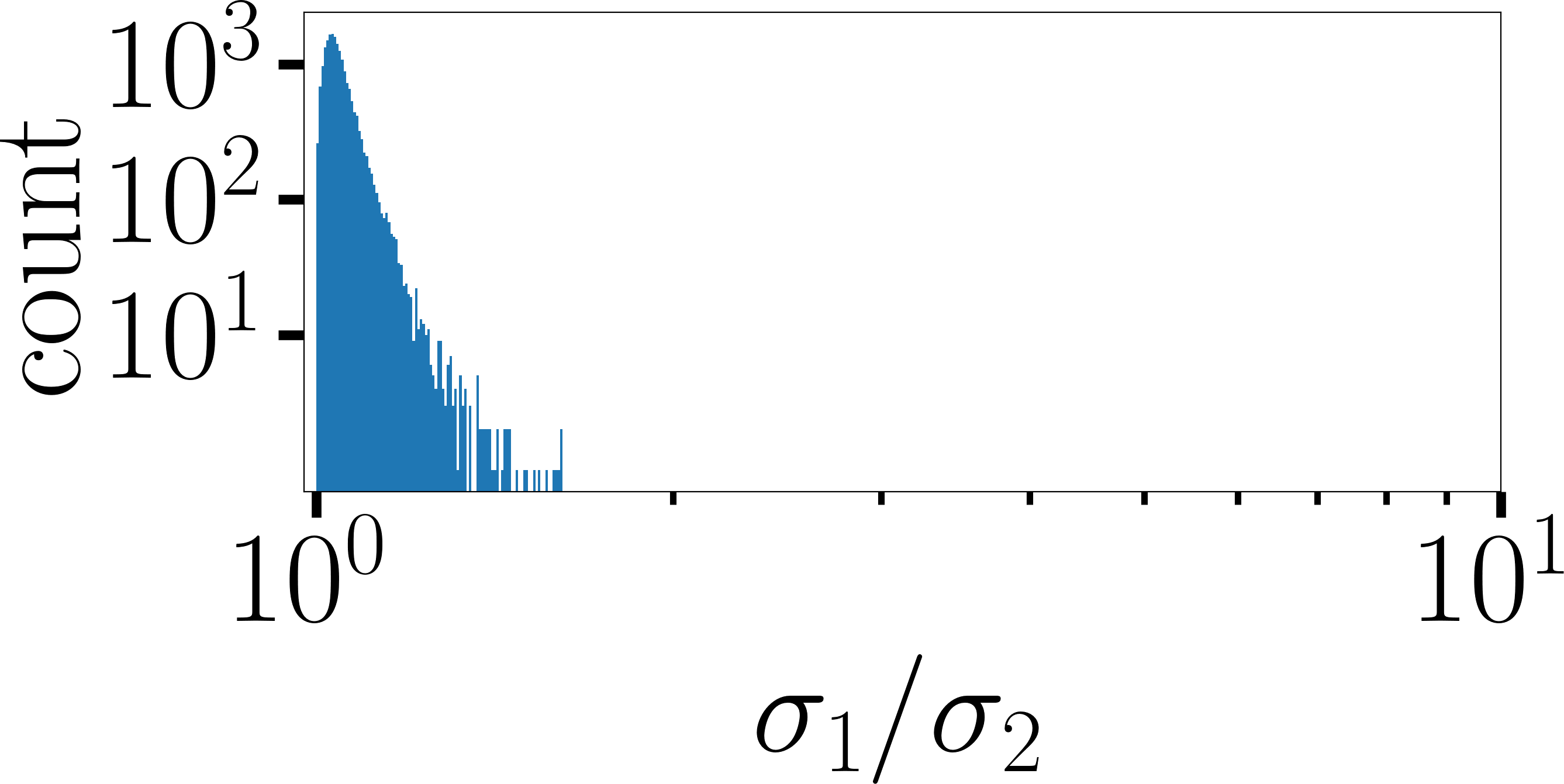}
		\includegraphics[width=\linewidth]{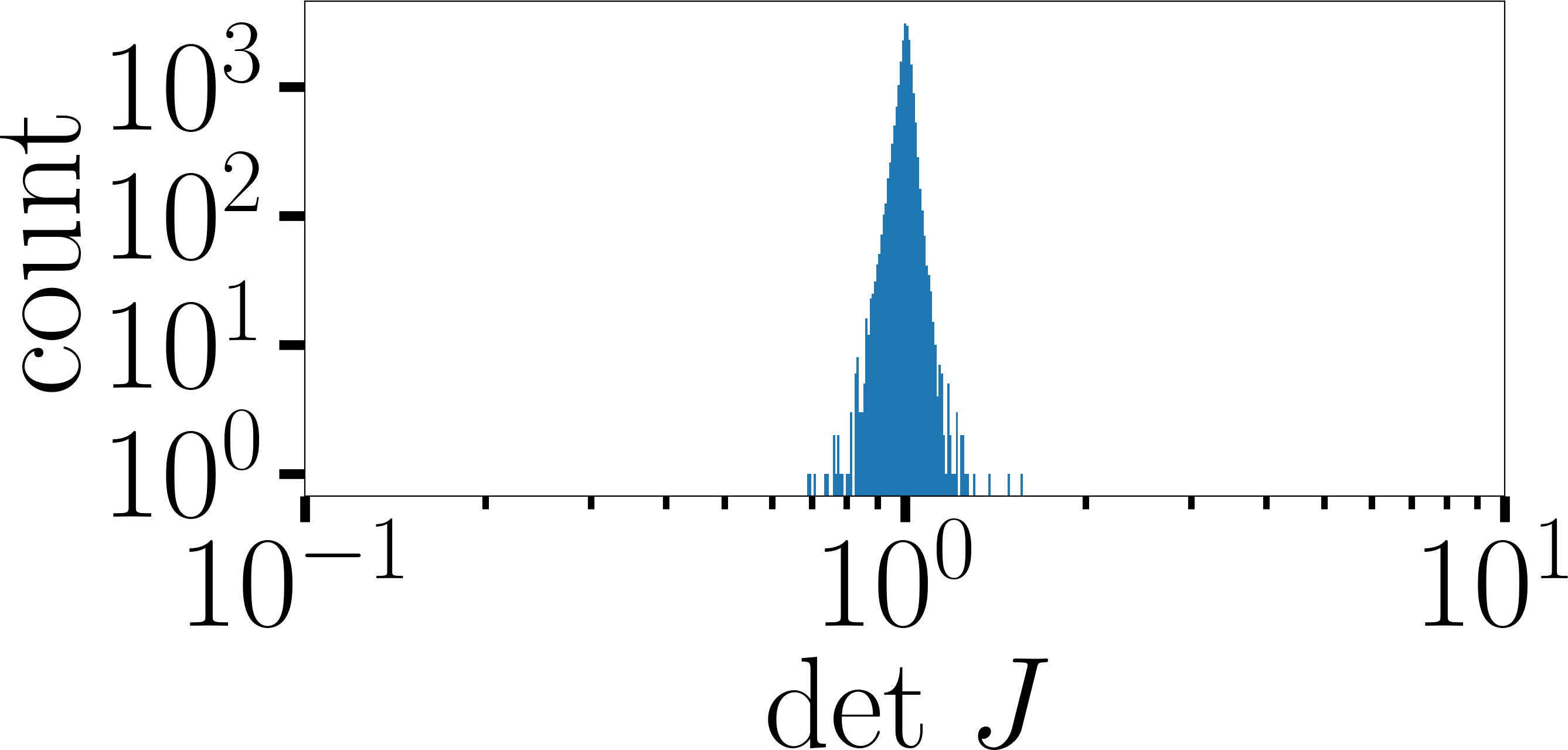}
		\textbf{(c)}
	\end{minipage}
	\caption{
		Comparison of LBD~\cite{LargeScaleBD:2015} vs QIS. \textbf{(a):} 3D surface to flatten is a regular triangular mesh of a square patch that was lifted and noised. \textbf{(b):} The map obtained by LBD. \textbf{(c):} QIS ($\theta=\frac12$).
		\textbf{Top row:} flattenings of \textbf{(a)}, colors correspond to the quality of elements (conditioning of the Jacobian). \textbf{Middle and bottom rows:} log-log element quality histograms.
	}
	\label{fig:LBD}
\end{figure}

\begin{figure}[!p]
	
	\begin{minipage}[!t]{.45\linewidth}\vspace{0pt}
		\centering
		\includegraphics[width=\linewidth]{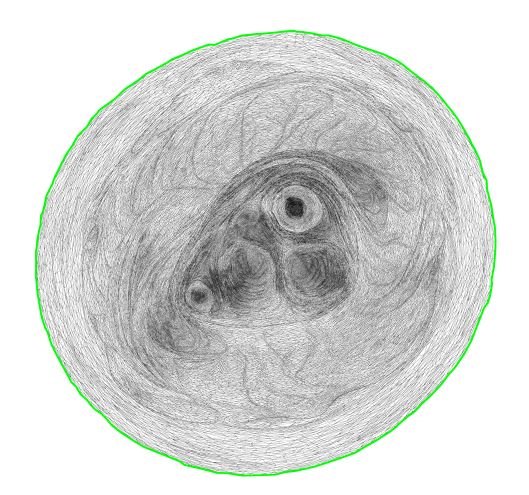}
		\includegraphics[width=.79\linewidth]{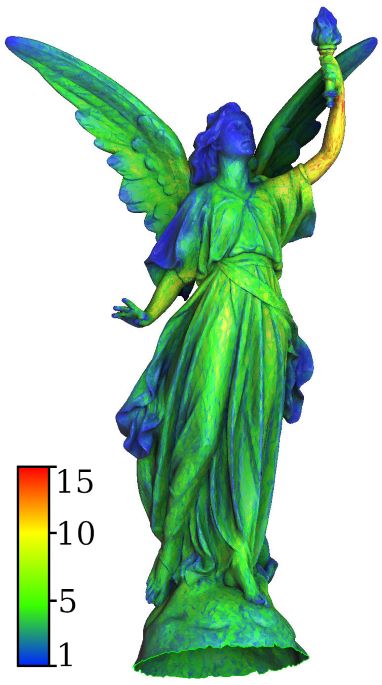}
		\includegraphics[width=\linewidth]{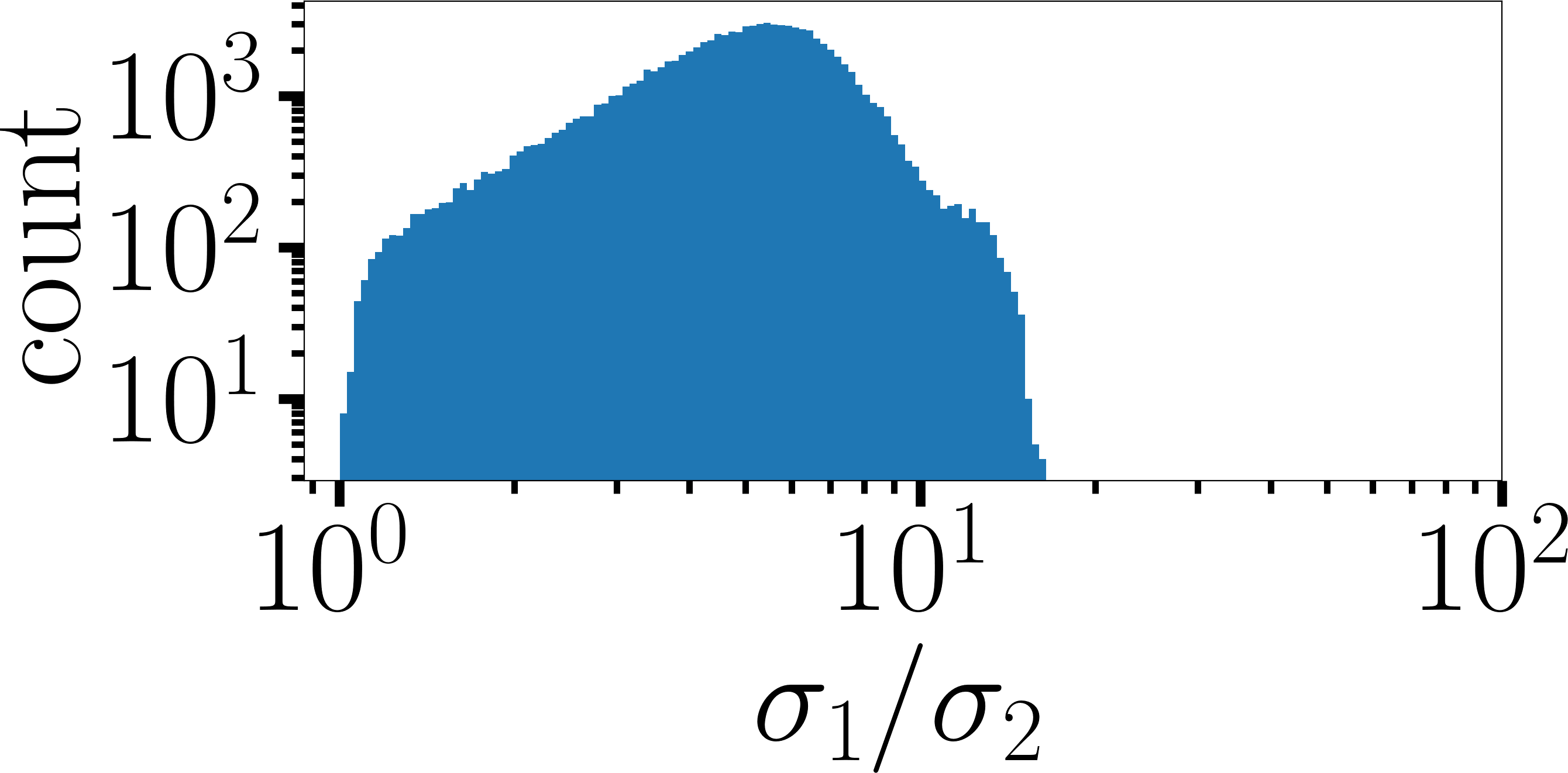}
		\vspace{3mm}
		\includegraphics[width=.79\linewidth]{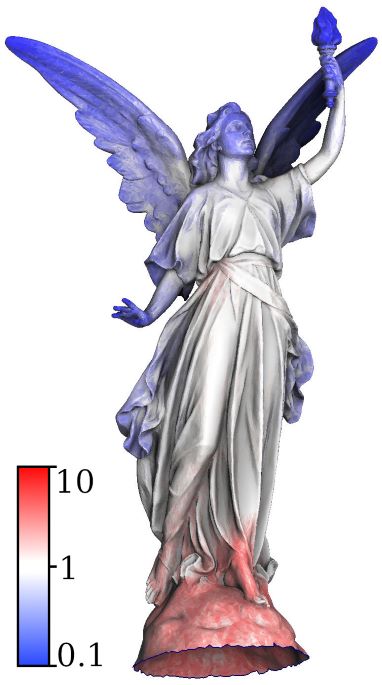}
		\includegraphics[width=\linewidth]{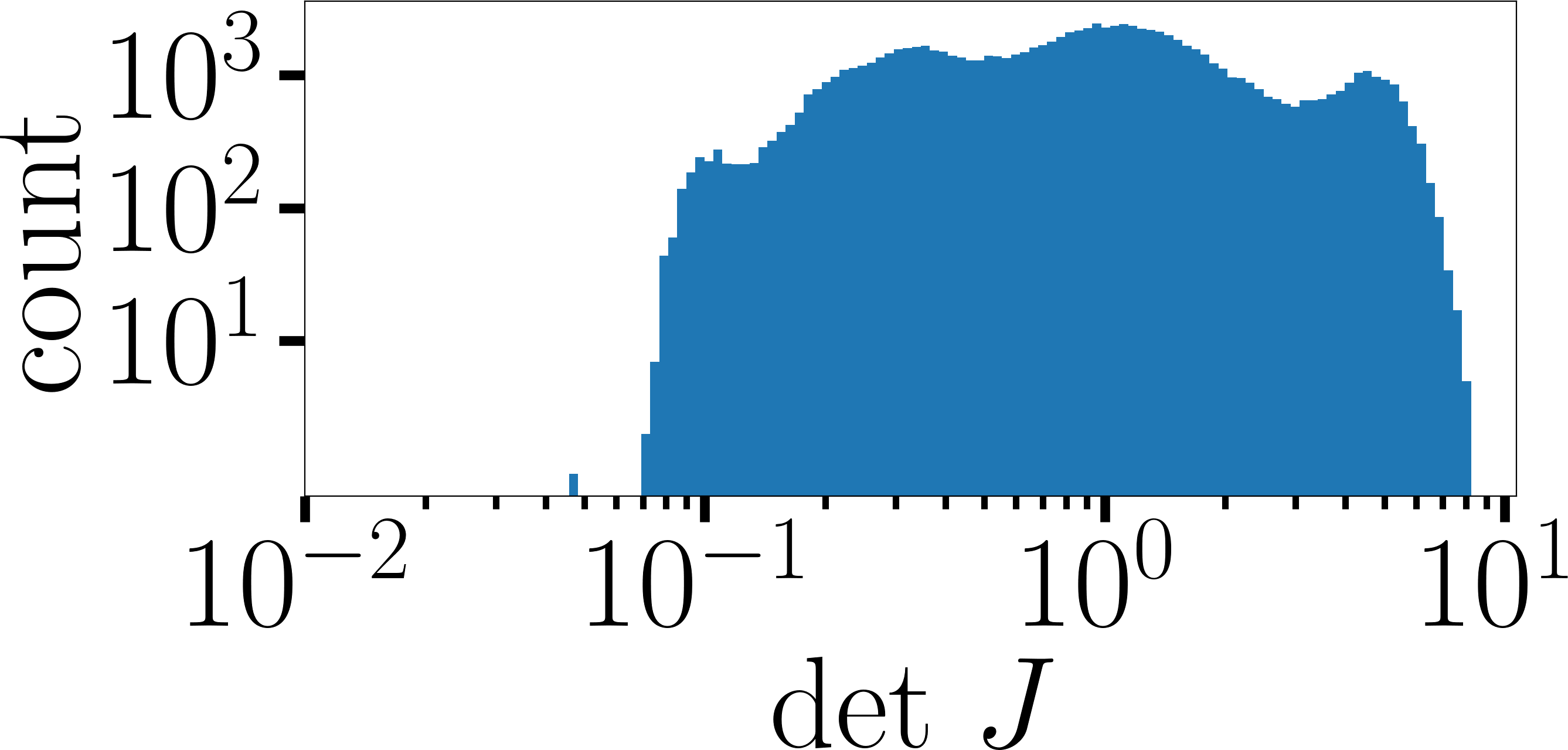}
		\textbf{(a)}
	\end{minipage}
	\begin{minipage}[!t]{.45\linewidth}\vspace{0pt}
		\centering
		\includegraphics[width=.9\linewidth]{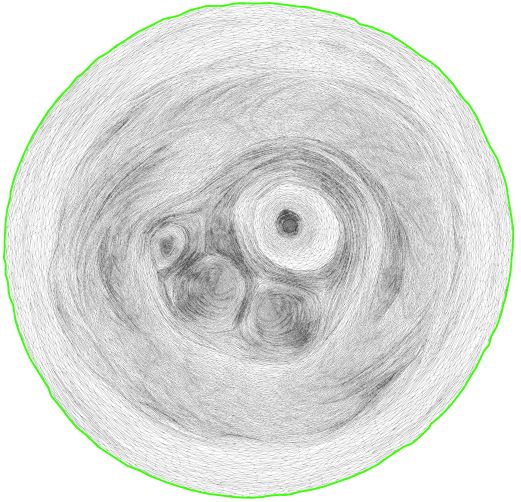}\\
		\vspace{2mm}
		\includegraphics[width=.79\linewidth]{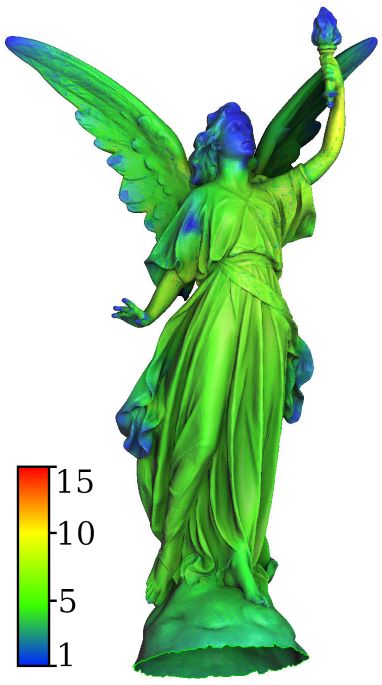}
		\includegraphics[width=\linewidth]{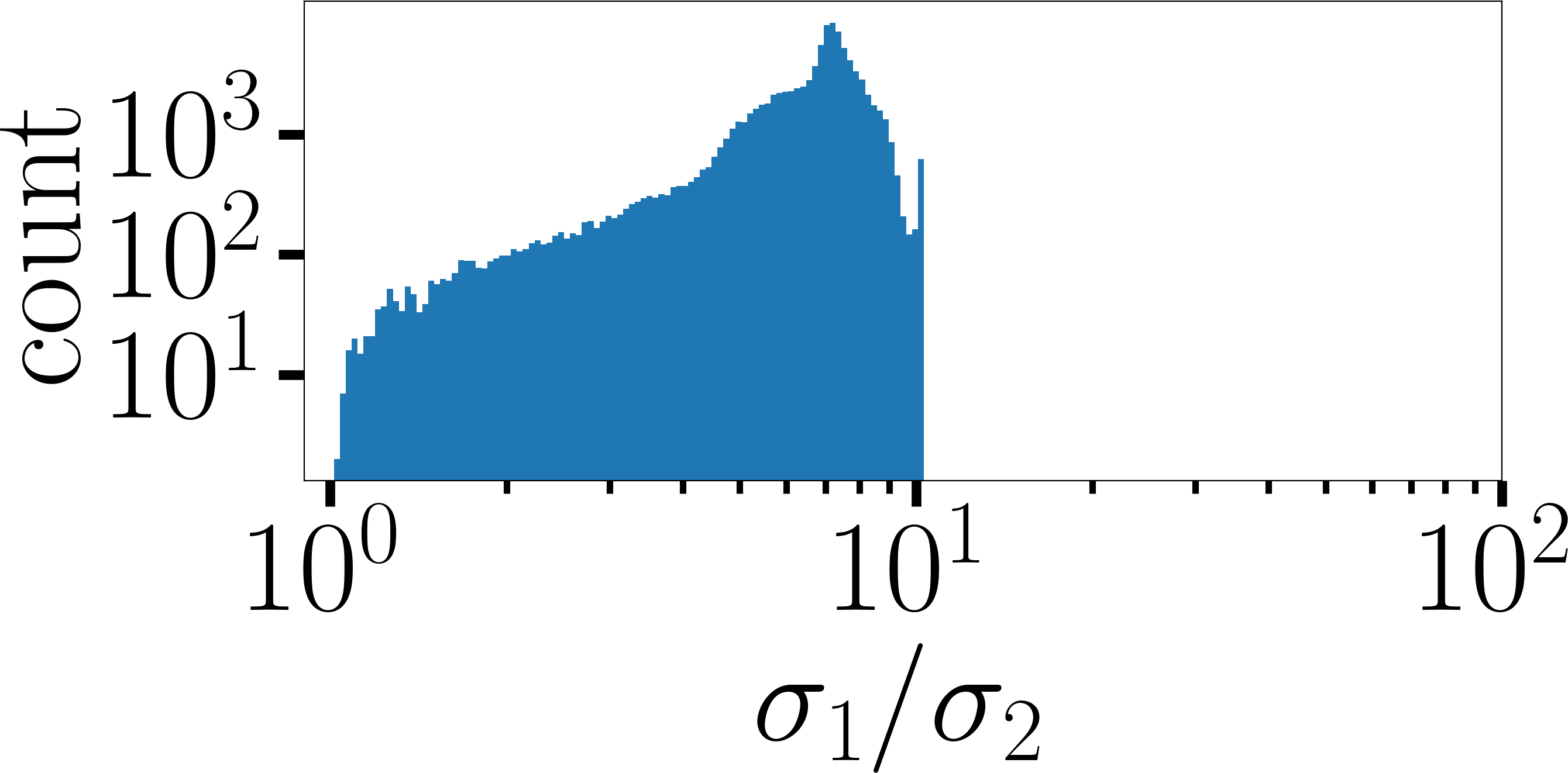}
		\vspace{3mm}
		\includegraphics[width=.79\linewidth]{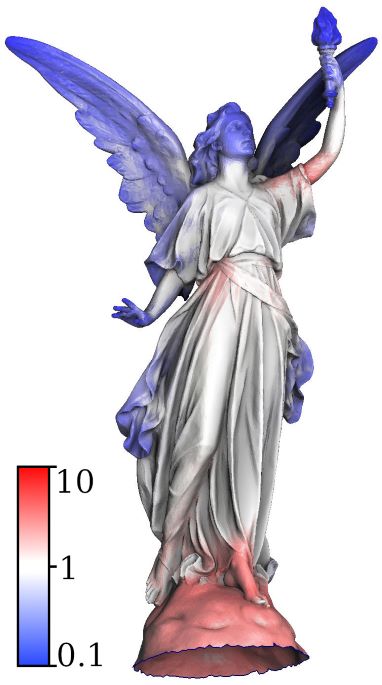}
		\includegraphics[width=\linewidth]{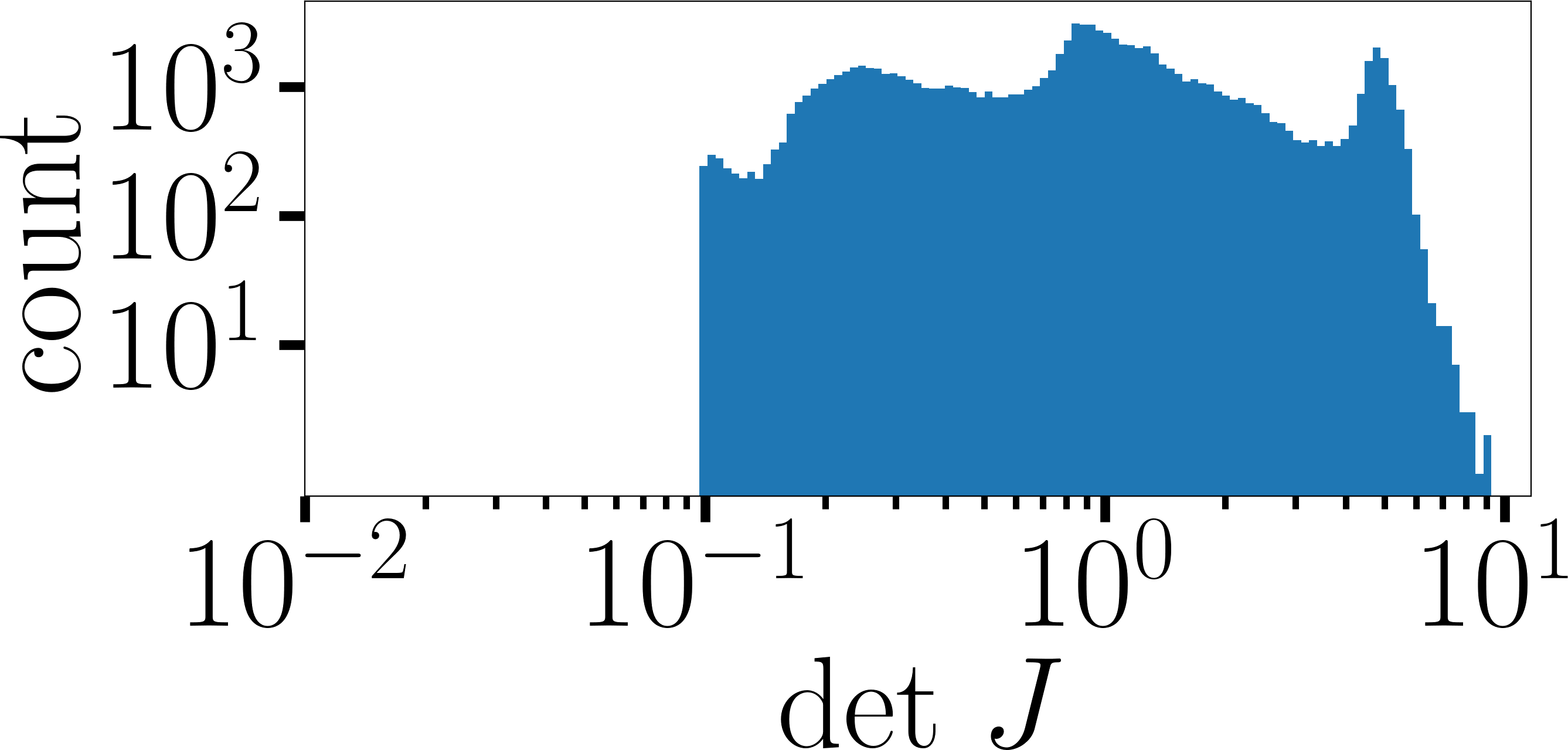}
		\textbf{(b)}
	\end{minipage}
	\caption{Two quasi-isometric maps for the ``Lucy'' mesh. \textbf{(a):} Simplex assembly, \textbf{(b):} QIS, $\theta=\frac12$.
		The quality of elements and its distribution over the surface is shown in log-log histograms and corresponding color plots.}
	\label{fig:lucy_theta}
\end{figure}


Fang et.al~\cite{Fang2021IDP} attempt to improve worst-element distortions by formulating a regularized min-max optimization
for IDP by applying a $p$-norm extension to the symmetric Dirichlet (SD) energy with exponential factor $p>1$.

In our next test we deform a cylinder tetrahedral mesh.
We applied two bone handles (two thin boxes of interior axis vertices) to bend it.
Fig.~\ref{fig:idp} shows the comparison of our results with IDP.
Locked vertices are shown in red.

As advised by Fang et al., we chose $p=5$. It improves slightly the worst-element distortion w.r.t regular IDP, but does not allow to eliminate it completely.
As can be seen in Fig.~\ref{fig:idp}--left, the stress is concentrated around the locked vertices (shown in magenta).
Our optimization ($\theta=\frac12$) allows to dissipate the stress over a larger area, thus improving both distortion measures:
the maximum stretch decreases from 5.05 to 1.94, and the minimum scale increases from 0.36 to 0.72.


\paragraph*{Comparison with Large-scale Bounded Distortion Mappings}
Our next test is LBD~\cite{LargeScaleBD:2015}.
Given an input map (potentially with folds), LBD looks for an injective map as close as possible to the input map, but satisfying some constraints such as the orientation as well as distortion bounds.
Generally speaking, the problem of minimizing an energy subject to bounded distortion constraints is known to be difficult and computationally demanding.
LBD alternates between energy minimization steps and projection to the constraints.

In our test (refer to Fig.~\ref{fig:LBD}), the 3D surface to flatten is a regular simplicial mesh of a rectangular patch that was lifted and noised.
Since LBD has an explicit optimization of the distortion bounds, the comparison is of a particular interest.
In our test the minimum Jacobian determinant increases from 0.30 to 0.69, while the maximum stretch \textit{increases} from 1.50 to 1.61.
Despite the increase in the max stretch, our variational problem leads to a much better overall element quality distribution.
Note that this test was performed with default stretch/scale trade-off parameter $\theta=\frac12$.
If we choose, for example, $\theta=\frac13$, both quality measures improve:
$\max \sigma_1/\sigma_2\approx 1.43$ and $\min \det J \approx 0.76$.


\paragraph*{Comparison with Simplex Assembly}
Our final test is confronting our optimization to Simplex Assembly (SA)~\cite{Fu2016}.
Simplex assembly is a method to compute inversion-free mappings with bounded distortion on simplicial meshes.
The idea is to project each simplex into the inversion-free and distortion-bounded space.
Having disassembled the mesh, the simplices are then assembled by minimizing the mapping distortion, while keeping the mapping feasible.
Fig.~\ref{fig:lucy_theta} provides a quality comparison of SA with our quasi-isometric ($\theta=\frac12$) map for a free-boundary mapping of the ``Lucy'' mesh.
This comparison is interesting for two different reasons: first, SA offers an explicit optimization for the distortion bound, and second, in 2D Fu et. al use exactly the same distortion measure as we do.

Our method shows consistently better maps w.r.t. Simplex Assembly.
In this test, the worst condition number $\max\frac{\sigma_1(J)}{\sigma_2(J)}$ is 16.46 for SA and 10.32 for our method.
The minimum scale $\min\det J$ is equal to 0.05 for SA and 0.10 for our method.

\begin{figure}
	\centering
	\begin{minipage}[!t]{.45\linewidth}\vspace{0pt}
		\centering
		\includegraphics[width=.9\linewidth]{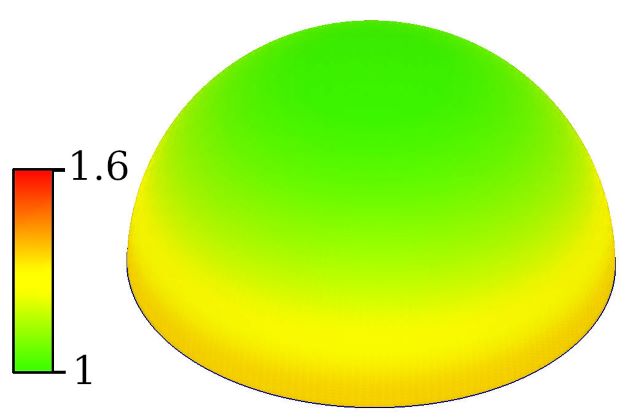}
		\includegraphics[width=.9\linewidth]{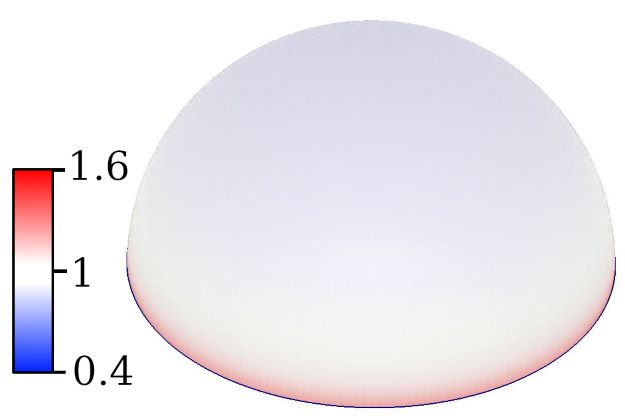}
		\textbf{(a)}
	\end{minipage}
	\begin{minipage}[!t]{.45\linewidth}\vspace{0pt}
		\centering
		\includegraphics[width=.9\linewidth]{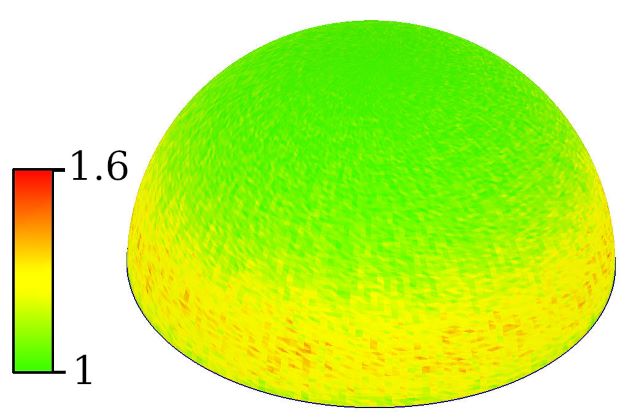}
		\includegraphics[width=.9\linewidth]{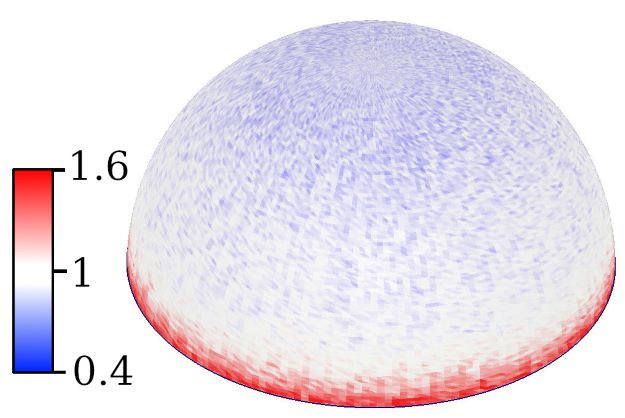}
		\textbf{(b)}
	\end{minipage}
	\caption{Mapping a half-sphere on a disc: our solution with $\theta=\frac12$ \textbf{(a)} vs Simplex Assembly \textbf{(b)}.
		\textbf{Top row:} condition number of the Jacobian matrix, \textbf{bottom row:} Jacobian determinant.
	}
	\label{fig:half-sphere}
\end{figure}

Note that while in 2D Simplex Assembly has the same objective function as our method, the way SA poses the problem (minimization of maximal distortion) leads to non-smooth solutions.
The fact that the solution is noisy can already be seen in Fig.~\ref{fig:lucy_theta},
but we chose to perform a second test to highlight the fact: we map a half-sphere onto a disc (Fig.~\ref{fig:half-sphere}).


This test case is interesting because it has a closed form solution.
It is possible to build an analytical flattening which has the smallest known quasi-isometry constant $\Gamma = \sqrt{\pi / 2}$.
This mapping can be obtained by isometric projection of meridians onto straight segments on the plane starting from the north pole while keeping angular projection uniform.
Obviously, singular values of this mapping range from $1$ to $\pi / 2$, and using best scaling we get $\Gamma = \sqrt{\pi / 2} \approx 1.253$.

As before, our result is better:
the distortion for our flattening is equal to $\sqrt{\frac{\max_i\sigma_1(J_i)}{\min_i\sigma_2(J_i)}} = 1.279$, which is within 2\% from the ideal bound,
and the quasi-isometry constant for SA is equal to 1.47.
This difference is due to the optimization scheme choice. Since we discretize a well-posed variational problem, our method provides smooth solutions, whereas SA result is noisy and loses angular symmetry.


\begin{figure*}
	\centerline{
		\begin{minipage}[!t]{.22\linewidth}\vspace{0pt}
			\centering
			\includegraphics[width=\linewidth]{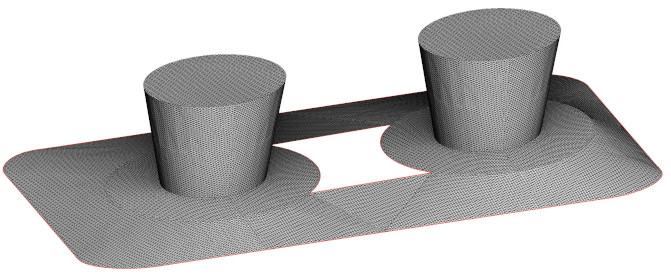}
			\includegraphics[width=\linewidth]{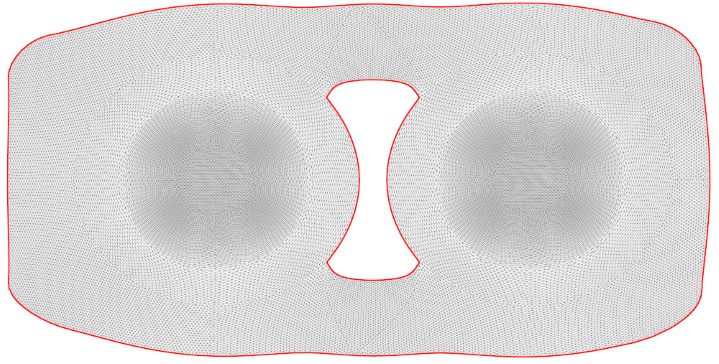}
			\includegraphics[width=\linewidth]{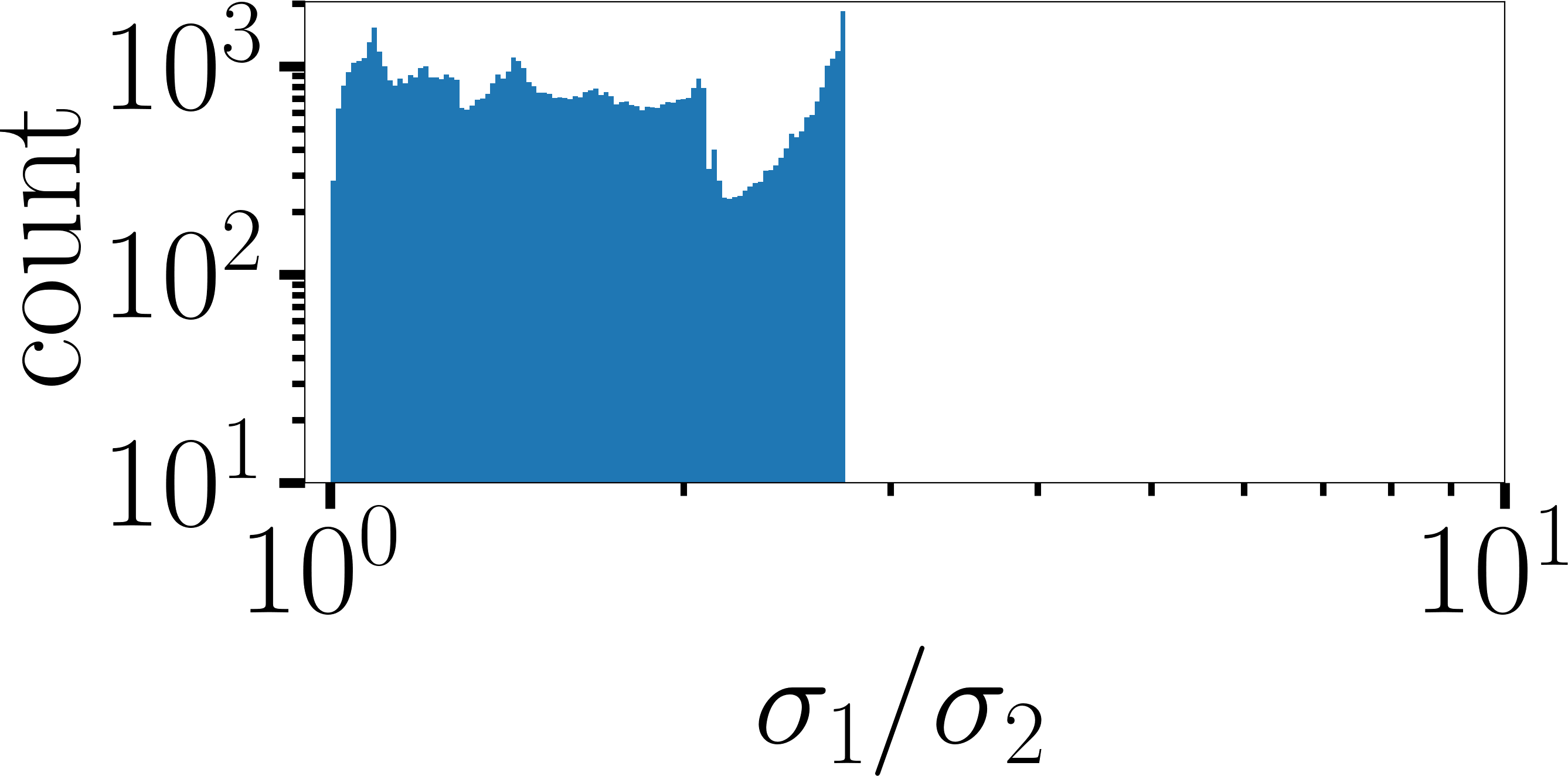}
			\includegraphics[width=\linewidth]{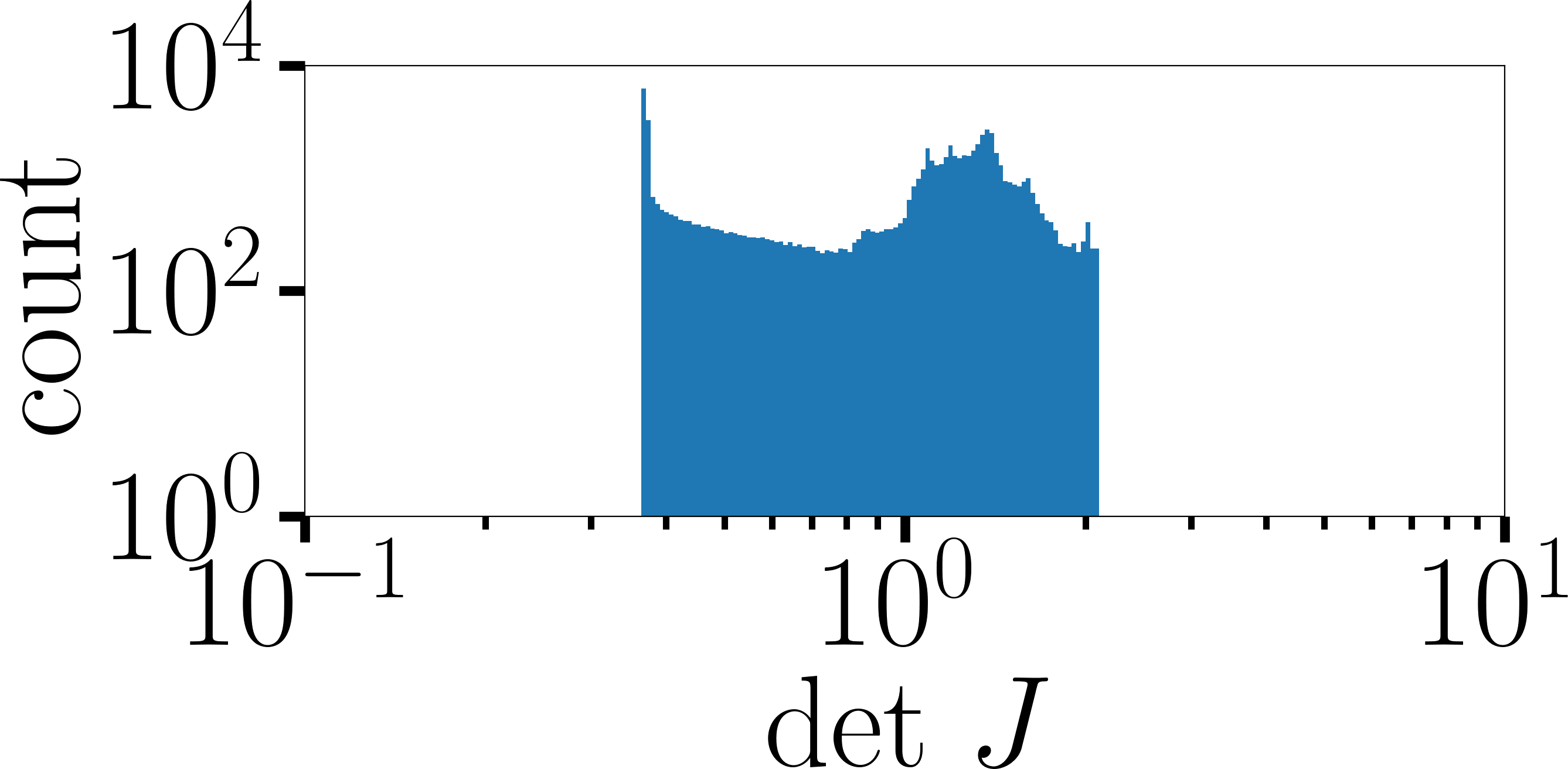}
			\textbf{(a)}
		\end{minipage}
		\begin{minipage}[!t]{.22\linewidth}\vspace{0pt}
			\centering
			\includegraphics[width=\linewidth]{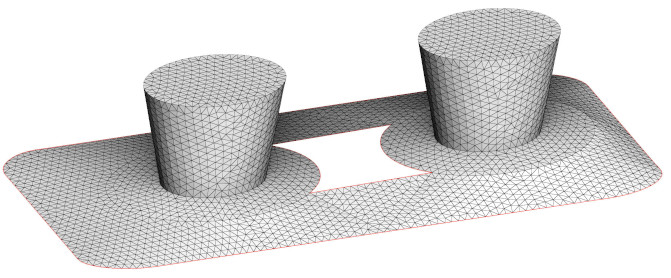}
			\includegraphics[width=\linewidth]{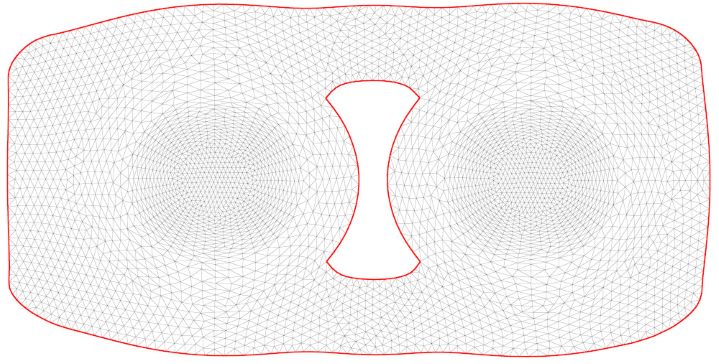}
			\includegraphics[width=\linewidth]{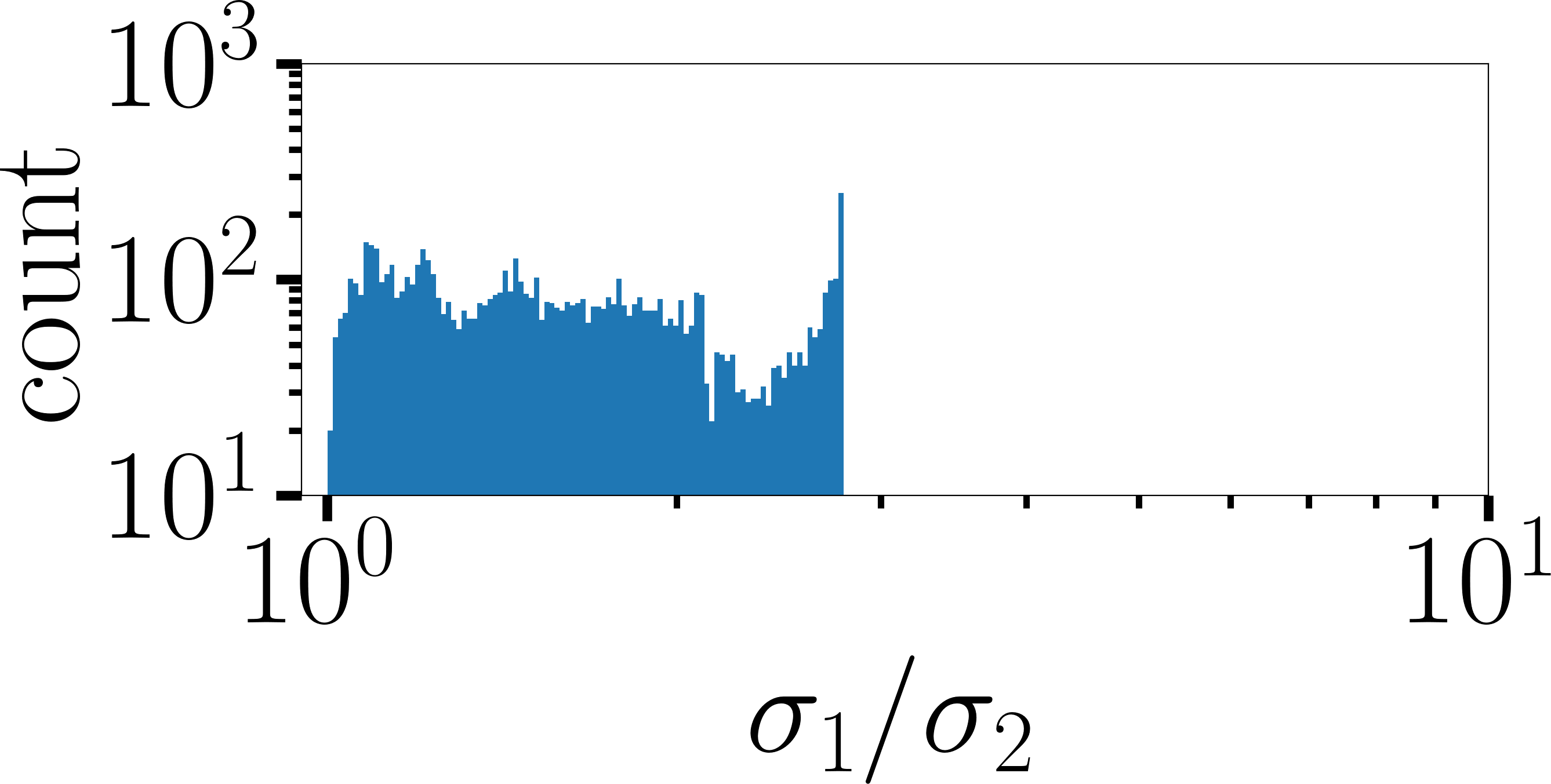}
			\includegraphics[width=\linewidth]{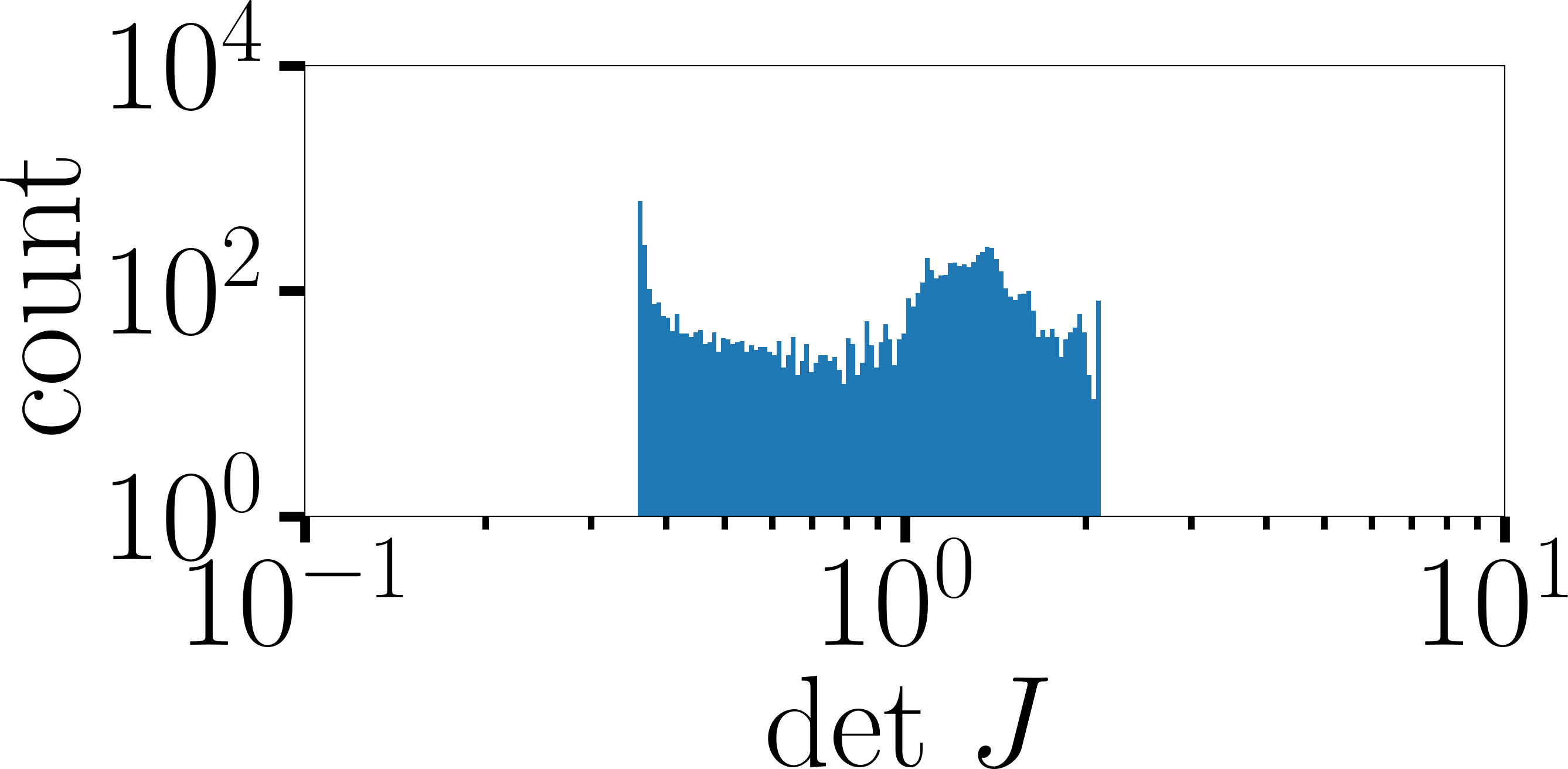}
			\textbf{(b)}
		\end{minipage}
		\begin{minipage}[!t]{.22\linewidth}\vspace{0pt}
			\centering
			\includegraphics[width=\linewidth]{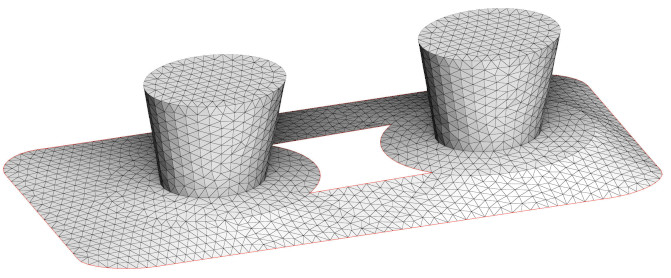}
			\includegraphics[width=\linewidth]{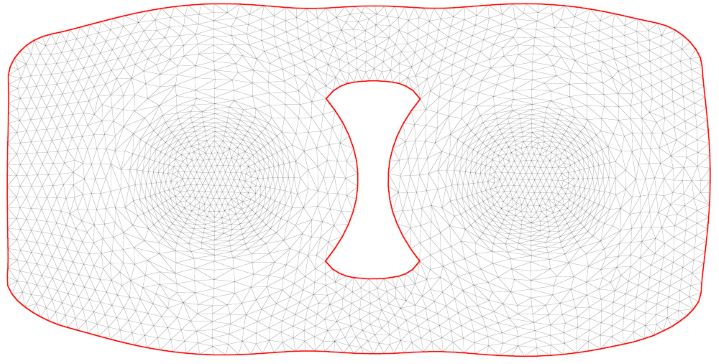}
			\includegraphics[width=\linewidth]{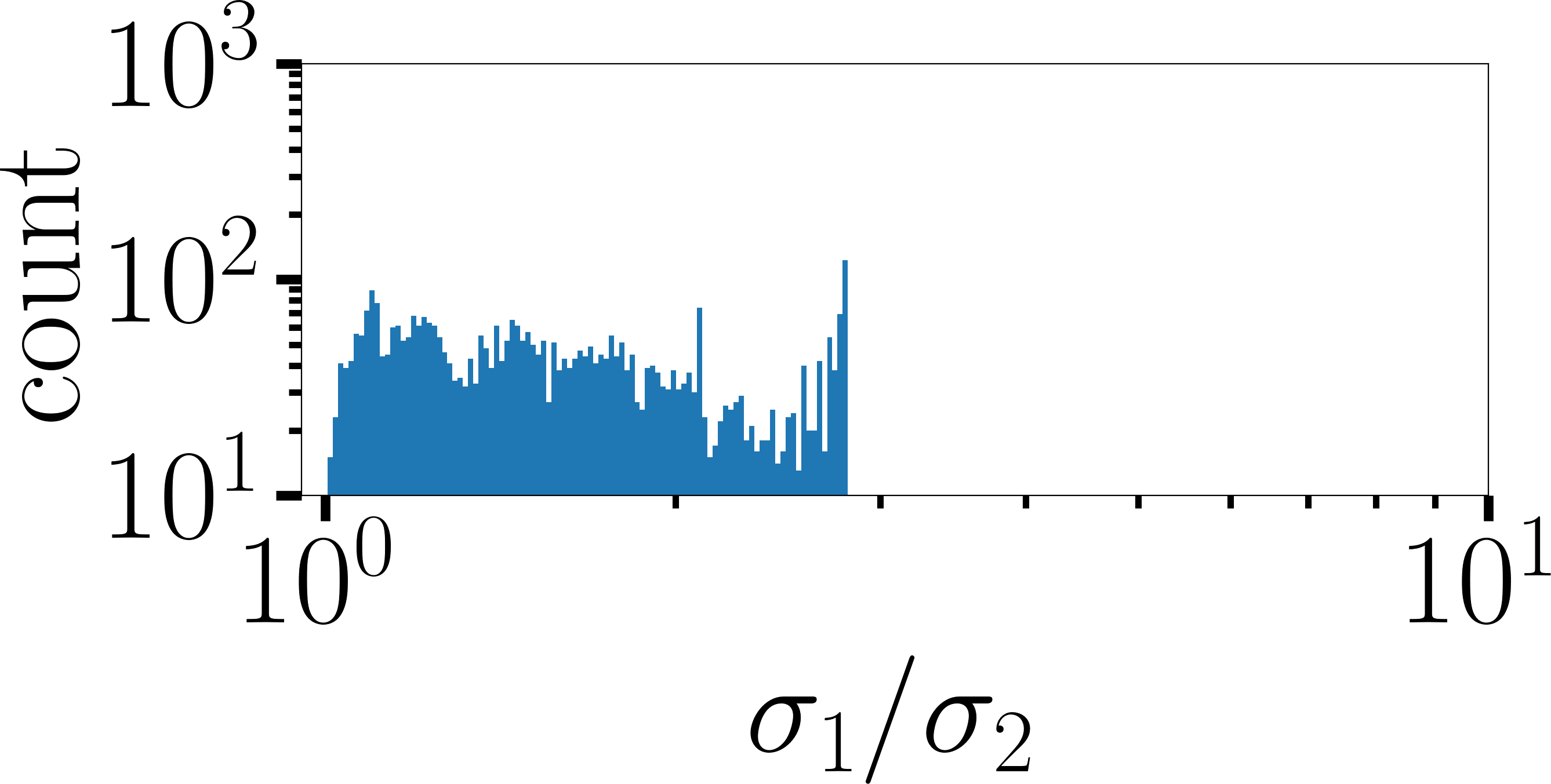}
			\includegraphics[width=\linewidth]{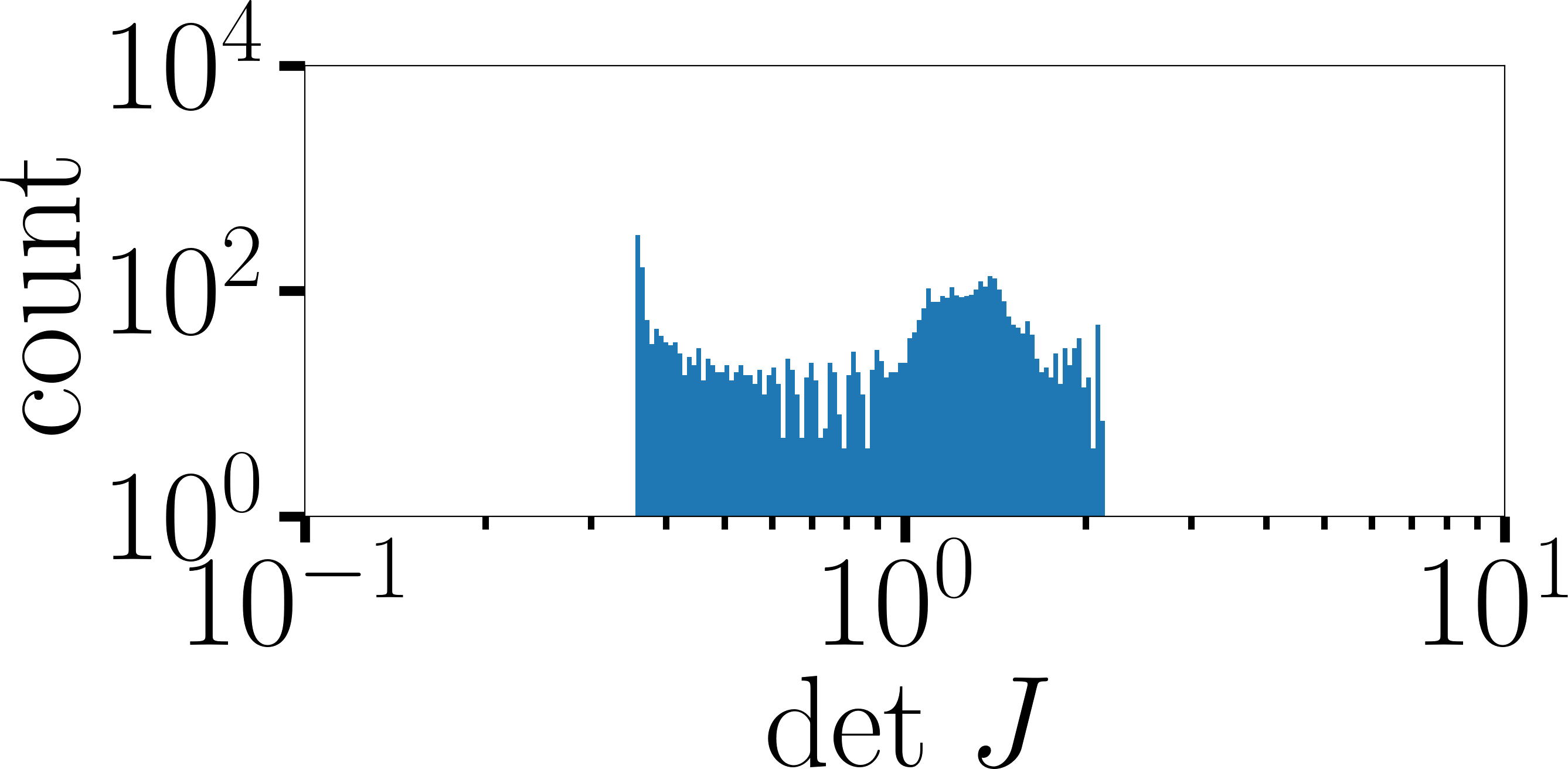}
			\textbf{(c)}
		\end{minipage}
		\begin{minipage}[!t]{.22\linewidth}\vspace{0pt}
			\centering
			\includegraphics[width=\linewidth]{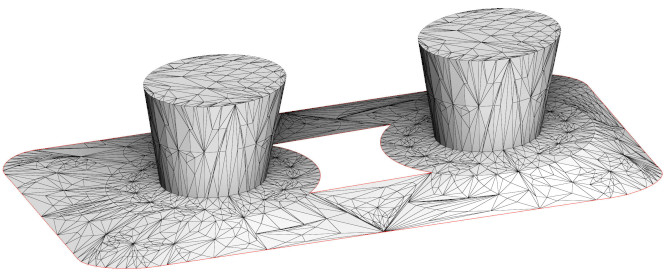}
			\includegraphics[width=\linewidth]{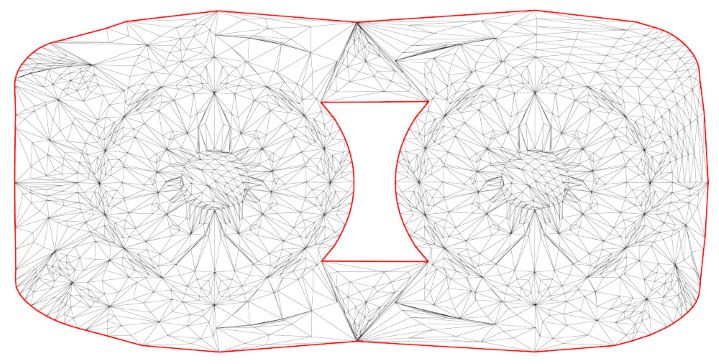}
			\includegraphics[width=\linewidth]{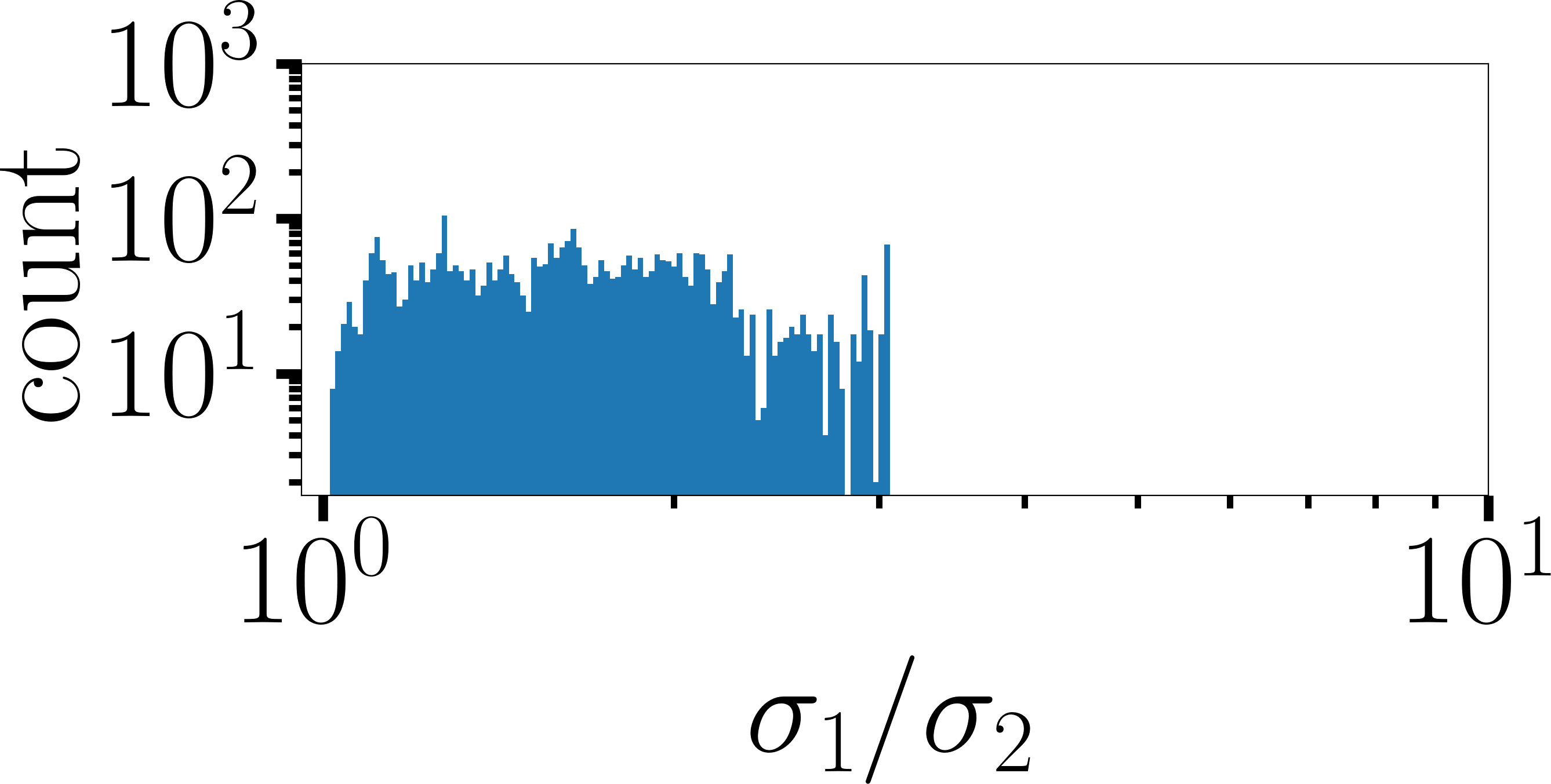}
			\includegraphics[width=\linewidth]{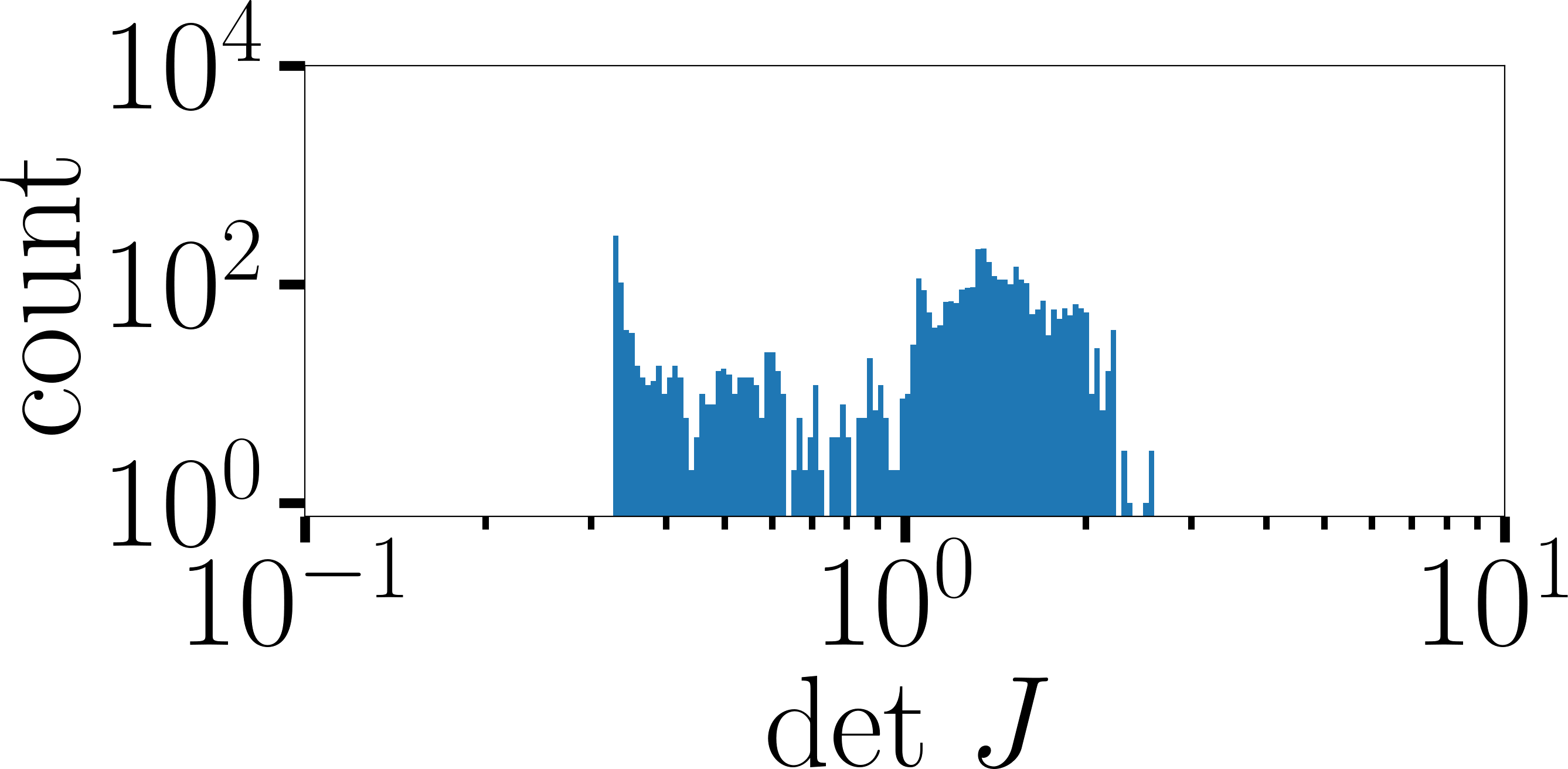}
			\textbf{(d)}
		\end{minipage}
	}
	\caption{Our method shows an excellent stability irrespective of input mesh quality. Here we show four free-boundary quasi-isometric maps ($\theta=\frac12$) computed on different meshes.
		\textbf{(a)--(c):} Behavior under coarsening (from 73k to 4k triangles) of an isotropic mesh. \textbf{(d):} Stress test, the input mesh has a very bad quality with the maximum aspect ratio of $10^8$ over the triangulation. Despite that, the resulting map is still of a good quality.}
	\label{fig:stability}
\end{figure*}

\paragraph*{Stability}
As we have already mentioned, our approach is a discretization of a well-posed variational scheme, and it has an advantage that type, size and quality of mesh elements in the deformed object have a weak influence on the computed deformation.
Here we show that attainable quality threshold estimates (quasi-isometry constants) do not deteriorate with mesh coarsening which is unique property of the proposed algorithm.
To illustrate this point, we have computed four free-boundary quasi-isometric maps on different meshes of the same object (Fig.~\ref{fig:stability}).
Under coarsening of an input isotropic mesh, the distortion bound remains almost the same (2.07, 2.10, 2.11 for the maximum jacobian condition number, respectively).
We have also perfomed a stress test Fig.~\ref{fig:stability}--(d) with maximum aspect ratio of $10^8$ over the input triangulation,
and the maximum jacobian condition number we obtained equals to 2.25.


\subsection{Untangling global parameterizations}
An important application of parameterizations is the generation of quad meshes.
Given a 3D surface and its 2D flattening, applying the inverse of the map to a 2D grid generates a grid on the surface, i.e. a quad mesh.
Naturally, this technique requires the map to be locally invertibile, hence the importance of the untangling approach.
For producing more complex quad meshes~\cite{MIQ2009}, it is also possible to introduce discontinuities in the map.
To generate a valid quad mesh, however, we need to impose constraints along these discontinuities: the transition function that maps one side of a cut to the other side must be grid-preserving, i.e. it transforms the 2D unit grid onto itself (see Fig.~\ref{fig:transition}).

All such grid preserving transition functions can be decomposed into a rotation of $k\pi/2$ plus an integer translation.
In practice, producing such a global parameterization requires two parameterization steps: one where only the rotation is known (from a frame field e.g. \cite{Desobry2021}) and one where the translation is also known (after quantization e.g. \cite{IGM2013}).
In both cases, these boundary constraints are affine and can be easily introduced in our optimization scheme using \cite{Bommes:2010:PMO}.
Figure~\ref{fig:pipeline} illustrates the global parameterization pipeline.

\begin{figure}
	\centering
	\includegraphics[width=.6\linewidth]{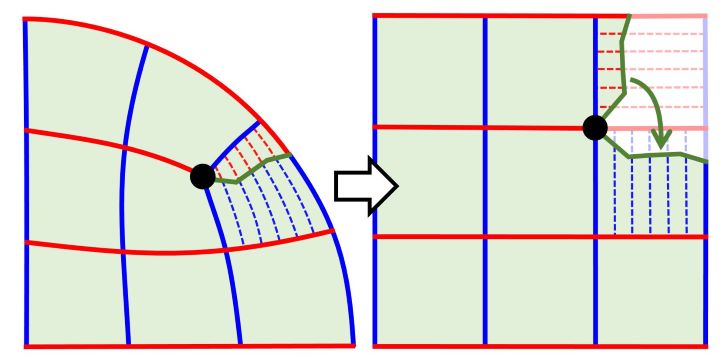}
	\caption{The map is discontinuous across the red cut, but the projection of the unit grid cells from the map (right) to the object (left) coincide thanks to the grid preserving transition function (green arrow). }
	\label{fig:transition}
\end{figure}

\begin{figure*}
	\centering
	\includegraphics[width=\linewidth]{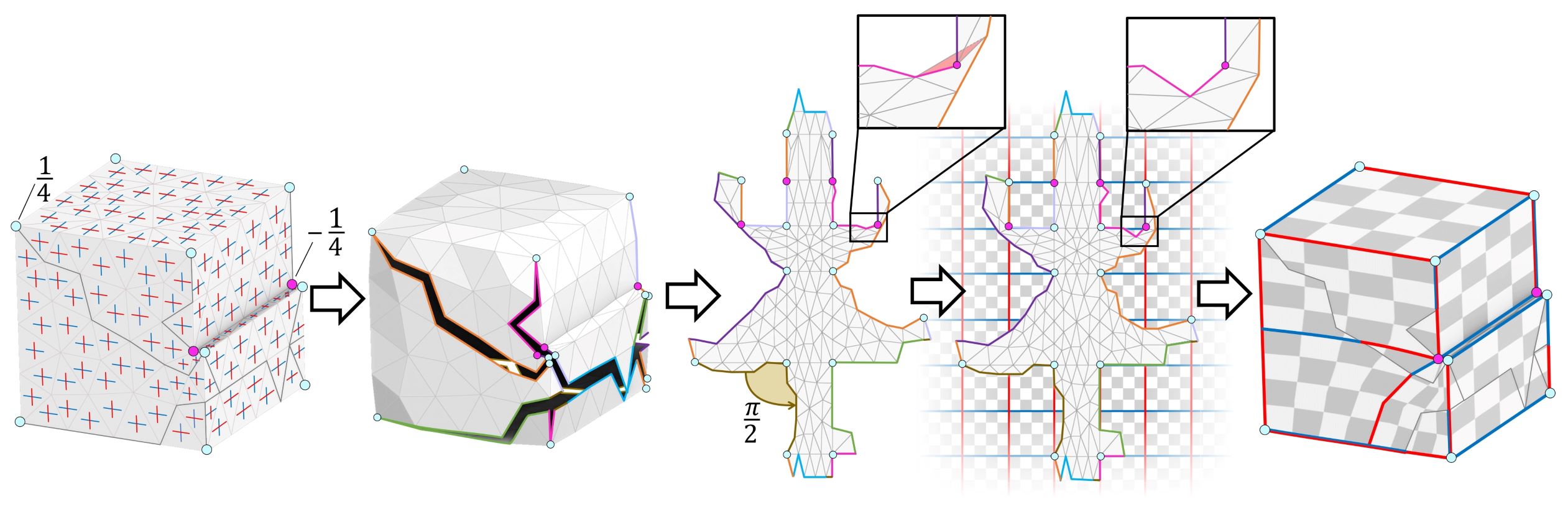}
	\textbf{(a)}\hspace{32mm}
	\textbf{(b)}\hspace{32mm}
	\textbf{(c)}\hspace{32mm}
	\textbf{(d)}\hspace{32mm}
	\textbf{(e)}
	\caption{Quad generation via global parameterization pipeline.
		\textbf{(a):} we compute a frame field over the input triangulation, it allows us to determine singular vertices.
		Then  the mesh is cut open~\textbf{(b)} and flattened under grid preserving constraints along the cuts.
		Least squares solution~\textbf{(c)} has inverted elements that we need to untangle~\textbf{(d)}.
		Finally, the unit grid is projected back to the mesh to define the quad mesh~\textbf{(e)}.
	}
	\label{fig:pipeline}
\end{figure*}

Recall that being free of inverted elements does not imply local invertibility~\cite[\S 3.3]{garanzha2021foldoverfree}.
In difficult cases, double coverings may appear. It is possible to prevent this with a brute force solution~\cite{garanzha2021b}.
Local injectivity can be enforced by adding extra (virtual) triangles. Unfortunately, this approach rigidifies the mesh and can be time consuming.
We can do better for global parameterizations.

\begin{figure}
	\centering
	\includegraphics[width=\linewidth]{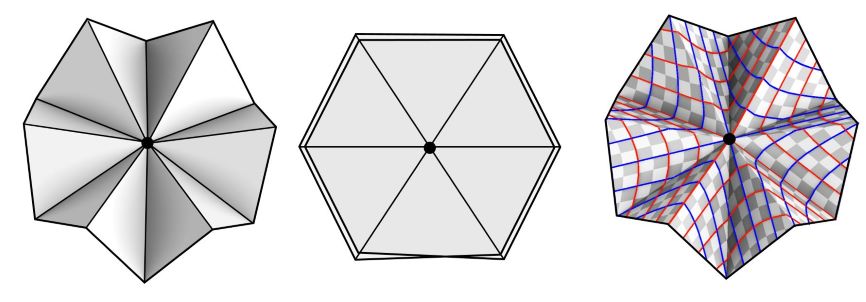}
	\caption{A double covering example. \textbf{Left:} a surface to flatten is made of 12 equilateral triangles. \textbf{Middle:} the surface can be mapped to a plane without inverted elements, producing a double covering.
		The map is not invertible in one point.
		\textbf{Right:} gradients of the parameterization form two orthogonal vector fields (shown in red and blue), both of them have a -1 singularity at the center.
	}
	\label{fig:covering}
\end{figure}

We can interpret the gradients of the parameterization as a frame field \cite{NSDF2008}.
Double coverings arise from index -1 singularities of this field (refer to Fig.~\ref{fig:covering}).
Poincaré-Hopf theorem states that the sum of indices is constant, so an index $1$ singularity must be placed somewhere to compensate for the index $-1$.
With free boundary mapping (as in our example in Fig.~\ref{fig:covering}), the index $1$ is placed outside of the domain, simply adding a loop to the boundary.
For the global parameterization case, there is no free boundary, so the index $1$ must appear at a vertex.
Recall however, that our maps have positive Jacobian determinant over all elements, and thus it is impossible to place singularity 1 at a regular vertex, since it is a pole singularity that would force the map to have degenerate elements.
Therefore, the only possibility for the solver is to place it on a vertex that already has a negative index singularity (typically $-1/4$) as illustrated in Fig.~\ref{fig:promotion}.
\begin{figure}
	\centering
	\includegraphics[width=.9\linewidth]{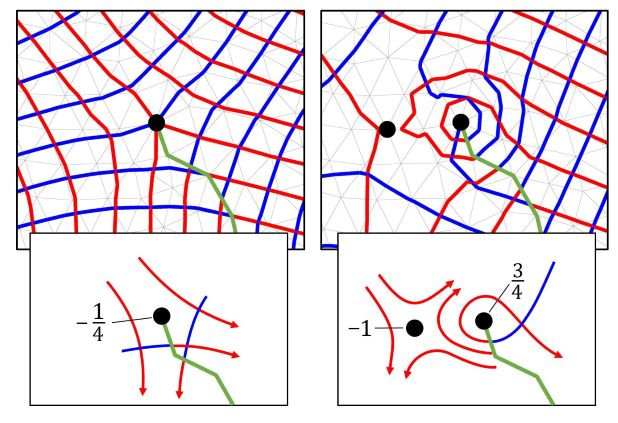}
	\caption{
		Quad mesh generation via global parameterization approach. The triangular domain is mapped to a plane, the green line corresponds to a (grid preserving) discontinuity in the map.
		Red and blue lines correspond to the flat unit grid under the action of the inverse map.
		\textbf{Left:} one singularity of index $-1/4$ is present in the domain, a valid quad mesh is generated.
		\textbf{Right:} a double covering ($-1$ singularity) can appear if and only if the $-1/4$ singularity is ``promoted'' to $3/4$, thus total sum is still equal to -1/4.
		The map is no longer invertbile (even if all elements have positive Jacobian!), leading to problems in quad mesh generation.
	}
	\label{fig:promotion}
\end{figure}

This observation allows to avoid double coverings by simply forcing vertices with negative index singularity to preserve the index.
To this end, for each such vertex $v$, we flatten its one ring, and compute the rotation + scale (with respect to $v$) that send each adjacent vertex to the next one.
The rotation is then scaled according to the index of the vertex ($\times 5/4$ for the index $-1/4$), and these affine equations tying in adjacent vertices are introduced as constraints to our system.
This solution forces the angle distortion to be perfectly distributed on the one ring of these vertices. A local mesh refinement is applied to prevent these new constraints to conflict.

As illustrated on Fig.~\ref{fig:quads}, our method provides injective maps, even with the high distortion required to produce coarse quad meshes.

\begin{figure}
	\centering
	\includegraphics[width=.45\linewidth]{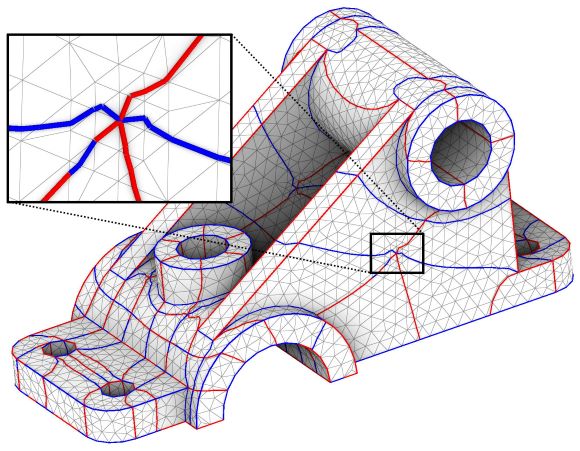}
	\includegraphics[width=.45\linewidth]{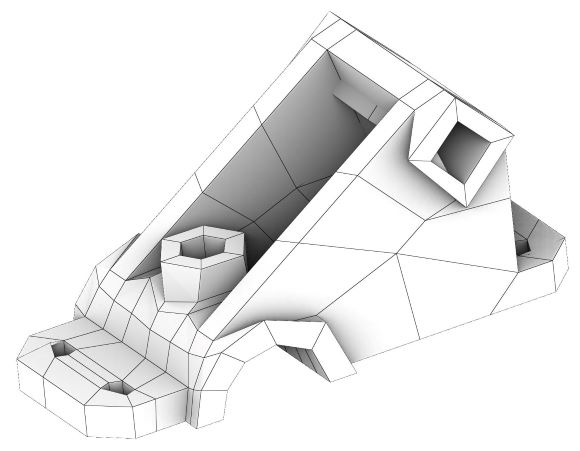}
	\caption{To avoid double coverings, we constrain singularities with negative index and untangle the map (left).
		This allows us to obtain valid parameterization even for coarsest quad layouts (right).}
	\label{fig:quads}
\end{figure}

\section{Limitations and future work}
Note that the ``finite number of steps'' theorem cannot be applied to the case when the diameter of the feasible set is zero.
This may happen, for example, when the parameter $t^*$ is set to the maximum attainable value.
To handle marginal cases of small size of the feasible set, we can assign to $\sigma_0$ parameter quite small value, say $10^{-3}$.
For a moment it is not clear whether we can attain deformation with best possible distortion numerically.

QIS algorithm is based on the idea of incremental contraction of the set of admissible mesh deformations, i.e. the deformations which provide finite value of QIS functional.
While rigorous analysis is missing, one can hope that with contraction the diameter of the admissible set tends to zero while its boundary, which is evidently infinitely smooth, is rounded and resembles the multidimensional sphere.
Evidently one should stop at some finite diameter, expecting that minimum in the contracted set converges to limiting point faster than the diameter coverges to zero.
Of course, much worser alternative is possible when measure of the feasible set tends to zero while its diameter is frozen near certain fixed value.
In order to get certain one should obtain estimates for positive definite part of the Hessian of QIS functional,
similar to that for untangling functional in \cite{garanzha2021foldoverfree}, and consider
limiting values, i.e. relation between diameter of the admissible set and curvature of its boundary.
This analysis is beyond the scope of present paper.

\section{Conclusion}

We formulate a set of variational problems potentially covering the complete technological chain for construction of optimal mappings and deformations with fixed as well as free boundaries.
We start with the continuation problem with respect to parameter $\varepsilon$,
this minimization allows us to compute optimal in the average deformations.
Then we formulate a continuation problem for worst quality measure maximization (quasi-isometric stiffening, QIS), which retains polyconvexity and smoothness of deformation.
At all stages we take care to demonstrate that finite number of basic optimization steps is enough in order to solve the problem.
We illustrate performance of our algorithm with challenging 2D and 3D numerical tests.
Importance of polyconvexity is underlined since some competing algorithms for mesh optimization which potentially may attain good quality criteria for deformations tend to lose deformation smoothness and symmetries even in the simplest test cases.

\appendix
\section{Relationship between distortion measures and the quasi-isometry constant}
\label{app:relation}

Given a deformation of a mesh, let us denote by $\Gamma$ its quasi-isometry constant (maximal relative length distortion of the map).
In our optimization procedure we do not optimize for $\Gamma$ directly, we minimize a mesh distortion measure $f$ instead.
Let us show the relation between the two.
More precisely, we want to estimate $\Gamma$ through available values of $f(J_k)$, where $J_k$ is the Jacobian matrix of $k$-the simplex of the mesh.

First of all, we can write down following estimation, where $\sigma_j(J_k)$ stand for (ordered) singular values of $J_k$:
\begin{equation}
	\label{eq:Gamma}
	\frac1{\Gamma} \leq \sigma_j(J_k) \leq \Gamma.
\end{equation}

In practice many different distortion measures can be used, but we are interested in those that satisfy $f(J_k) \geq 1$ and guarantee that inequality
\begin{equation}
	\label{eq:bounded-distortion}
	f(J_k) < \frac1t,\quad t < 1
\end{equation}
implies \eqref{eq:Gamma}.
Naturally, for QIS scheme to work with a certain distortion measure $f$, we need the bound for $\Gamma$ to tend to 1 as $t \to 1$.
In this section we analyze two possible choices of distortion measures, namely, Eq.~\eqref{eq:distortion} and Symmetric Dirichlet~\cite{Schuller2013}.

Let us start with $f := (1 - \theta) f_s + \theta f_v$.
First of all, let us note that $f_s \geq 1$ and $f_v \geq 1$,
so following inequalities hold:
\[
\frac{(\frac1d \tr J_i^\top J_i)^{d/2}}{\det J_i} < c_1, \quad c_1 :=  \left( \frac{1 - t  \theta}{t (1 - \theta)} \right)^{d/2}
\]
and
\[
\frac12 \left( \det J_i + \frac1{\det J_i} \right) < c_2, \quad c_2 :=   1 + \frac{1-t}{t \theta}.
\]

Reshetnyak's inequality \cite{reshetnyak1966space} implies that
\[
\frac{\sigma_1(J_i)}{\sigma_d(J_i)} < c_1 + \sqrt{c_1^2 -1}
\]

From the above estimates we obtain the required bounds for $\sigma_j(J_i)$~\cite{garanzha2000barrier}, so
\[
\Gamma < \left(c_1 + \sqrt{c_1^2 -1}\right)^{(d-1)/d} \left(c_2 + \sqrt{c_2^2 -1}\right)^{1/d}
\]
Note that in the 2D case and with $\theta = 1/2$ the bound for $\Gamma$ takes the simplest form
\[
\Gamma < \frac{(1 + \sqrt{1-t})^2}{t}.
\]
Indeed, Eq.~\eqref{eq:distortion} forces the quasi-isometry constant $\Gamma$ to tend to unity when $t\rightarrow 1$,
so our QIS techinque is correct.

Simple estimates for $\Gamma$ can be also derived for Symmetric Dirichlet (SD) distortion measure.
This measure can be written as follows:
\[
f(J) := \frac1{2d} \sum\limits_{j=1}^d \left(\sigma_j^2 + \frac1{\sigma_j^2}\right).
\]
In this case
\[
\Gamma < c_3 + \sqrt{c_3^2 -1}, \ \ c_3 = 1 + d \frac{1-t}{t}
\]
The bound satisfies the requirements for QIS algorithm, so SD can be used in our stiffening scheme.
Note, however, that even though SD distortion measure is a convex function of singular values, it is not clear whether it is polyconvex (convex with respect to minors of $J$).

\bibliographystyle{ACM-Reference-Format}
\bibliography{quality_mapping}


\begin{thebibliography}{37}


\ifx \showCODEN    \undefined \def \showCODEN     #1{\unskip}     \fi
\ifx \showDOI      \undefined \def \showDOI       #1{#1}\fi
\ifx \showISBNx    \undefined \def \showISBNx     #1{\unskip}     \fi
\ifx \showISBNxiii \undefined \def \showISBNxiii  #1{\unskip}     \fi
\ifx \showISSN     \undefined \def \showISSN      #1{\unskip}     \fi
\ifx \showLCCN     \undefined \def \showLCCN      #1{\unskip}     \fi
\ifx \shownote     \undefined \def \shownote      #1{#1}          \fi
\ifx \showarticletitle \undefined \def \showarticletitle #1{#1}   \fi
\ifx \showURL      \undefined \def \showURL       {\relax}        \fi
\providecommand\bibfield[2]{#2}
\providecommand\bibinfo[2]{#2}
\providecommand\natexlab[1]{#1}
\providecommand\showeprint[2][]{arXiv:#2}

\bibitem[\protect\citeauthoryear{Andersen and Andersen}{Andersen and
  Andersen}{2000}]%
        {Andersen2000}
\bibfield{author}{\bibinfo{person}{Erling~D. Andersen} {and}
  \bibinfo{person}{Knud~D. Andersen}.} \bibinfo{year}{2000}\natexlab{}.
\newblock \bibinfo{booktitle}{\emph{The Mosek Interior Point Optimizer for
  Linear Programming: An Implementation of the Homogeneous Algorithm}}.
\newblock \bibinfo{publisher}{Springer US}, \bibinfo{address}{Boston, MA},
  \bibinfo{pages}{197--232}.
\newblock
\showISBNx{978-1-4757-3216-0}
\urldef\tempurl%
\url{https://doi.org/10.1007/978-1-4757-3216-0_8}
\showDOI{\tempurl}


\bibitem[\protect\citeauthoryear{Ball}{Ball}{1976}]%
        {ball1976convexity}
\bibfield{author}{\bibinfo{person}{John~M Ball}.}
  \bibinfo{year}{1976}\natexlab{}.
\newblock \showarticletitle{Convexity conditions and existence theorems in
  nonlinear elasticity}.
\newblock \bibinfo{journal}{\emph{Archive for rational mechanics and Analysis}}
  \bibinfo{volume}{63}, \bibinfo{number}{4} (\bibinfo{year}{1976}),
  \bibinfo{pages}{337--403}.
\newblock


\bibitem[\protect\citeauthoryear{Bommes, Campen, Ebke, Alliez, and
  Kobbelt}{Bommes et~al\mbox{.}}{2013}]%
        {IGM2013}
\bibfield{author}{\bibinfo{person}{David Bommes}, \bibinfo{person}{Marcel
  Campen}, \bibinfo{person}{Hans-Christian Ebke}, \bibinfo{person}{Pierre
  Alliez}, {and} \bibinfo{person}{Leif Kobbelt}.}
  \bibinfo{year}{2013}\natexlab{}.
\newblock \showarticletitle{Integer-Grid Maps for Reliable Quad Meshing}.
\newblock \bibinfo{journal}{\emph{ACM Trans. Graph.}} \bibinfo{volume}{32},
  \bibinfo{number}{4}, Article \bibinfo{articleno}{98} (\bibinfo{date}{jul}
  \bibinfo{year}{2013}), \bibinfo{numpages}{12}~pages.
\newblock
\showISSN{0730-0301}
\urldef\tempurl%
\url{https://doi.org/10.1145/2461912.2462014}
\showDOI{\tempurl}


\bibitem[\protect\citeauthoryear{Bommes, Zimmer, and Kobbelt}{Bommes
  et~al\mbox{.}}{2009}]%
        {MIQ2009}
\bibfield{author}{\bibinfo{person}{David Bommes}, \bibinfo{person}{Henrik
  Zimmer}, {and} \bibinfo{person}{Leif Kobbelt}.}
  \bibinfo{year}{2009}\natexlab{}.
\newblock \showarticletitle{Mixed-Integer Quadrangulation}.
\newblock \bibinfo{journal}{\emph{ACM Trans. Graph.}} \bibinfo{volume}{28},
  \bibinfo{number}{3}, Article \bibinfo{articleno}{77} (\bibinfo{date}{jul}
  \bibinfo{year}{2009}), \bibinfo{numpages}{10}~pages.
\newblock
\showISSN{0730-0301}
\urldef\tempurl%
\url{https://doi.org/10.1145/1531326.1531383}
\showDOI{\tempurl}


\bibitem[\protect\citeauthoryear{Bommes, Zimmer, and Kobbelt}{Bommes
  et~al\mbox{.}}{2012}]%
        {Bommes:2010:PMO}
\bibfield{author}{\bibinfo{person}{David Bommes}, \bibinfo{person}{Henrik
  Zimmer}, {and} \bibinfo{person}{Leif Kobbelt}.}
  \bibinfo{year}{2012}\natexlab{}.
\newblock \showarticletitle{Practical mixed-integer optimization for geometry
  processing}. In \bibinfo{booktitle}{\emph{Proceedings of the 7th
  international conference on Curves and Surfaces}} (Avignon, France).
  \bibinfo{publisher}{Springer-Verlag}, \bibinfo{address}{Berlin, Heidelberg},
  \bibinfo{pages}{193--206}.
\newblock


\bibitem[\protect\citeauthoryear{Ciarlet and Geymonat}{Ciarlet and
  Geymonat}{1982}]%
        {Ciarlet}
\bibfield{author}{\bibinfo{person}{P.G. Ciarlet} {and} \bibinfo{person}{G.
  Geymonat}.} \bibinfo{year}{1982}\natexlab{}.
\newblock \showarticletitle{Sur les lois de comportement en elasticite
  non-lineaire compressible}.
\newblock \bibinfo{journal}{\emph{C.R. Acad.Sci. Paris Ser.II}}
  \bibinfo{volume}{295} (\bibinfo{year}{1982}), \bibinfo{pages}{423 -- 426}.
\newblock


\bibitem[\protect\citeauthoryear{Ciarlet and Necas}{Ciarlet and Necas}{1985}]%
        {ciarlet1985stiffening}
\bibfield{author}{\bibinfo{person}{Philippe Ciarlet} {and}
  \bibinfo{person}{Jindrich Necas}.} \bibinfo{year}{1985}\natexlab{}.
\newblock \showarticletitle{Unilateral problems in nonlinear, three-dimensional
  elasticity}.
\newblock \bibinfo{journal}{\emph{Arch. Rational Mech. Anal.}}
  \bibinfo{volume}{87} (\bibinfo{year}{1985}), \bibinfo{pages}{319--338}.
\newblock
\urldef\tempurl%
\url{https://doi.org/10.1007/BF00250917}
\showDOI{\tempurl}


\bibitem[\protect\citeauthoryear{Desobry, Protais, Ray, Corman, and
  Sokolov}{Desobry et~al\mbox{.}}{2021}]%
        {Desobry2021}
\bibfield{author}{\bibinfo{person}{David Desobry},
  \bibinfo{person}{Fran{\c{c}}ois Protais}, \bibinfo{person}{Nicolas Ray},
  \bibinfo{person}{Etienne Corman}, {and} \bibinfo{person}{Dmitry Sokolov}.}
  \bibinfo{year}{2021}\natexlab{}.
\newblock \showarticletitle{Frame Fields for CAD Models}. In
  \bibinfo{booktitle}{\emph{Advances in Visual Computing}},
  \bibfield{editor}{\bibinfo{person}{George Bebis}, \bibinfo{person}{Vassilis
  Athitsos}, \bibinfo{person}{Tong Yan}, \bibinfo{person}{Manfred Lau},
  \bibinfo{person}{Frederick Li}, \bibinfo{person}{Conglei Shi},
  \bibinfo{person}{Xiaoru Yuan}, \bibinfo{person}{Christos Mousas}, {and}
  \bibinfo{person}{Gerd Bruder}} (Eds.). \bibinfo{publisher}{Springer
  International Publishing}, \bibinfo{address}{Cham},
  \bibinfo{pages}{421--434}.
\newblock
\showISBNx{978-3-030-90436-4}


\bibitem[\protect\citeauthoryear{Du, Aigerman, Zhou, Kovalsky, Yan, Kaufman,
  and Ju}{Du et~al\mbox{.}}{2020}]%
        {Du2020}
\bibfield{author}{\bibinfo{person}{Xingyi Du}, \bibinfo{person}{Noam Aigerman},
  \bibinfo{person}{Qingnan Zhou}, \bibinfo{person}{Shahar~Z. Kovalsky},
  \bibinfo{person}{Yajie Yan}, \bibinfo{person}{Danny~M. Kaufman}, {and}
  \bibinfo{person}{Tao Ju}.} \bibinfo{year}{2020}\natexlab{}.
\newblock \showarticletitle{Lifting Simplices to Find Injectivity}.
\newblock \bibinfo{journal}{\emph{ACM Trans. Graph.}} \bibinfo{volume}{39},
  \bibinfo{number}{4}, Article \bibinfo{articleno}{120} (\bibinfo{date}{July}
  \bibinfo{year}{2020}), \bibinfo{numpages}{17}~pages.
\newblock
\showISSN{0730-0301}
\urldef\tempurl%
\url{https://doi.org/10.1145/3386569.3392484}
\showDOI{\tempurl}


\bibitem[\protect\citeauthoryear{Fang, Li, Jiang, and Kaufman}{Fang
  et~al\mbox{.}}{2021}]%
        {Fang2021IDP}
\bibfield{author}{\bibinfo{person}{Yu Fang}, \bibinfo{person}{Minchen Li},
  \bibinfo{person}{Chenfanfu Jiang}, {and} \bibinfo{person}{Danny~M. Kaufman}.}
  \bibinfo{year}{2021}\natexlab{}.
\newblock \showarticletitle{Guaranteed Globally Injective 3D Deformation
  Processing}.
\newblock \bibinfo{journal}{\emph{ACM Trans. Graph. (SIGGRAPH)}}
  \bibinfo{volume}{40}, \bibinfo{number}{4}, Article \bibinfo{articleno}{75}
  (\bibinfo{year}{2021}).
\newblock


\bibitem[\protect\citeauthoryear{Fu and Liu}{Fu and Liu}{2016}]%
        {Fu2016}
\bibfield{author}{\bibinfo{person}{Xiao-Ming Fu} {and} \bibinfo{person}{Yang
  Liu}.} \bibinfo{year}{2016}\natexlab{}.
\newblock \showarticletitle{Computing Inversion-Free Mappings by Simplex
  Assembly}.
\newblock \bibinfo{journal}{\emph{ACM Trans. Graph.}} \bibinfo{volume}{35},
  \bibinfo{number}{6}, Article \bibinfo{articleno}{216} (\bibinfo{date}{Nov.}
  \bibinfo{year}{2016}), \bibinfo{numpages}{12}~pages.
\newblock
\showISSN{0730-0301}
\urldef\tempurl%
\url{https://doi.org/10.1145/2980179.2980231}
\showDOI{\tempurl}


\bibitem[\protect\citeauthoryear{Garanzha}{Garanzha}{2000}]%
        {garanzha2000barrier}
\bibfield{author}{\bibinfo{person}{Vladimir Garanzha}.}
  \bibinfo{year}{2000}\natexlab{}.
\newblock \showarticletitle{The barrier method for constructing quasi-isometric
  grids}.
\newblock \bibinfo{journal}{\emph{Computational Mathematics and Mathematical
  Physics}}  \bibinfo{volume}{40} (\bibinfo{year}{2000}),
  \bibinfo{pages}{1617--1637}.
\newblock


\bibitem[\protect\citeauthoryear{Garanzha and Kaporin}{Garanzha and
  Kaporin}{1999}]%
        {Garanzha99}
\bibfield{author}{\bibinfo{person}{Vladimir Garanzha} {and}
  \bibinfo{person}{Igor Kaporin}.} \bibinfo{year}{1999}\natexlab{}.
\newblock \showarticletitle{Regularization of the barrier variational method of
  grid generation}.
\newblock \bibinfo{journal}{\emph{Comput. Math. Math. Phys.}}
  \bibinfo{volume}{39}, \bibinfo{number}{9} (\bibinfo{year}{1999}),
  \bibinfo{pages}{1426--1440}.
\newblock


\bibitem[\protect\citeauthoryear{Garanzha, Kaporin, Kudryavtseva, Protais, Ray,
  and Sokolov}{Garanzha et~al\mbox{.}}{2021a}]%
        {garanzha2021foldoverfree}
\bibfield{author}{\bibinfo{person}{Vladimir Garanzha}, \bibinfo{person}{Igor
  Kaporin}, \bibinfo{person}{Liudmila Kudryavtseva}, \bibinfo{person}{François
  Protais}, \bibinfo{person}{Nicolas Ray}, {and} \bibinfo{person}{Dmitry
  Sokolov}.} \bibinfo{year}{2021}\natexlab{a}.
\newblock \showarticletitle{Foldover-free maps in 50 lines of code}.
\newblock \bibinfo{journal}{\emph{ACM Transactions on Graphics}}
  \bibinfo{volume}{40}, \bibinfo{number}{4} (\bibinfo{year}{2021}).
\newblock
\urldef\tempurl%
\url{https://doi.org/10.1145/3450626.3459847}
\showDOI{\tempurl}


\bibitem[\protect\citeauthoryear{Garanzha, Kaporin, Kudryavtseva, Protais, Ray,
  and Sokolov}{Garanzha et~al\mbox{.}}{2021b}]%
        {garanzha2021b}
\bibfield{author}{\bibinfo{person}{Vladimir Garanzha}, \bibinfo{person}{Igor
  Kaporin}, \bibinfo{person}{Liudmila Kudryavtseva}, \bibinfo{person}{Fran{\c
  c}ois Protais}, \bibinfo{person}{Nicolas Ray}, {and} \bibinfo{person}{Dmitry
  Sokolov}.} \bibinfo{year}{2021}\natexlab{b}.
\newblock \showarticletitle{{On Local Invertibility and Quality of
  Free-boundary Deformations}}. In \bibinfo{booktitle}{\emph{{IMR 2021 - 29th
  International Meshing Roundtable}}}. \bibinfo{address}{Virtual, United
  States}.
\newblock
\urldef\tempurl%
\url{https://doi.org/10.5281/zenodo.5559040}
\showDOI{\tempurl}


\bibitem[\protect\citeauthoryear{Garanzha and Kudryavtseva}{Garanzha and
  Kudryavtseva}{2019}]%
        {10.1007/978-3-030-10934-9_35}
\bibfield{author}{\bibinfo{person}{Vladimir Garanzha} {and}
  \bibinfo{person}{Liudmila Kudryavtseva}.} \bibinfo{year}{2019}\natexlab{}.
\newblock \showarticletitle{Hypoelastic Stabilization of Variational Algorithm
  for Construction of Moving Deforming Meshes}. In
  \bibinfo{booktitle}{\emph{Optimization and Applications}},
  \bibfield{editor}{\bibinfo{person}{Yury Evtushenko},
  \bibinfo{person}{Milojica Ja{\'{c}}imovi{\'{c}}}, \bibinfo{person}{Michael
  Khachay}, \bibinfo{person}{Yury Kochetov}, \bibinfo{person}{Vlasta Malkova},
  {and} \bibinfo{person}{Mikhail Posypkin}} (Eds.).
  \bibinfo{publisher}{Springer International Publishing},
  \bibinfo{address}{Cham}, \bibinfo{pages}{497--511}.
\newblock
\showISBNx{978-3-030-10934-9}


\bibitem[\protect\citeauthoryear{Garanzha, Kudryavtseva, and
  Utyuzhnikov}{Garanzha et~al\mbox{.}}{2014}]%
        {Garanzha2014}
\bibfield{author}{\bibinfo{person}{Vladimir Garanzha},
  \bibinfo{person}{Liudmila Kudryavtseva}, {and} \bibinfo{person}{Sergei
  Utyuzhnikov}.} \bibinfo{year}{2014}\natexlab{}.
\newblock \showarticletitle{Variational method for untangling and optimization
  of spatial meshes}.
\newblock \bibinfo{journal}{\emph{J. Comput. Appl. Math.}}
  \bibinfo{volume}{269} (\bibinfo{year}{2014}), \bibinfo{pages}{24 -- 41}.
\newblock
\showISSN{0377-0427}
\urldef\tempurl%
\url{https://doi.org/10.1016/j.cam.2014.03.006}
\showDOI{\tempurl}


\bibitem[\protect\citeauthoryear{Gillespie, Springborn, and Crane}{Gillespie
  et~al\mbox{.}}{2021}]%
        {CEPS}
\bibfield{author}{\bibinfo{person}{Mark Gillespie}, \bibinfo{person}{Boris
  Springborn}, {and} \bibinfo{person}{Keenan Crane}.}
  \bibinfo{year}{2021}\natexlab{}.
\newblock \showarticletitle{Discrete Conformal Equivalence of Polyhedral
  Surfaces}.
\newblock \bibinfo{journal}{\emph{ACM Trans. Graph.}} \bibinfo{volume}{40},
  \bibinfo{number}{4} (\bibinfo{year}{2021}).
\newblock


\bibitem[\protect\citeauthoryear{Hestenes and Stiefel}{Hestenes and
  Stiefel}{1952}]%
        {Hestenes1952MethodsOC}
\bibfield{author}{\bibinfo{person}{Magnus~R. Hestenes} {and}
  \bibinfo{person}{Eduard Stiefel}.} \bibinfo{year}{1952}\natexlab{}.
\newblock \showarticletitle{Methods of conjugate gradients for solving linear
  systems}.
\newblock \bibinfo{journal}{\emph{Journal of research of the National Bureau of
  Standards}}  \bibinfo{volume}{49} (\bibinfo{year}{1952}),
  \bibinfo{pages}{409--435}.
\newblock


\bibitem[\protect\citeauthoryear{Hormann and Greiner}{Hormann and
  Greiner}{2000}]%
        {Hormann2000MIPS}
\bibfield{author}{\bibinfo{person}{K. Hormann} {and} \bibinfo{person}{G.
  Greiner}.} \bibinfo{year}{2000}\natexlab{}.
\newblock \showarticletitle{MIPS: An Efficient Global Parametrization Method}.
  In \bibinfo{booktitle}{\emph{Curve and Surface Design}}.
  \bibinfo{publisher}{Vanderbilt University press}.
\newblock


\bibitem[\protect\citeauthoryear{Ivanenko}{Ivanenko}{2000}]%
        {Ivanenko2000}
\bibfield{author}{\bibinfo{person}{S.A. Ivanenko}.}
  \bibinfo{year}{2000}\natexlab{}.
\newblock \showarticletitle{Control of cell shapes in the course of grid
  generation}.
\newblock \bibinfo{journal}{\emph{Zh. Vychisl. Mat. Mat. Fiz.}}
  \bibinfo{volume}{40} (\bibinfo{date}{01} \bibinfo{year}{2000}),
  \bibinfo{pages}{1662--1684}.
\newblock


\bibitem[\protect\citeauthoryear{Jacquotte}{Jacquotte}{1988}]%
        {jacquotte1988mechanical}
\bibfield{author}{\bibinfo{person}{Olivier-P Jacquotte}.}
  \bibinfo{year}{1988}\natexlab{}.
\newblock \showarticletitle{A mechanical model for a new grid generation method
  in computational fluid dynamics}.
\newblock \bibinfo{journal}{\emph{Computer methods in applied mechanics and
  engineering}} \bibinfo{volume}{66}, \bibinfo{number}{3}
  (\bibinfo{year}{1988}), \bibinfo{pages}{323--338}.
\newblock


\bibitem[\protect\citeauthoryear{Kovalsky, Aigerman, Basri, and
  Lipman}{Kovalsky et~al\mbox{.}}{2015}]%
        {LargeScaleBD:2015}
\bibfield{author}{\bibinfo{person}{Shahar~Z. Kovalsky}, \bibinfo{person}{Noam
  Aigerman}, \bibinfo{person}{Ronen Basri}, {and} \bibinfo{person}{Yaron
  Lipman}.} \bibinfo{year}{2015}\natexlab{}.
\newblock \showarticletitle{Large-scale bounded distortion mappings}.
\newblock \bibinfo{journal}{\emph{ACM Transactions on Graphics (proceedings of
  ACM SIGGRAPH Asia)}} \bibinfo{volume}{34}, \bibinfo{number}{6}
  (\bibinfo{year}{2015}).
\newblock


\bibitem[\protect\citeauthoryear{Lipman}{Lipman}{2012}]%
        {Lipman2012}
\bibfield{author}{\bibinfo{person}{Yaron Lipman}.}
  \bibinfo{year}{2012}\natexlab{}.
\newblock \showarticletitle{Bounded Distortion Mapping Spaces for Triangular
  Meshes}.
\newblock \bibinfo{journal}{\emph{ACM Trans. Graph.}} \bibinfo{volume}{31},
  \bibinfo{number}{4}, Article \bibinfo{articleno}{108} (\bibinfo{date}{July}
  \bibinfo{year}{2012}), \bibinfo{numpages}{13}~pages.
\newblock
\showISSN{0730-0301}
\urldef\tempurl%
\url{https://doi.org/10.1145/2185520.2185604}
\showDOI{\tempurl}


\bibitem[\protect\citeauthoryear{Lévy, Petitjean, Ray, and Maillo~t}{Lévy
  et~al\mbox{.}}{2002}]%
        {levy2002}
\bibfield{author}{\bibinfo{person}{Bruno Lévy}, \bibinfo{person}{Sylvain
  Petitjean}, \bibinfo{person}{Nicolas Ray}, {and} \bibinfo{person}{Jérome
  Maillo~t}.} \bibinfo{year}{2002}\natexlab{}.
\newblock \showarticletitle{Least Squares Conformal Maps for Automatic Texture
  Atlas Generation}. In \bibinfo{booktitle}{\emph{ACM SIGGRAPH conference
  proceedings}}, \bibfield{editor}{\bibinfo{person}{ACM}} (Ed.).
\newblock


\bibitem[\protect\citeauthoryear{Naitsat, Zhu, and Zeevi}{Naitsat
  et~al\mbox{.}}{2020}]%
        {Naitsat2019}
\bibfield{author}{\bibinfo{person}{Alexander Naitsat}, \bibinfo{person}{Yufeng
  Zhu}, {and} \bibinfo{person}{Yehoshua~Y Zeevi}.}
  \bibinfo{year}{2020}\natexlab{}.
\newblock \showarticletitle{Adaptive Block Coordinate Descent for Distortion
  Optimization}. In \bibinfo{booktitle}{\emph{Computer Graphics Forum}},
  Vol.~\bibinfo{volume}{39}. Wiley Online Library, \bibinfo{pages}{360--376}.
\newblock


\bibitem[\protect\citeauthoryear{Prager}{Prager}{1957}]%
        {prager1957stiffening}
\bibfield{author}{\bibinfo{person}{William Prager}.}
  \bibinfo{year}{1957}\natexlab{}.
\newblock \showarticletitle{On ideal locking materials}.
\newblock \bibinfo{journal}{\emph{Transactions of the Society of Rheology}}
  \bibinfo{volume}{1}, \bibinfo{number}{1} (\bibinfo{year}{1957}),
  \bibinfo{pages}{169--175}.
\newblock


\bibitem[\protect\citeauthoryear{Ray, Vallet, Li, and L\'{e}vy}{Ray
  et~al\mbox{.}}{2008}]%
        {NSDF2008}
\bibfield{author}{\bibinfo{person}{Nicolas Ray}, \bibinfo{person}{Bruno
  Vallet}, \bibinfo{person}{Wan~Chiu Li}, {and} \bibinfo{person}{Bruno
  L\'{e}vy}.} \bibinfo{year}{2008}\natexlab{}.
\newblock \showarticletitle{N-Symmetry Direction Field Design}.
\newblock \bibinfo{journal}{\emph{ACM Trans. Graph.}} \bibinfo{volume}{27},
  \bibinfo{number}{2}, Article \bibinfo{articleno}{10} (\bibinfo{date}{may}
  \bibinfo{year}{2008}), \bibinfo{numpages}{13}~pages.
\newblock
\showISSN{0730-0301}
\urldef\tempurl%
\url{https://doi.org/10.1145/1356682.1356683}
\showDOI{\tempurl}


\bibitem[\protect\citeauthoryear{Reshetnyak}{Reshetnyak}{1966}]%
        {reshetnyak1966space}
\bibfield{author}{\bibinfo{person}{Yu~G Reshetnyak}.}
  \bibinfo{year}{1966}\natexlab{}.
\newblock \showarticletitle{Bounds on moduli of continuity for certain
  mappings}.
\newblock \bibinfo{journal}{\emph{Siberian Mathematical Journal}}
  \bibinfo{volume}{7}, \bibinfo{number}{5} (\bibinfo{year}{1966}),
  \bibinfo{pages}{879--886}.
\newblock


\bibitem[\protect\citeauthoryear{Rumpf}{Rumpf}{1996}]%
        {rumpf1996variational}
\bibfield{author}{\bibinfo{person}{Martin Rumpf}.}
  \bibinfo{year}{1996}\natexlab{}.
\newblock \showarticletitle{A variational approach to optimal meshes}.
\newblock \bibinfo{journal}{\emph{Numer. Math.}} \bibinfo{volume}{72},
  \bibinfo{number}{4} (\bibinfo{year}{1996}), \bibinfo{pages}{523--540}.
\newblock


\bibitem[\protect\citeauthoryear{S.~K.~Godunov}{S.~K.~Godunov}{1995}]%
        {Godunov1995}
\bibfield{author}{\bibinfo{person}{G.~A.~Chumakov S.~K.~Godunov, V.
  M.~Gordienko}.} \bibinfo{year}{1995}\natexlab{}.
\newblock \showarticletitle{Variational principle for 2-D regular
  quasi-isometric grid generation}.
\newblock \bibinfo{journal}{\emph{International Journal of Computational Fluid
  Dynamics}} \bibinfo{volume}{5}, \bibinfo{number}{1-2} (\bibinfo{year}{1995}),
  \bibinfo{pages}{99--118}.
\newblock


\bibitem[\protect\citeauthoryear{Sawhney and Crane}{Sawhney and Crane}{2017}]%
        {BFF}
\bibfield{author}{\bibinfo{person}{Rohan Sawhney} {and} \bibinfo{person}{Keenan
  Crane}.} \bibinfo{year}{2017}\natexlab{}.
\newblock \showarticletitle{Boundary First Flattening}.
\newblock \bibinfo{journal}{\emph{ACM Trans. Graph.}} \bibinfo{volume}{37},
  \bibinfo{number}{1}, Article \bibinfo{articleno}{5} (\bibinfo{date}{dec}
  \bibinfo{year}{2017}), \bibinfo{numpages}{14}~pages.
\newblock
\showISSN{0730-0301}
\urldef\tempurl%
\url{https://doi.org/10.1145/3132705}
\showDOI{\tempurl}


\bibitem[\protect\citeauthoryear{Sch{\"u}ller, Kavan, Panozzo, and
  Sorkine-Hornung}{Sch{\"u}ller et~al\mbox{.}}{2013}]%
        {Schuller2013}
\bibfield{author}{\bibinfo{person}{Christian Sch{\"u}ller},
  \bibinfo{person}{Ladislav Kavan}, \bibinfo{person}{Daniele Panozzo}, {and}
  \bibinfo{person}{Olga Sorkine-Hornung}.} \bibinfo{year}{2013}\natexlab{}.
\newblock \showarticletitle{Locally Injective Mappings}.
\newblock \bibinfo{journal}{\emph{Computer Graphics Forum (proceedings of
  Symposium on Geometry Processing)}} \bibinfo{volume}{32}, \bibinfo{number}{5}
  (\bibinfo{year}{2013}).
\newblock


\bibitem[\protect\citeauthoryear{Sokolov}{Sokolov}{2022}]%
        {supplemental}
\bibfield{author}{\bibinfo{person}{Dmitry Sokolov}.}
  \bibinfo{year}{2022}\natexlab{}.
\newblock \bibinfo{title}{Supplemental material for ``Practical lowest
  distortion mapping''}.
\newblock \bibinfo{howpublished}{\url{https://github.com/ssloy/QIS}}.
\newblock
\newblock
\shownote{Accessed: 2022-01-27.}


\bibitem[\protect\citeauthoryear{Sorkine, Cohen-Or, Goldenthal, and
  Lischinski}{Sorkine et~al\mbox{.}}{2002}]%
        {Sorkine2002}
\bibfield{author}{\bibinfo{person}{O. Sorkine}, \bibinfo{person}{D. Cohen-Or},
  \bibinfo{person}{R. Goldenthal}, {and} \bibinfo{person}{D. Lischinski}.}
  \bibinfo{year}{2002}\natexlab{}.
\newblock \showarticletitle{Bounded-distortion piecewise mesh
  parameterization}. In \bibinfo{booktitle}{\emph{IEEE Visualization, 2002. VIS
  2002.}} \bibinfo{pages}{355--362}.
\newblock
\urldef\tempurl%
\url{https://doi.org/10.1109/VISUAL.2002.1183795}
\showDOI{\tempurl}


\bibitem[\protect\citeauthoryear{Springborn, Schr\"{o}der, and
  Pinkall}{Springborn et~al\mbox{.}}{2008}]%
        {CETM}
\bibfield{author}{\bibinfo{person}{Boris Springborn}, \bibinfo{person}{Peter
  Schr\"{o}der}, {and} \bibinfo{person}{Ulrich Pinkall}.}
  \bibinfo{year}{2008}\natexlab{}.
\newblock \showarticletitle{Conformal Equivalence of Triangle Meshes}.
\newblock \bibinfo{journal}{\emph{ACM Trans. Graph.}} \bibinfo{volume}{27},
  \bibinfo{number}{3} (\bibinfo{date}{aug} \bibinfo{year}{2008}),
  \bibinfo{pages}{1–11}.
\newblock
\showISSN{0730-0301}
\urldef\tempurl%
\url{https://doi.org/10.1145/1360612.1360676}
\showDOI{\tempurl}


\bibitem[\protect\citeauthoryear{Zhu, Byrd, and Nocedal}{Zhu
  et~al\mbox{.}}{1997}]%
        {LBFGS}
\bibfield{author}{\bibinfo{person}{C. Zhu}, \bibinfo{person}{R.~H. Byrd}, {and}
  \bibinfo{person}{J. Nocedal}.} \bibinfo{year}{1997}\natexlab{}.
\newblock \showarticletitle{{L-BFGS-B}: {Algorithm} 778: {L-BFGS-B}, {FORTRAN}
  routines for large scale bound constrained optimization}.
\newblock \bibinfo{journal}{\emph{ACM Trans. Math. Software}}
  \bibinfo{volume}{23}, \bibinfo{number}{4} (\bibinfo{year}{1997}),
  \bibinfo{pages}{550--560}.
\newblock


\end{thebibliography}
\end{document}